\documentclass[%
 reprint,
 amsmath,amssymb,
 aps,
prx,
floatfix,
superscriptaddress]{revtex4-2}

\usepackage{graphicx}
\usepackage{dcolumn}
\usepackage{bm}
\usepackage{comment}
\usepackage{bbm}
\usepackage{calrsfs}
\usepackage{eucal}
\usepackage{dsfont}
\usepackage{float}
\usepackage{color}
\usepackage[colorlinks=true,linkcolor=blue,citecolor=blue,urlcolor=blue]{hyperref}
\usepackage{amsthm}
\newtheorem{theorem}{Theorem}[section]

\newcommand{\<}{\langle}
\renewcommand{\>}{\rangle}

\providecommand{\tr}{{\rm tr}}
\renewcommand{\phi}{\varphi}
\newcommand{\trho}[1]{{\rm tr}[#1 \rho]}

\begin{document}

\preprint{APS/123-QED}

\title{Emergent complex quantum networks in 
continuous-variables non-Gaussian states}

\author{Mattia Walschaers}
\email{mattia.walschaers@lkb.upmc.fr}
\affiliation{Laboratoire Kastler Brossel, Sorbonne Universit\'{e}, CNRS, ENS-Universit\'{e} PSL,  Coll\`{e}ge de France, 4 place Jussieu, F-75252 Paris, France}

\author{Nicolas Treps}
\affiliation{Laboratoire Kastler Brossel, Sorbonne Universit\'{e}, CNRS, ENS-Universit\'{e} PSL, Coll\`{e}ge de France, 4 place Jussieu, F-75252 Paris, France}

\author{Bhuvanesh Sundar}
\affiliation{Institute for Quantum Optics and Quantum Information of the Austrian Academy of Sciences, Innsbruck A 6020, Austria}
\author{Lincoln D. Carr}
\affiliation{Department of Physics, Colorado School of Mines, Golden, Colorado 80401, USA}

\author{Valentina Parigi}
\email{valentina.parigi@lkb.upmc.fr}
\affiliation{Laboratoire Kastler Brossel, Sorbonne Universit\'{e}, CNRS, ENS-Universit\'{e} PSL, Coll\`{e}ge de France, 4 place Jussieu, F-75252 Paris, France}

\date{\today}
\begin{abstract} 
We use complex network theory to study a class of continuous variable quantum states that present both multipartite entanglement and non-Gaussian statistics. We consider the intermediate scale of several dozens of components at which such systems are already hard to characterize.  In particular, the states are built from an initial \emph{imprinted} cluster state created via Gaussian entangling operations according to a complex network structure. We then engender non-Gaussian statistics via multiple photon subtraction operations acting on a single node. We replicate in the quantum regime some of the models that mimic real-world complex networks in order to test their structural properties under local operations. We then go beyond the already known single-mode effects, by studying  the \emph{emergent} network of photon-number correlations via complex networks measures. We analytically prove that the imprinted network structure defines a vicinity of nodes, at a distance of four steps from the photon-subtracted node, in which the emergent network changes due to photon subtraction. Moreover, our numerical analysis shows that the emergent structure is greatly influenced by the structure of the imprinted network. Indeed, while the mean and the variance of the degree and clustering distribution of the emergent network always increase, the higher moments of the distributions are governed by the specific structure of the imprinted network. Finally, we show that the behaviour of nearest neighbours of the subtraction node depends on how they are connected to each other in the imprinted structure.
\end{abstract} 

\maketitle

\section{\label{sec:intro}Introduction}

Large multiparty quantum systems are extremely hard to describe, although the complex behavior of their quantum states is what makes them appealing resources for quantum information processing. Intensive efforts have been dedicated to the direct representation of  quantum states via numerical and analytical approaches with the aim of classifying and detecting truly non-classical and useful quantum features, like entanglement\cite{Islam15,Laflorencie16}.
For example, tensor networks have been demonstrated to be very powerful tools for entanglement classification, and have  been applied in many fields ranging from condensed matter physics to conformal field theory~\cite{Orus19}. 
In the quantum information scenario, resource theories 
offer a general theoretical framework to classify useful resources to reach desired quantum features~\cite{Chitambar19}.
Different strategies that are becoming more and more popular exploit machine learning procedures~\cite{Carleo19,Dunjko18} to classify general complex features in the quantum realm
~\cite{Deng17,Carleo17,Carrasquilla20} often by representing them via neural networks~\cite{guo2019backpropagation}.  

\begin{figure*}[t]
\centering
\includegraphics[width=0.9\textwidth]{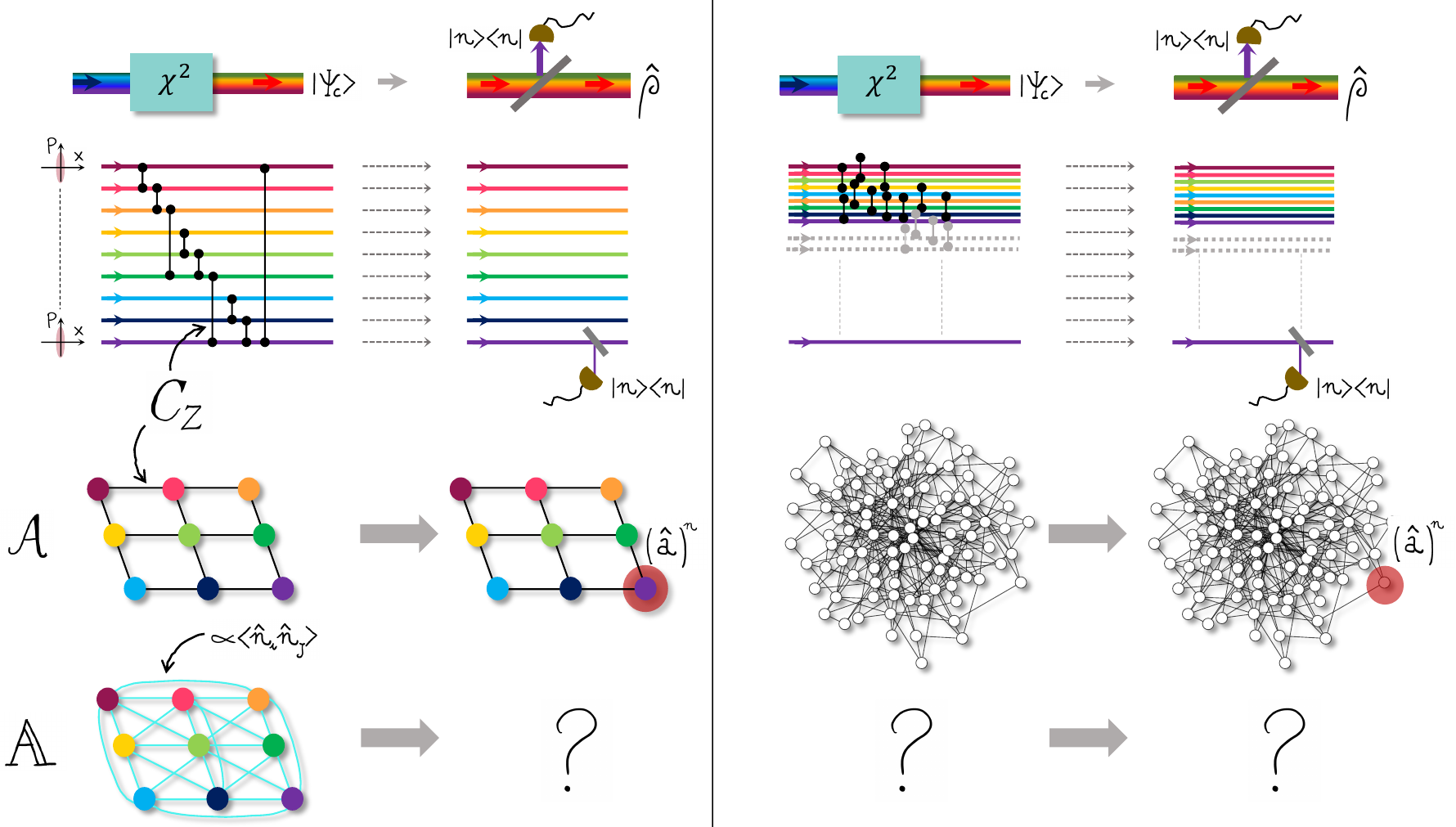}
\caption{\label{sketch} Left column, first row: sketch of nonlinear multi-mode  $\chi^{(2)}$ optical processes that demonstrate the deterministic implementation of large CV networks and mode-dependent multi-photon subtraction. Second row: circuit representation of the two processes. Third row: graphical representation of the so-called cluster states (imprinted networks) where the nodes are the different modes of the field and edges correspond to $C_Z$ gates between modes. Last and fourth row: emergent networks where edges are photon-number correlations between modes. Right column: same as  left but considering imprinted cluster with shapes typical of complex network models.}
\end{figure*}

In this work we focus on the complex behaviour of quantum states in continuous variable (CV) quantum systems and we tackle it via complex network theory.  Our work is motivated by all-optical platforms, based on continuous quantum observables, that can already produce large entangled networks~\cite{Asavanant19,Larsen19,Chen14}. These networks are made of traveling light fields with quantum correlations  between amplitude and phase values of different modes of the field, e.g. light at different colors.  They have Gaussian measurement statistics for amplitude and phase continuous variables, so that they can be easily simulated via classical computer. They are essential resources for measurement-based quantum computing but, in order to perform quantum protocols, they must acquire non-Gaussian statistics of the continuous variables. Non-Gaussian statistics can be induced   via mode-selective addition and subtraction of photons ~\cite{Ra19,Biagi20}, that are then called non-Gaussian operations. When the number of entangled systems --  in our case optical modes -- and the number of non-Gaussian operations grow, these systems quickly become hard to benchmark~\cite{WalschaersPRA17,Zhuang18,Albarelli18,Takagi18,cimini2020neural,walschaers2020conditional}. Standard quantum optical methods such as homodyne tomography rapidly become intractable, and the properties of the resulting quantum states are hard to unravel. 

In the most general scenario, even measurement outcomes of such systems become computationally hard to simulate~\cite{Chabaud17}, so that the associated sampling problem is one of the many variations of bosons sampling \cite{PhysRevLett.119.170501,10.1145/1993636.1993682}. In such sampling problems, photon-number correlations have shown to be efficient tools to extract  properties of such intricate systems \cite{Walschaers_2016,Giordani,PhysRevA.99.023836,Zhong1460,zhong2021phaseprogrammable,Xanadu-Advantage}. 
Average values of photon-number correlations can be analytically tractable, but they can only  unveil global properties of the system. Yet, to acquire a more detailed image of the state, we can use network theory to analyse distributions of photon-number correlations.

In this work, the use of network theory will be two-fold. First, the optical CV entangled networks can in fact be easily reconfigured in arbitrary shape~\cite{Cai17}. We can thus consider entangled networks which are generated through models studied in network theory to reproduce the features of real-world networks. They provide an excellent playground to explore whether mimicking real-world complex network structures~\cite{Nokkala18a,Sansavini20} provides an advantage for quantum information technologies, including quantum simulation and communication~\cite{Cai17, Arzani19} in a future quantum internet. Second, we will use network theory as a powerful tool for benchmarking entangled networks when affected by local non-Gaussian operations by building weighted networks of photon-number correlations. 
We can then  study how  different shapes of the initial entangled network can support, enhance, spread or destroy the non-Gaussian features provided by local subtraction of multiple photons.

\subsection{\label{subsec:sketch} Conceptual scheme}
The conceptual scheme of our analysis is shown in Fig.~\ref{sketch}.
Large CV entangled networks have been deterministically implemented via non-linear $\chi^{(2)}$ optical processes.
This operation entangles different modes (be they spatial, spectral or temporal) of the fields via an appropriately engineered parametric interaction~\cite{Roslund14,Cai17,Chen14,Yokoyama13,Asavanant19,Larsen19}. These processes are sketched in the left corner of the upper row of Fig.~\ref{sketch}, where the different modes are represented by different colors.
 The circuit representation of the generated states is sketched in the second row of the figure: the non-linear optical process is equivalent to different travelling optical modes occupied by squeezed vacuum states that are entangled via $C_Z$ gates.  This generates the {\em cluster state} \cite{Menicucci06,Gu09}. In the third row we show the graphical representation of the state: the different optical modes are represented by different nodes of the network which are linked by $C_Z$ gate entangling operations, counted by the entry 1 in the adjacency matrix $\mathcal{A}$ of the network. This structure is an {\em imprinted} network, as it builds the initial quantum states. In the right side of the first column of Fig.~\ref{sketch} we picture the action of multiple-photon subtraction, i.e. the repeated application of the photon annihilation operator  $\hat{a}$ on one specific mode of the field.
The probabilistic implementation of this operation consists of a  mode-selective beam-splitter that sends a small fraction of light to a photon counter: when $n$ photons are detected  an $n$-photon-subtracted state is heralded. The process can be implemented via non-linear interaction with supplementary gate fields~\cite{Ra19}. 
In the fourth row of the first column we show the network of photon-number correlations between the different field modes that emerge from the imprinted network. This is the {\em emergent} network. Its adjacency matrix 
$\mathbb{A}$ contains continuous values between $0$ and $1$, indicating the strength of photon-number correlations between couple of nodes.
In this work we are interested in following the changes of this emergent network of photon number correlations after photon-subtractions on one node as a benchmark of the desired non-Gaussian properties of CV quantum states.

As described above complex networks have a dual role in this work.
First, we consider complex network structures for the {\em imprinted} network of entangling $C_Z$ gates, as shown in the third row of the second column of Fig.~\ref{sketch}. We will use complex network models that reproduce particular features of real-world networks. Second,  we probe the impact of photon subtraction on the emergent networks by analyzing typical network measures, like degree and clustering, that quantify the number of links and their structure and that are commonly used to characterize complex networks.

\subsection{\label{subsec:summary} Summary of the results}

The entangling $C_Z$-gates in the imprinted network generate short range correlations between nodes.  Also, the non-Gaussian operations we consider here - photon-subtractions - are applied locally on a single node. Under such conditions  the effect of photon subtraction on a regular graph is limited \cite{Walschaers18}. On the contrary, here we probe imprinted networks constructed from complex network models, where typical distances between nodes are short. Our results are as follows:
\begin{itemize}
\item The effect of photon-subtraction leads to local changes in the emergent network of photon-number correlations. We analytically prove that the imprinted network structure defines a vicinity of nodes around the photon-subtracted node in which the emergent network changes due to photon subtraction.
\item For imprinted complex networks we see highly connected emergent networks of photon-number correlations. In such systems we show that photon-subtraction changes the structure of the correlation networks in a profound way in the vicinity of the photon-subtracted node.
\item We compare different complex network models, like  Barab\'{a}si-Albert  or Watts-Strogatz, and their emergent correlation structure before and after subtraction of ten photons. The analysis of the distribution of degrees and clustering coefficients of the emergent correlation networks reveals that the amount of randomness of links or the inhomogeneity  of number of links in different network models affects the emergent correlation networks, before and after the photon subtraction. We see that some networks are more efficient than others in spreading the non-Gaussian effect.  
\item For networks in which a majority of nodes is affected by photon subtraction, we show that the global clustering and degree distributions are altered in the same way: the mean and variance increase. This is in contrast with the structure in the tails of the distributions, characterized by higher moments, where the effect of photon subtraction is governed by the specific structure of the imprinted network.
\item We unveil  
the importance of the local network structure in the vicinity of the node of photon subtraction. It is the connectivity of the imprinted sub-network spanned by this node and its neighbors that drives the changes due to photon subtraction.
\end{itemize}

\subsection{\label{subsec:structure} Structure of the Paper}
The Article is outlined as follows. Section~\ref{sec:CVNetworks} introduces the basic concepts of complex networks and CV cluster states used in this work to make the Article accessible to readers with different backgrounds -- readers with one or both areas of expertise may choose to skip this section or particular subsections. We describe the imprinted network structure and  review  non-Gaussian operations and their importance for getting non-Gaussian CV cluster states. We introduce the emergent correlation networks and complex network measures. Then in Sec.~\ref{sec:GaussianNetworks} we look at emergent correlation networks for Gaussian cluster states when different complex network models are used for the imprinted network. This section forms a baseline for the ensuing non-Gaussian analysis. In Sec.~\ref{sec:NonGaussianNetworks} we describe the evolution of the emergent correlation networks when repeated photon subtractions are applied. We show that photon subtraction in a single node only affects a certain vicinity of the subtraction node. We then analyse the global impact of the non-Gaussian operation on the emergent network. In Sec.~\ref{sec:AllDistance} we analyse the local effects of photon subtraction. We show  that the sub-networks formed by nodes at different distances from the subtracting point have a different influence on the statistics of the non-Gaussian graphs. We then reveal the driving mechanism for the sub-networks composed by all the node at distance one from the subtracting point , i.e. all the nodes that have a direct link with the subtracting node in the imprinted network.  Finally,  in Sec.~\ref{sec:Discussion}  we comment on general features of non-Gaussian correlations in photon-subtracted networks and specific features dependent on the imprinted network model. 

\section{Quantum Complex Network theory and Continuous Variable Quantum Systems}\label{sec:CVNetworks}
\subsection{\label{subsec:CN} Complex networks and quantum physics}
In the last decades, network theory has made significant progress in describing collective features and functionality of complex systems \cite{Newman18,Barabasi16}.
Network-based descriptions are pivotal in social and biological science as well as  in technological infrastructures like power grids and information networks such as the internet. The study of complex network structures has spread in physics~\cite{AlbertRMP02,DorogovtsevRMP08} helping in the description of complex physical systems.  Subfields in physics utilizing complex networks include statistical physics, condensed matter and quantum physics with, e.g., the study of the Ising model and Bose Einstein condensation~\cite{Bianconi15,Halu13,Jahnke08,Burioni01,Mulken11,Valdez17,Nokkala18}.  More recently the study of complex networks has become relevant for quantum systems and procedures employed in quantum information technologies~\cite{Biamonte19,Sansavini20,Cabot18,Nokkala18a,Chakraborty16,Cuquet09,Faccin14} indicating that a dedicated theory of quantum complex networks needs to be built, especially for networks with no classical equivalent, like those based on quantum correlations or quantum mutual information. In particular, emergent complex networks based on quantum mutual information have determined critical points for quantum phase transitions~\cite{Valdez17,sundar2018complex,Buca19}. 
Likewise, complex network theory has been successful in determining self-similarity in entanglement structure of spin-chains~\cite{Sokolov2020}  as well as new kinds of structured entanglement emerging from quantum cellular automata~\cite{hillberry2020entangled} on qubit/gate/ciruit-based quantum computers. 
Networks are naturally evoked in the quantum regime in relation to the quantum internet~\cite{Kimble08}, where it is not clear yet if the best arrangement of its components will take a complex shape like the classical internet.
Networks are however pivotal in all quantum technologies.
Indeed, quantum information algorithms and quantum transport can be mapped to quantum walks on regular and complex networks~\cite{PhysRevLett.102.180501,Mohseni2008,Plenio_2008,annurev2016}.  Complex networks have also a crucial role in near-term quantum information processing because they describe networked noisy intermediate-scale quantum computers~\cite{awschalom2019development,altman2019quantum}. Thus complex network theory provides a versatile toolbox, as it can be applied to different quantum features, and is very efficient in revealing emerging collective structural mechanisms.

Here we apply, for the first time, complex network analysis to CV multipartite quantum states. We focus on the ones that can be generated in the more advanced optical platforms, but the method can be applied to general CV states.

Networks are a collection of nodes and links. This is a very versatile conceptual structure that can be applied to any kind of relation between a collection of physical systems: receivers and senders in an information network linked by physical channels; atomic spins interacting via magnetic forces; or physical observables linked by correlation relations. Modern complex network theory in the physical, social, and life sciences is built on graph theory from mathematics and computer science.
 The adjacency matrix is the central mathematical object of complex network theory: a non-zero term in the matrix indicates a link between two nodes, where the indices of the matrix determine the nodes. The work done in this Article is based on constructing the adjacency matrices for relevant networks, and subsequently extracting relevant properties from them.\\

In quantum information, graphs define the structure of the  so-called \emph{graph} or \emph{cluster states}~\footnote{Cluster is sometimes reserved for  graphs  allowing  for  universal  quantum computing.  In  this  work, however,  we  use  the  terms  \emph{cluster state}  and  \emph{graph  state}  as synonyms.}.  They correspond to multipartite quantum states with a specific entanglement structure introduced in the context of measurement-based quantum computing~\cite{Raussendorf01,Raussendorf03,Menicucci06,Gu09}. For such cluster states, a non-zero term in the adjacency matrix indicates that an entangling gate has been applied between two qubits or between two quantum fields in two different optical modes (where the information is encoded in discrete or continuous variables, respectively). In this section we first provide a brief introduction to CV quantum optics in \ref{sec:QuantumOptics} and we review the CV cluster states that are induced by the imprinted networks in Sec.~\ref{sec:sub-clust}.  Then in Sec.~\ref{sec:nonGaussIntro} we introduce the photon subtraction operation that creates non-Gaussian features in these quantum states. We define the emergent network of photon-number correlations in Sec.~\ref{sec:emer-corr}.  Finally, in Sec.~\ref{sec:compNet} we review network measures in the context of complex network models.

\subsection{Continuous variable quantum optics}\label{sec:QuantumOptics}

A $m$-mode light field \cite{RevModPhys.92.035005} can be described as an ensemble of $m$ quantum harmonic oscillators with creation and annihilation operators which obey the commutation relation $[\hat a_j, \hat a_k^{\dag}]= \delta_{j,k}$. In the CV framework, we focus on ``position'' and ``momentum'' variables of these harmonic oscillators, $\hat x_k =\hat a^{\dag}_k + \hat a_k$ and $\hat p_k = i(\hat a^{\dag}_k - \hat a_k)$, also called \emph{quadratures}.
Generic quantum states of such quantum harmonic oscillators are hard to characterize, but the subclass of Gaussian states is very well understood. These states are completely described by the quadrature expectation values (the mean field) and covariance matrix $V$. To define the latter, let us introduce the $2m$-dimensional vector $\vec{\hat{\xi}} = (\hat x_1, \dots, \hat x_m, \hat p_1, \dots \hat p_m)^{\top}$, and introduce
\begin{equation}\label{eq:CovMatrixDef}
    V = {\rm Re}~\langle \vec{\hat{\xi}} \,  \vec{\hat{\xi}}^{\top} \rangle - \langle \vec{\hat{\xi}}  \rangle\langle \vec{\hat{\xi}}^{\top}\rangle, 
\end{equation}
where $\langle . \rangle$ denotes the expectation value of the observables in the state $\rho$. If this state is Gaussian, all its higher order correlations can be expressed in terms of $V$ \cite{Verbeure2011} (note that this property is explicitly used in Appendix \ref{app:corr}).

In this work, we will associate specific optical modes with the nodes of a network, as shown in Fig.~\ref{sketch}. Such networks are naturally realized in the cluster state formalism of CV measurement-based quantum computing. It was experimentally demonstrated that such cluster states can be generated in arbitrary shapes \cite{Cai17}.

\subsection{Clusters: imprinted quantum networks}
\label{sec:sub-clust}
\begin{figure}
\centering
\includegraphics[width=0.45\textwidth]{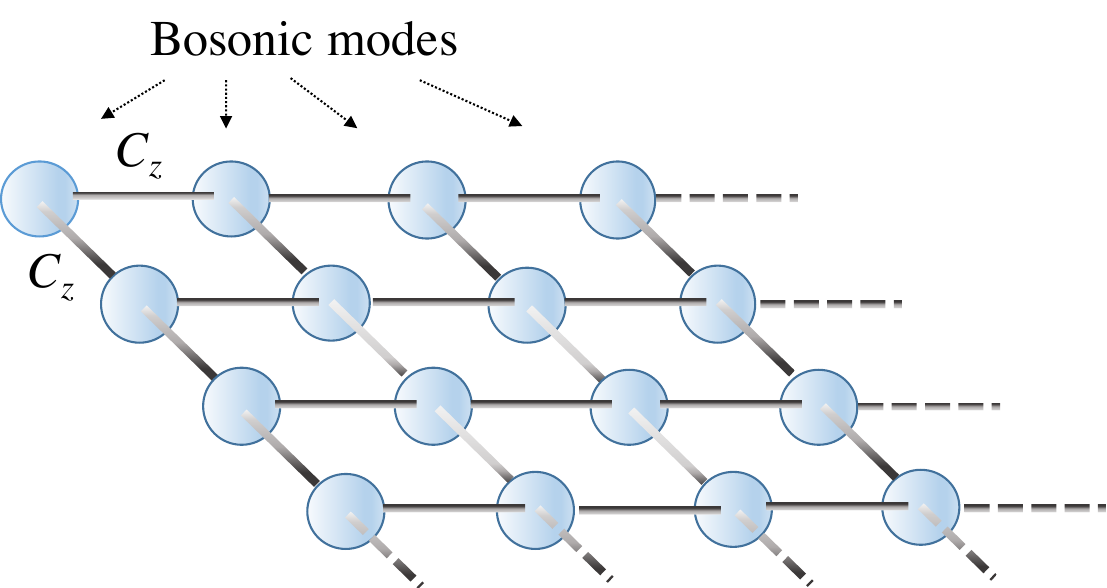}
\caption{Network representation of a two dimensional quantum cluster  state, here pictured as a network where the nodes represent the bosonic modes and the links represent $C_Z$ operations between pairs of nodes. \label{clusterpicture}}
\end{figure}
In quantum optics ideal cluster states require infinite energy to produce, it is therefor common to consider approximate cluster states, based on applying $C_Z$ gates on a set of squeezed vacuum modes. 
The finitely squeezed states can be written as $\hat{S}(s)\vert 0\rangle^{\bigotimes N}$, where $\hat{S}(s)$ is the squeezing operator, and $\vert 0\rangle$ the vacuum state. The parameter $s > 1$ denotes the squeezing, which for simplicity is chosen to be the same in all $N$ copies. The unitary $C_Z$ gates that entangle these squeezed vacuum modes are given by $C_Z= \exp(\imath \hat{x}_i \otimes \hat{x}_j)$.

This results in a Gaussian state with covariance matrix (\ref{eq:CovMatrixDef}) given by $V_s = {\rm diag}[s, \dots, s, 1/s, \dots, 1/s]$. The first $N$ elements in the diagonal are the variances of the $x$ quadrature of the $N$ modes (nodes) $\langle \hat{x}_i^2 \rangle=s\langle \hat{x}^2 \rangle_v=s$ where $ \langle x^2 \rangle_v$ is the variance of the quadrature for the vacuum state which is taken equal to 1. The last $N$ elements are the variances of the $p$ quadrature $\langle \hat{p}_i^2 \rangle=1/s$ of the $N$ modes.
The approximate cluster state that results by acting on $\hat{S}(s)\vert 0\rangle^{\bigotimes N}$ with a network of $C_Z$ gates is then described by \cite{Meni11} :
\begin{align}\label{eq:V}
    V = \begin{pmatrix}V_{xx} & V_{xp} \\ V_{px} & V_{pp}\end{pmatrix} =  \begin{pmatrix}\mathds{1} & 0 \\ {\cal A} & \mathds{1}\end{pmatrix} V_s \begin{pmatrix}\mathds{1} & {\cal A} \\ 0 & \mathds{1}\end{pmatrix} \notag \\ =\begin{pmatrix}s\mathds{1} & s {\cal A} \\ s {\cal A} & s {\cal A}^2+\mathds{1}/s\end{pmatrix}
\end{align}
Here, $V$ is a $2N \times 2N$ matrix divided into four $N \times N$ blocks. $V_{xx}$ and $V_{pp}$ describe the correlations among the $x$- and $p$-quadratures, respectively, whereas  $V_{xp}$ and $V_{px}$ contain all correlations between $x$- and $p$-quadratures.
The presence of ${\cal A}^2$ in $V_{pp}$ highlights that correlations extend not only between nearest neighbor nodes, but also between next-nearest neighbors of the imprinted network.
The elements $[\mathcal{A}^2]_{ij}$ are in fact known to correspond to the number of walks of  exactly two steps  from $j$ to $i$ of the network $\mathcal{A}$~\cite{Gu09,Walschaers18}. 

The $C_Z$ gates, which create the entanglement, can be implemented according to any network shape, i.e., any adjacency matrix $\mathcal{A}$. 
When used in measurement-based quantum computing,  cluster states are built to ensure persistence of entanglement~\cite{Briegel01}. This means that a measurement on one node only locally affects  the state and the surviving entanglement links can be further exploited for the next steps in measurement-based computing. To this end, some regular 2D graph structures, e.g., hexagonal or triangular lattices, have been proven to allow for universal computing.  That is, arbitrary unitary operations can be performed  via local operations and classical communication on the cluster.  In contrast, others have been discarded, e.g. the tree graph~\cite{VandenNest06}. 
Here we go beyond such regular structures, motivated by the fact that CV quantum networks in optical setups can be easily reconfigured to arbitrary shapes~\cite{Cai17}. We want to indeed replicate in the quantum regime  some of the models  that  mimic real-world complex networks~\cite{Nokkala18a,Sansavini20} in order to test their structural properties under local operations. \\
In the remainder of the Article we refer to the network that describes the pattern $\mathcal{A}$ of $C_Z$ gates that are applied to create the Gaussian cluster state as the {\em imprinted network}. 

\subsection{Non-Gaussian operations in continuous variable platforms}\label{sec:nonGaussIntro}
Cluster states are characterized by Gaussian statistics of quadrature measurements, which allows for a compact statistical description even when they have a large size. However, for quantum computing protocols, cluster states must also acquire non-Gaussian quadrature statistics via non-Gaussian operations. Unlike the Gaussian case, the quantum features of such non-Gaussian networks are not trivial to classify~\cite{Walschaers17,Walschaers18,Walschaers19s}.  Examples of non-Gaussian operations are the conditional implementation of  single-photon subtraction and addition, i.e. the action of  
annihilation and creation operators $\hat{a}$ and $\hat{a}^{\dag}$ ~\cite{Ra19,Ra17,Wenger04,Biagi20,Parigi07,Zavatta04,Lvovsky20}.
Such operations have long been investigated as primitive for two important operations for quantum protocols: entanglement distillation and the generation of Wigner negativity.  
Single-photon subtraction and addition can also be combined to engender high-order non-Gaussian operations~\cite{Zhuang18,Takagi18,Arzani17,Yukawa13}.\\

In this Article, we focus on multi-photon-subtraction operations that can ideally be represented, in a multimode case, by the following operation on a state $\rho$:
\begin{equation}\label{subrho}
\rho \mapsto \frac{\hat a_{S_n}\dots \hat a_{S_1} \rho \hat a^{\dag}_{S_1} \dots \hat a^{\dag}_{S_n}}{\tr[\hat a^{\dag}_{S_1} \dots \hat a^{\dag}_{S_n}\hat a_{S_n}\dots \hat a_{S_1} \rho ]},
\end{equation}
where ${S_i}$ denotes a particular mode. 
In general, Eq.~(\ref{subrho})  describes repeated subtractions from different nodes, or even from superpositions of different nodes, that have recently been experimentally implemented~\cite{Ra17,Ra19}.
When these operations are applied on multimode quantum states characterization of the resulting states is not a trivial task.  Recent results have depicted the rules of thumb for entanglement and Wigner negativity~\cite{Walschaers17,WalschaersPRA17,Walschaers18,Walschaers19s}, indicating that a deeper structural analysis would be beneficial for a more comprehensive picture.

Here, we specifically consider repeated photon subtractions from one single node of the cluster state in Eq.~(\ref{eq:V}).
There are two main reasons for this choice. First, we focus on the simplest scheme providing significant statistics. In fact, when multiple subtractions from an arbitrary superposition of nodes are considered, the analysis becomes computationally hard ~\cite{Chabaud17}. Second, we aim at probing the extent of the effect of photon subtraction in the most local way possible. 
In previous work, we have shown that photon subtraction on a given node induces non-Gaussian features in its nearest and next-to-nearest neighbor nodes~\cite{Walschaers18}. This Article goes beyond single-point features such as local averages.  Instead, we focus on the changes induced in the \emph{correlations} between those nodes, including \emph{beyond} next-nearest neighbors. 
We  then study how different imprinted network shapes  $\mathcal{A}$  spread or destroy the non-Gaussian features created by photon subtraction.  This is a problem  that is very typical for classical information networks, studied here in the new context of quantum correlations. 

\subsection{Emergent complex networks \\ of photon number correlations} \label{sec:emer-corr}
Covariance matrices are sufficient to explain the behaviour of Gaussian states. In the case of non-Gaussian states  expectation values of higher order operators are needed. 
In this Article, we focus on photon-number correlations, that are simple non-Gaussian observables with a clear physical interpretation.
Photon number correlations can be written in terms of quadrature correlations. The expression will then involve fourth moments of quadratures, which are sensitive to the non-Gaussianity -- i.e., departure from Gaussian shape -- of the quadrature distribution~\cite{WalschaersPRA17}.  
Moreover, photon-number correlations are extensively used to study features of non-Gaussian processes in quantum optics, such as Hong-Ou-Mandel interference~\cite{PhysRevLett.59.2044}, photon bipartite entanglement~\cite{PhysRevLett.102.193601}, and photon distinguishability~\cite{Walschaers_2020}. They also serve to benchmark single-photon sources~\cite{Senellart}, $n$-photon sources ~\cite{Lachman19}, and quantum protocols such as boson sampling~\cite{Walschaers_2016,Giordani}. 

To consider structural effects, we introduce a second network for each cluster state, composed by the emergent structure of photon-number correlations between pairs of modes.
As such, we define the correlation matrix $\mathbb{C}$:
\begin{align}\label{Cij}
[\mathbb{C}]_{ij} = \frac{\lvert \<\hat{n}_i \hat{n}_j\>-\<\hat{n}_i \>\<\hat{n}_j\>\rvert}{\sqrt{\left(\<\hat{n}_i^2\>-\<\hat{n}_i\>^2\right)\left(\<\hat{n}_j^2\>-\<\hat{n}_j\>^2\right)}},
\end{align}
where we take the absolute value of the correlation, since we are purely interested in the strength of the correlation, rather than its sign. 
The values of $\lvert \<\hat{n}_i \hat{n}_j\>-\<\hat{n}_i \>\<\hat{n}_j\>\rvert$ depend on the number of photons in the system; it may be higher for two weakly correlated nodes with very high photon numbers, than for strongly correlated nodes with very small photon numbers. Due to its conditional nature, photon subtraction locally changes the photon number in the system, thus making it impossible to genuinely compare the resulting values of $\lvert \<\hat{n}_i \hat{n}_j\>-\<\hat{n}_i \>\<\hat{n}_j\>\rvert$ in the two cases. The denominator in Eq.~(\ref{Cij}) solves this problem by renormalizing the correlation to be confined between zero and one, where one implies that both nodes contain the same number of photons, regardless of how many photons there are.

Eq.~(\ref{Cij}) ultimately allows us to look at the correlation network, as given by its weighted adjacency matrix:
\begin{equation}\label{corr}
\mathbb{A} = \mathbb{C} - \mathds{1}.
\end{equation}
In the following we will characterize what kind of correlation networks ($\mathbb{A}$) emerge, in the same spirit of mutual information networks in~\cite{Valdez17},  when photon-subtraction operations are applied on cluster states with different shapes $\mathcal{A}$. Hence, network structures are considered at two different levels: one is the imprinted network which describes the entangling gates that induce Gaussian entanglement between position and momentum in the cluster state of Eq.~(\ref{eq:V}); the other is the emergent network of photon-number correlations defined in Eq.~(\ref{corr}) via Eq.~(\ref{Cij}). This second network will be analyzed via measures and metrics typical of complex network theory, as we detail in the following in Sec.~\ref{sec:compNet}.

\begin{figure}
\centering
 {\includegraphics[width=0.28\textwidth]{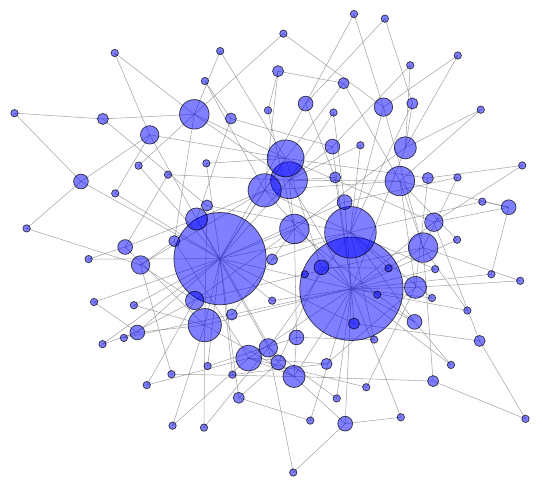}} \quad
   {\includegraphics[width=0.28\textwidth]{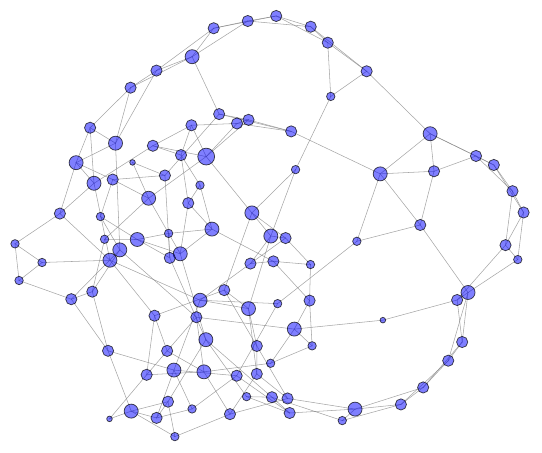}} \quad
     {\includegraphics[width=0.28\textwidth]{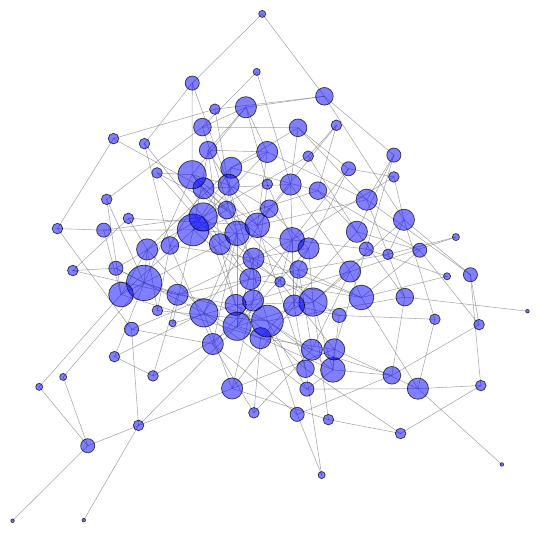}} \quad
\caption{Three complex networks of $100$ nodes.  Top: Barab\'{a}si-Albert (BA) network built via preferential attachment with $m=2$ nodes added at each step.  Middle: Watts-Strogatz (WS) network built via rewiring, with probability $p_{WS} = 0.2$, a regular network with degree per node $d=2$.  Bottom: Erdos-Renyi (ER) network with connection probability $p_{\mathrm{ER}}=0.04$. The size of the nodes is proportional to their degree. \label{network}}
\end{figure}

\subsection{Complex network measures \\ in complex network models}\label{sec:compNet}
Quantitative measures of network structures have been introduced by network theory~\cite{Newman18,Barabasi16}. From the adjacency matrix components we can calculate the degree $D_i$ for each node $i$, i.e., the number of links connected to it, as \begin{equation}\label{eq:degree}D_i= \sum_j \mathbb{A}_{ij}.\end{equation} The degree distribution $p(D)$ gives the probability for a randomly picked node to have the degree $D$.
Even if is not possible to get the full information on the structure of a network by its degree distribution, it can be  informative to look at the shape of the distribution. Similarity arguments between different real-world networks have been based in part on analogy between their degree distributions, which often follow a power-law rule. Moreover, many crucial properties of networks, like their robustness to perturbations and the spread of contamination, are determined by  the functional $p(D)$ \cite{Barabasi16}.

A second quantitative measure of complexity is the local clustering coefficient. It gives information on the connections between the neighbors of a specific node, thus keeping track of local correlations around a point. 
A common way of defining the clustering coefficient is the number of triangles to which the node belongs divided by the number of triplets. It can be recovered from $\mathbb{A}$ as
\begin{equation}\label{eq:cluster}
Cl_i=\dfrac{\sum_{ j \neq  k} \mathbb{A}_{ij} \mathbb{A}_{jk} \mathbb{A}_{ki}}{\sum_{ j \neq  k} \mathbb{A}_{ij} \mathbb{A}_{ik}  }
\end{equation}
for $i \neq j \neq  k$.\\

 Here we briefly review some of the paradigmatic models that have been proposed for real-world networks: the random network model called Erd\H{o}s-R\'{e}nyi (ER), the Barab\'{a}si-Albert (BA) model and the Watts-Strogatz (WS) model.
The ER model builds networks by randomly connecting nodes according to a uniformly random probability $p_{\mathrm{ER}}$  for two nodes to be connected. The resulting networks exhibit a binomial distribution of links per node.  The ER model is able to reproduce the typical average shortest path distances between nodes of real networks. 

A second model that has been introduced to reproduce typical complexity signatures of real networks is the BA model. It describes network formation processes based on the preferential attachment model: the network grows by adding new nodes.  These new nodes attach with $m$ links to old nodes.  The probability of connection is proportional to the  degrees of the existing nodes, such that the highest degree nodes are the preferred ones. This model is able to reproduce the power-law distribution in the degree, and thus the existence of ``hubs'', i.e., nodes with very large degree, as in real-world networks.

Finally, the WS model is able to reproduce the small-world mechanism, where any node is a short path from any other in the network.  Specifically, the distance between any two nodes grows as the log of the total number of nodes. It is built by starting from a regular network in which each vertex has a fixed degree $k$; for instance, $k=2$ would correspond to a lattice in tight binding approximation. Then nodes are rewired according to a probability $p_{WS}$. One interesting feature of this model is that it allows one to tune continuously from regular ($p_{WS}=0$)  to random ($p_{WS}=1$) networks.

To achieve reasonable statistics, we consider many realisations of networks made of $100$ nodes for each model. For every model we also explore different parameters. These networks are small compared to typical real-world networks, but even for this small scale the different models exhibit visibly different features. In Fig.~\ref{network}, we show a BA network built by adding $m=2$ new nodes at each step in network growth; a WS network built starting from a regular network with degree per node $k=\langle D\rangle =2$ and rewired with a probability $p_{WS} = 0.2$; and an ER network with connection probability $p_{\mathrm{ER}}=0.04$. One observes clear differences between the three networks, with, for example, the emergence of easily visible hubs in the BA model, shown as large blue discs in the figure. By taking $100$ network realizations for each model one observes that the resulting degree distribution, shown in Fig.~(\ref{histo}), is distinct in the three cases, even if they have similar average value. In particular, the logarithmic scale shows the power-law distribution for the BA networks. 

In the following, we choose particular network models to shape the adjacency matrix  $\mathcal{A}$ of the imprinted networks. In the rest of the Article we will consider only the  BA and WS models as the ER network shows very similar features as the WS models with high rewiring probability $p_{WS} \to 1$. 
With the probabilistic generation of a statistically significant number of networks for each model, it will be possible to reveal specific features, in this case quantum ones, that are determined by the structure of the network.

\begin{figure}
\centering
\includegraphics[width=0.45\textwidth]{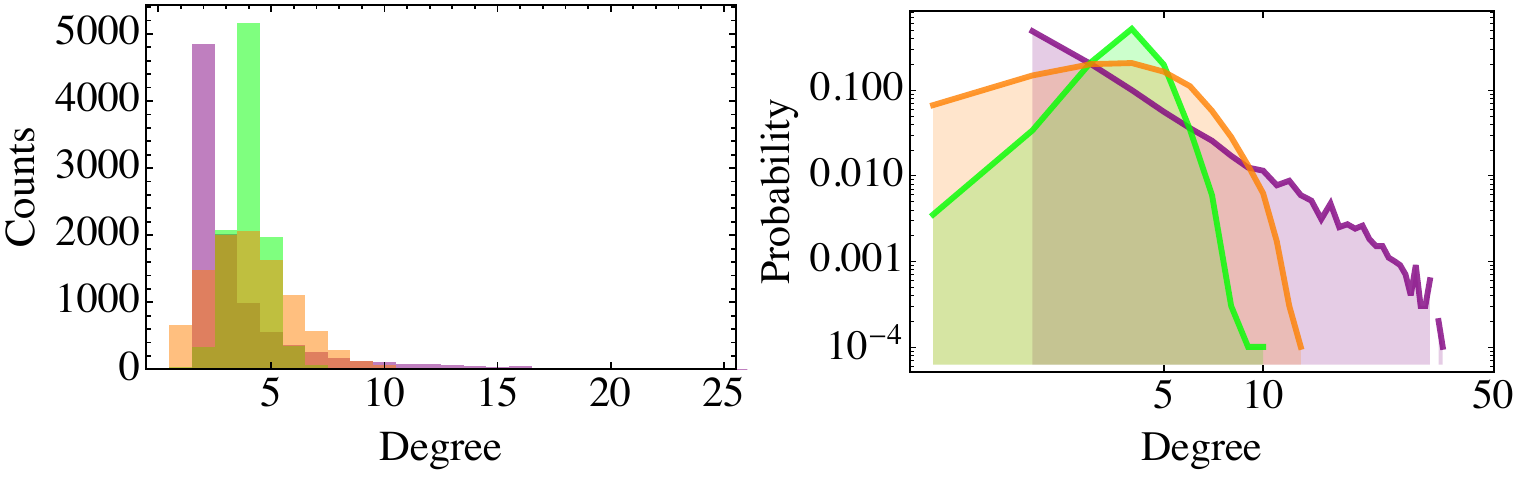}
\caption{Histogram of the degree distribution for $100$ networks of the three types shown in Fig.~\ref{network},  BA (purple), WS (green), ER (orange). On the left the linear scale is used while on the right the scale is double-logarithmic to emphasize the appearance of the power-law tail for the BA networks. The average degree is $\langle {\cal D} \rangle= 4.44$ (BA), $4.0$ (WS) and $4.09$ (ER). \label{histo}}
\end{figure}

\section{Emergent networks in complex Gaussian cluster states}\label{sec:GaussianNetworks}

\begin{figure*}
\centering
\includegraphics[width=0.95\textwidth]{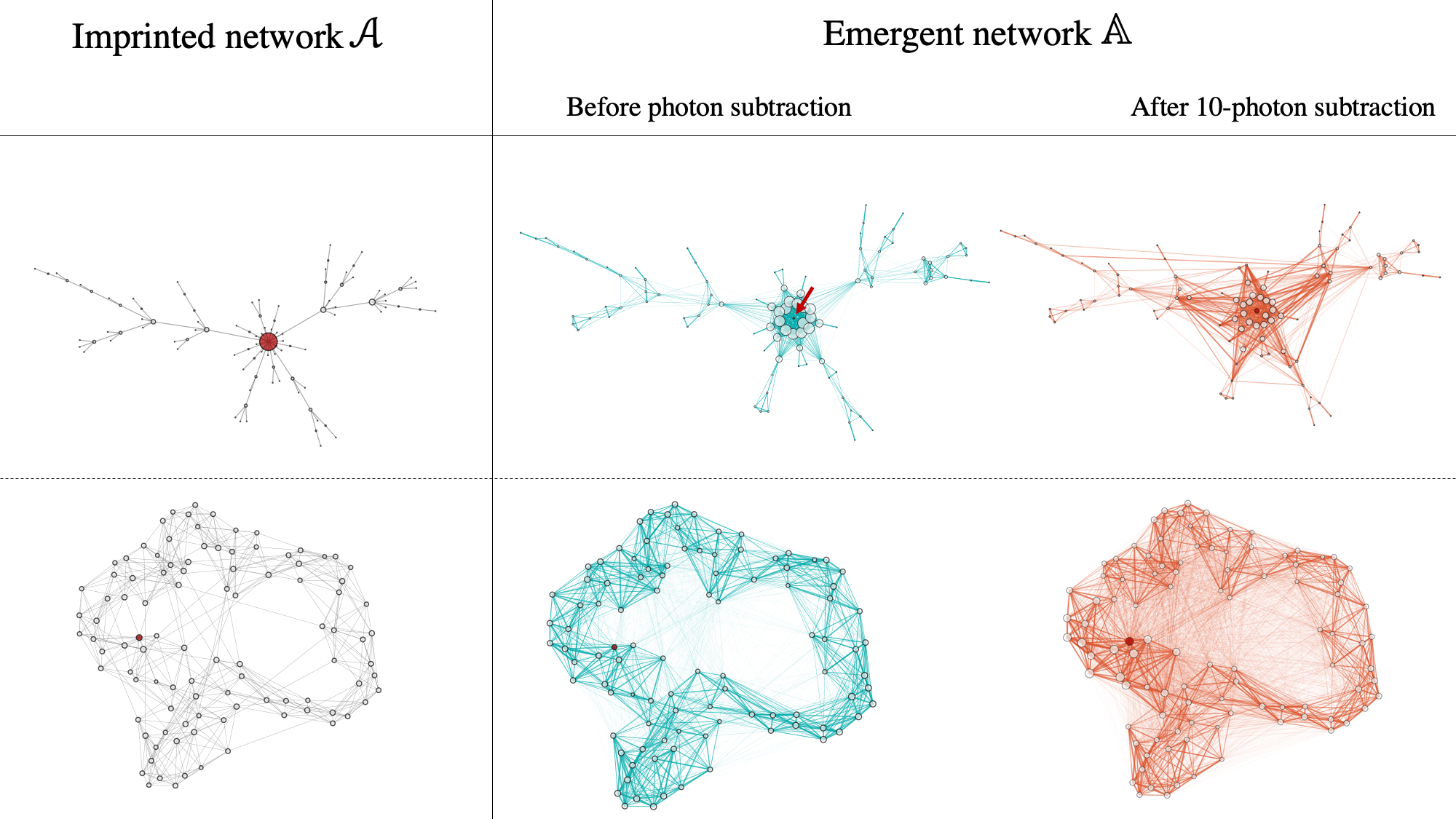}
\caption{Imprinted networks $\mathcal{A}$ (left column) give rise to emergent networks $\mathbb{A}$ (middle column), which can then further undergo photon subtraction (right column). Photon subtraction at indicated red node; node sizes show the degree. Imprinted networks include Barab\'{a}si-Albert with $m=1$ (top row) and Watts-Strogatz generated from a regular one-dimensional network in which every node is connected to its $k=5$ nearest neighbors with a rewiring probability $p_{WS}=0.05$ (bottom row).  \label{networks}}
\end{figure*}

In this section, we explore the emergent photon-number correlation networks for different imprinted networks before any photon subtraction. The quantum state of such networks hence exhibits Gaussian statistics of quadratures. The results of this section form a benchmark to compare with the  effect of photon subtraction in Sec.~\ref{sec:NonGaussianNetworks} and Sec.~\ref{sec:AllDistance}.

The imprinted networks are obtained by applying $C_Z$ gates to a set of squeezed vacuum modes according to an adjacency matrix ${\cal A}$ for the BA and WS models defined in Sec.~\ref{sec:compNet}.  We then examine the emergent network with adjacency matrix $\mathbb{A}$. Throughout all our simulations, we fix the amount of squeezing to $15 {\rm dB}$ (i.e., $s\approx 31.6$ units of shot noise) for each squeezed vacuum mode.

As described in Sec.~\ref{sec:sub-clust}, the correlation between quadratures of different modes goes beyond the graphical structure imprinted by the $C_Z$ gates, as they appear between nearest neighbours but also between next-nearest-neighbors. We then expect photon number correlations to inherit the same behaviour.

The calculation of photon number correlations for the cluster before photon subtraction can be carried out analytically by using the techniques of Appendix \ref{app:corr}. We obtain the weighted adjacency matrix (as derived in Appendix \ref{App:Gaussian})
\begin{align}\label{Cij-Gaus}
[\mathbb{A}]_{ij}^G = \begin{cases}\frac{\frac{s^2}{8}(([\mathcal{A}^2]_{ij})^2+2\mathcal{A}_{ij})}{\sqrt{\mathfrak{N}(s,{\cal D}_i)\mathfrak{N}(s,{\cal D}_j)}}, &\text{for $i \neq j$}\\
0 &\text{for $i = j$}
\end{cases}
\end{align}
where $\mathfrak{N}(s,{\cal D}_k)=( s^2 + 1/s^2 + s^2({\cal D}_k)^2 + 2{\cal D}_k - 2)/8$ is a normalization factor depending only on the initial squeezing value $s$ and the degree ${\cal D}_k$ of the node $k$ in the imprinted network. Recall from Sec.~\ref{sec:sub-clust} that $[\mathcal{A}^2]_{ij}$ is the number of different walks of exactly two steps that connect nodes $i$ and $j$ in the imprinted structure.
Therefore, as anticipated, the links between nodes $i$ and $j$ in the emergent network are non-zero if either $i$ and $j$ are connected in the imprinted network ($\mathcal{A}_{ij}$ = 1) or when they are next-nearest neighbors ($[\mathcal{A}^2]_{ij} \neq 0$).

So emergent networks of photon-number correlations have larger number of links than the imprinted networks. Also, the number of walks of distance two between different nodes in complex networks are larger than in regular structures (like grid shapes). Hence we expect to have a larger number of links  for emergent networks of cluster with complex imprinted network.

We now look at specific features dependent on the different network structures.

\subsection{Barab\'{a}si-Albert networks -- 
Emergent triangles and clustering}
The imprinted BA networks have a multitude of weakly connected nodes that are organized around a few highly connected hubs. 
\begin{figure}
\centering
\includegraphics[width=0.5\textwidth]{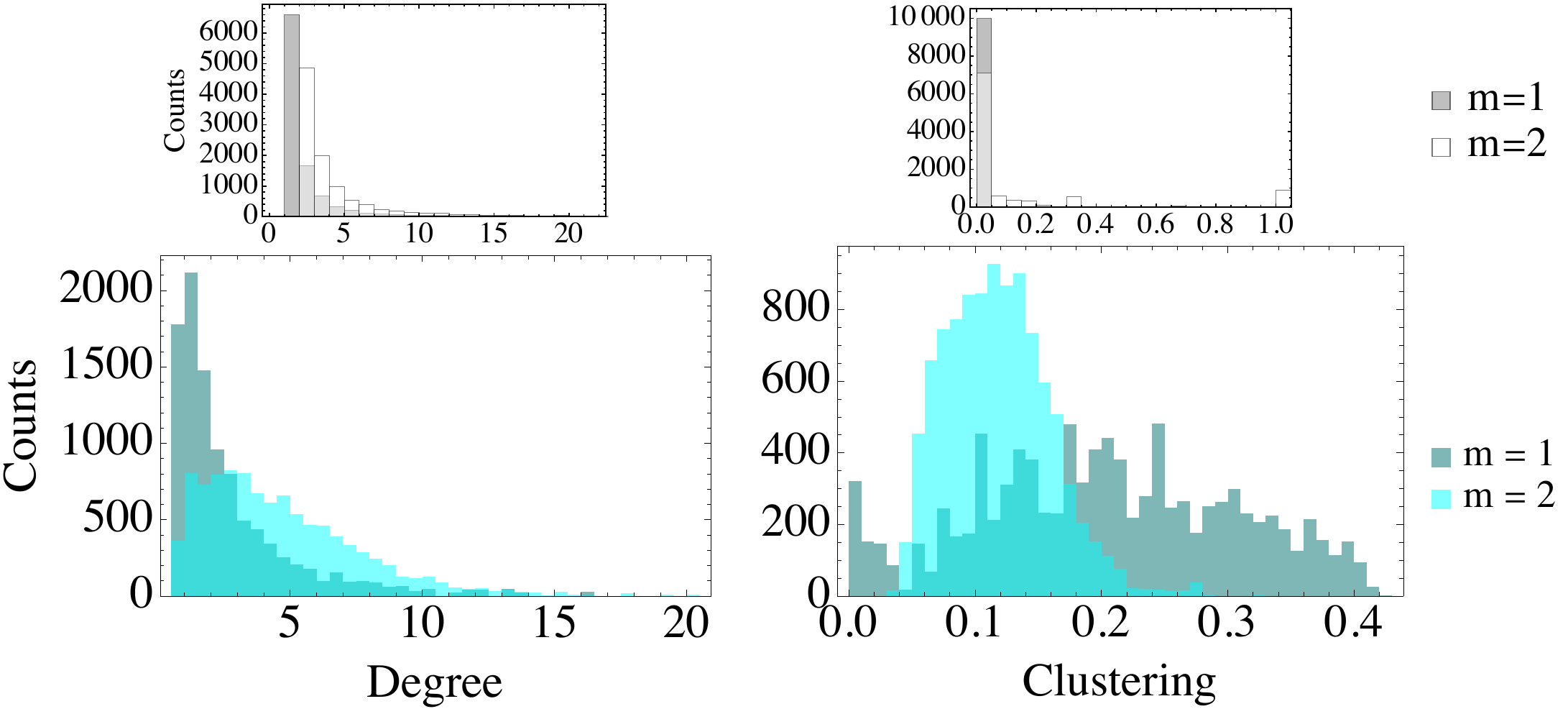}
\caption{\textit{Statistics for imprinted BA networks and emergent networks in the Gaussian case.} Histogram of degree and clustering for the imprinted BA network with $m=1$ and $m=2$ (top row) and associated emergent network of photon number correlations (bottom row) in the Gaussian case, i.e., when no photon is subtracted.}\label{EmGaussianBA}
\end{figure}
We collect statistics of $100$ different networks of $100$ nodes both for the parameter $m=1$ (e.g. top row of Fig.~\ref{networks}) and for $m=2$. 
In the example in Fig.~\ref{networks}, we see that the number of links in the emergent network are larger when compared to the imprinted network, as told above.
A more quantitative understanding is acquired from the histograms of the degree and clustering coefficients in the emergent network of Fig.~\ref{EmGaussianBA} for $m=1,2$ in comparison with the original distribution of the imprinted network. The emergent degree distributions inherit the features of the imprinted network, with a small number of nodes with high degrees, although with larger variances. In contrast, the histogram of clustering is dissimilar to the clustering in the imprinted network. 
The BA network with $m=1$ is an excellent example to illustrate the difference: this imprinted network's tree-like structure combined with the randomness of the BA growth process makes that many nodes have only one connection. In the emergent network, however, all nodes have at least two connections due to what we discussed above, i.e., the presence of walks at distances two in the imprinted network. So we have more triangles  than in the imprinted network. Thus clustering is zero for all nodes while the emergent correlation network has non-zero values quite uniformly distributed but only in the range of $0.0 - 0.4$.

\subsection{Watts-Strogatz networks -- \\ More randomness  for 
lower degree and clustering }

As mentioned in the introduction to complex networks of Sec.~\ref{sec:compNet}, WS networks  have a tuneable degree of randomness. In particular, in the limit of a vanishing rewiring probability we recover a completely regular network, whereas in the limit of high rewiring probability the network closely resembles an ER network.
In our simulations, we start from a regular one-dimensional network with $100$ nodes, each of which is connected to $2k$ other nodes. This network can be represented by organizing the nodes in a circle, where every node is connected up to its $k^\mathrm{th}$ neighbor. Subsequently we rewire the connections with probability $p_{WS}$. For various choices of $p_{WS}$, we implemented $100$ of these WS networks as imprinted structures to apply $C_Z$ gates \footnote{The case where $p_{WS} = 0.05$ forms an exception. Here we consider $74$ realizations.}. As for the BA case we look at the statistics of degree and clustering.

The bottom row of networks in Fig.~\ref{networks} shows a typical realization of a WS network with $k=5$ and $p_{WS}=0.05$. The imprinted network is therefore reasonably close to a regular network in which each node has $2k=10$ connections. We observe that the emergent network before photon subtraction, with a weighted adjacency matrix $\mathbb{A}$, has a richer structure in its connections. Nevertheless, we can still see a qualitative resemblance between the imprinted and the emergent network.
\begin{figure}
\centering
\includegraphics[width=0.5\textwidth]{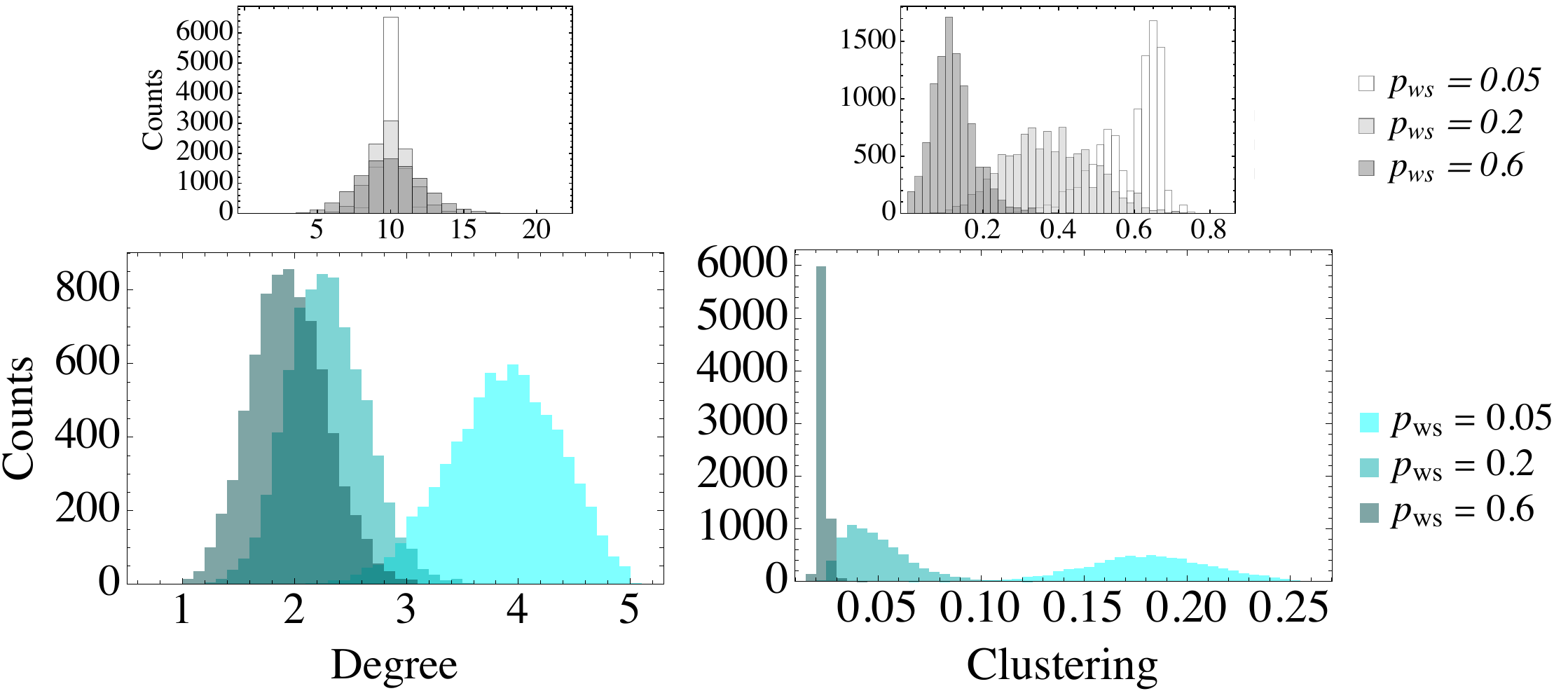}
\caption{\textit{Statistics for imprinted WS networks and emergent networks in the Gaussian case..} Histogram for the degree and clustering for the imprinted networks (top row) and for the emergent network of photon number correlations in the Gaussian case (when no photon is subtracted) (bottom row). Here results on the WS network model with $p_{WS}=0.05, 0.2$ and $0.6$ are reported. }\label{EmGaussianWS}
\end{figure}

\begin{figure*}
\centering
\includegraphics[width=0.99\textwidth]{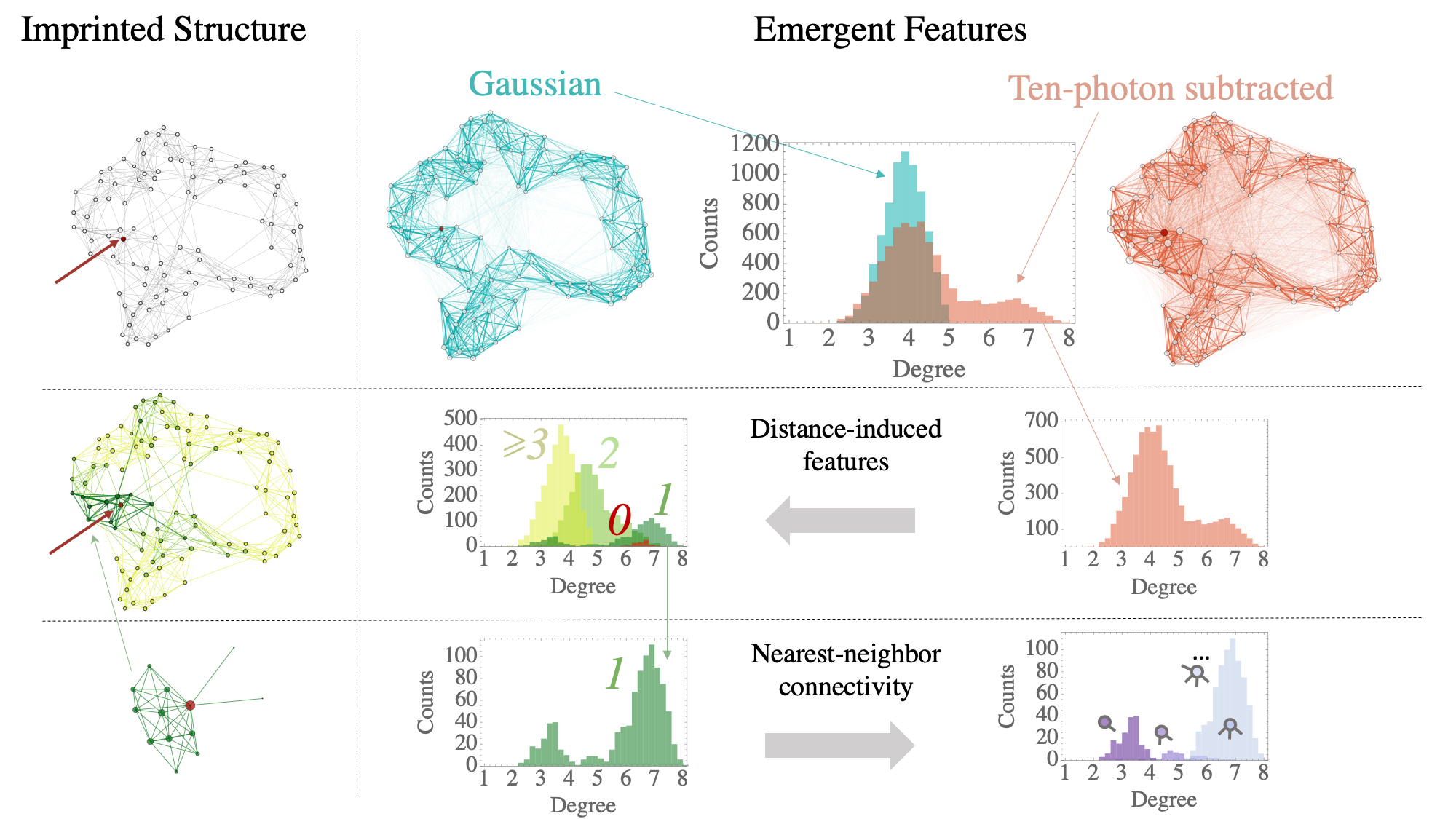}
\caption{\textit{Analysis and structures for an imprinted WS network.}  Generated from a regular one-dimensional network in which every node is connected to its $k=5$ nearest neighbors with a rewiring probability $p_{WS}=0.05$. Top row: imprinted complex network (left, node with highest degree is highlighted in red); and emergent networks before subtraction (Gaussian, middle) and after subtraction (Ten-photon subtracted, right) with their corresponding degree distributions. Middle row: degree distribution of the photon-subtracted case (zoom on right) is broken up, in the central panel, according to the distance ($0$, $1$, $2$, or $\geqslant 3$) of nodes from the subtraction node. The color code used for different distances is adopted for the nodes in the imprinted network (left). Lower row: Structure of the next-neighbor nodes (distance $1$) is highlighted (left); zoom of the degree distribution at distance one (middle); statistics is broken up according to the connectivity of the different nodes (right). \label{overview}}
\end{figure*}

In Fig.~\ref{EmGaussianWS} we examine the difference in degree and clustering coefficient between imprinted and emergent networks. We observe that the properties of the imprinted WS networks strongly influence the structure of the emergent correlation networks.
The degree distribution for the imprinted networks is always centered around $2k=10$ with larger variances for larger $p_{WS}$.
The degree distributions for the emergent networks are centered around different mean values   for the three $p_{WS}$ cases. The $p_{WS} = 0.05$  case shows a broader and more skewed distribution of significantly higher degrees. Hence, for the emergent networks, in contrast to the imprinted ones, the largest variance is for the lowest $p_{WS}$. In general, we conclude that an increased probability of rewiring (and thus more randomness) in the imprinted network decreases the degree (which is essentially the total amount of correlation of every node) in the correlation network of Gaussian clusters.
The  histogram for the clustering coefficient is qualitatively similar to that of the degree, in the sense that increased rewiring leads to a decrease in clustering, and it is also very similar to the clustering of the imprinted networks \footnote{This reduction in the degree and clustering should be explained by a reduction in the weight of the connections because, from (\ref{Cij-Gaus}), we can demonstrate that the number of connection in the emergent network is higher for larger values of $p_{WS}$. This is related to the choice of the normalization given by the denominator in the elements of the correlation matrix, as explained in Sec.~\ref{sec:emer-corr}.}.

\section{The effect of non-Gaussian operations on Emergent Networks}\label{sec:NonGaussianNetworks}
In this section we study the effect of photon subtraction, introduced in Sec.~\ref{sec:nonGaussIntro}, on the emergent network of photon-number correlations. Because the operation is locally applied on a single node in the imprinted network, one might consider this to be a single-node attack, as in classical complex network theory. Moreover in quantum networks, operations like node removal or link shortening have to be introduced in the context of cluster states via Gaussian (homodyne) measurements~\cite{Gu09}.  However, we emphasize that no nodes are removed in the photon subtraction process.  So it is not a single-node attack in the classical sense.  

Here we are not interested in modifying the size of the imprinted cluster by removing nodes, but we instead want to analyze the spreading of non-Gaussian correlations in cluster states when affected by photon subtraction in one node.
Previous results show that  repeated photon subtraction in the same node may increase correlations in the system due to entanglement distillation \cite{PhysRevA.86.012328}. Also   we know that   photon subtraction in a given node creates correlations between previously uncorrelated nodes~\cite{Walschaers17}.

However, there is no general result on how the structure of the  correlation  in the network is influenced by  the topology of the imprinted network.

To address this question, we monitor the effect of subtracting ten photons for the emergent photon-correlation networks. Our procedure is the following: i) we first provide analytical results on the reach of the effect of photon subtraction.  ii)  We then compare the qualitative features that are seen in the histograms of numerically generated distributions of degrees and clustering coefficients. iii) To get a complementary quantitative view, we perform a moment analysis and probe the effect of the non-Gaussian operation on the mean, variance, skewness, and kurtosis. Readers unfamiliar with these quantities can find their definitions in Appendix~\ref{sec:Moments}. The results of this section will guide the analysis of distance -induced structures in the follwing section  \ref{sec:AllDistance}. 
An overview of the path followed in our network analysis can be found in Fig.~\ref{overview}.  In the remainder of the Article we will explain each one of these steps in detail.

The number of photons to be subtracted (ten) is chosen in order to to have a large effect on the emergent network, although we do not find qualitatively different results for somewhat larger or smaller numbers of subtracted photons. However, increasing the amount of squeezing in the initial imprinted network or the number of photon subtractions does quantitatively enhance the observed features. 

\subsection{The effect of photon subtraction is strictly local}\label{sec:local}

The subtraction of a single photon in a cluster state is known to only affect vertices in the vicinity of the node of subtraction \cite{Walschaers18}. In Appendix \ref{App:Analitics}, we extend this understanding to the correlations between observables that are defined on different regions of the system: the correlation $\<\hat X \hat Y\> - \<\hat X\>\< \hat Y\>$ between two observables $\hat X$ and $\hat Y$ can only be influenced by photon subtraction when both observables have a support on modes that are correlated to the mode of photon subtraction. This result applies regardless of the number of photons that are subtracted.

For the networks in this work, we subtract photons in one specific node. In the initial Gaussian state, Eq.~(\ref{eq:V}) shows that this node is correlated to all nodes that are either nearest-neighbours (given by ${\cal A}$) or next-to-nearest neighbours (given by ${\cal A}^2$) in the imprinted network of $C_Z$ gates. We label $S$ the node of photon subtraction and denote the expectation value in the photon subtracted state $\<\dots\>$, whereas $\trho{\dots}$ is the expectation value in the Gaussian cluster state. Our general result of Appendix \ref{App:Analitics} then shows that $\<\hat n_i \hat n_j\> - \<\hat n_i \>\< \hat n_j\> = \trho{\hat n_i \hat n_j} - \trho{\hat n_i}\trho{\hat n_j}$ for all nodes $i$ and $j$ which satisfy the condition that either $\delta_{S,i} = {\cal A}_{S,i} = ({\cal A}^2)_{S,i} = 0$ or $\delta_{S,j} = {\cal A}_{S,j} = ({\cal A}^2)_{S,j} = 0$. More colloquially phrased, photon subtraction in $S$ only alters the value of $\<\hat n_i \hat n_j\> - \<\hat n_i \>\< \hat n_j\> $ if either node $i$ or node $j$ is a nearest or next-to-nearest neighbour of $S$ in the imprinted network.

Looking back at Eq.~(\ref{Cij}), we must also consider the effect of photon subtraction on the denominator on the final emergent network. Let us assume that node $i$ is within the vicinity of $S$, i.e., $\max\{\delta_{S,i}, {\cal A}_{S,i}, ({\cal A}^2)_{S,i}\} \geqslant 1$, but $j$ is further away, i.e., $\delta_{S,j} = {\cal A}_{S,j} = ({\cal A}^2)_{S,j} = 0$. In this case, after photon subtraction in $S$ we obtain the Gaussian value when we calculate $\<\hat n_i \hat n_j\> - \<\hat n_i \>\< \hat n_j\>$. Likewise, we will find that $\<\hat n_j^2\> - \< \hat n_j\>^2$ remains unaffected. However, because $i$ is in the vicinity of $S$, we do find that $\<\hat n_i^2\> - \< \hat n_i\>^2$ changes its value. In other words, photon subtraction will still have an effect on the emergent network as constructed via Eq.~(\ref{Cij}), unless $\<\hat n_i \hat n_j\> - \<\hat n_i \>\< \hat n_j\> = \trho{\hat n_i \hat n_j} - \trho{\hat n_i}\trho{\hat n_j} = 0$. From Eq.~\eqref{Cij-Gaus}, we see that this implies that the effects of photon subtraction \emph{spread up to four steps in the emergent correlation network}.\\

\begin{figure}
\centering
\includegraphics[width=0.49\textwidth]{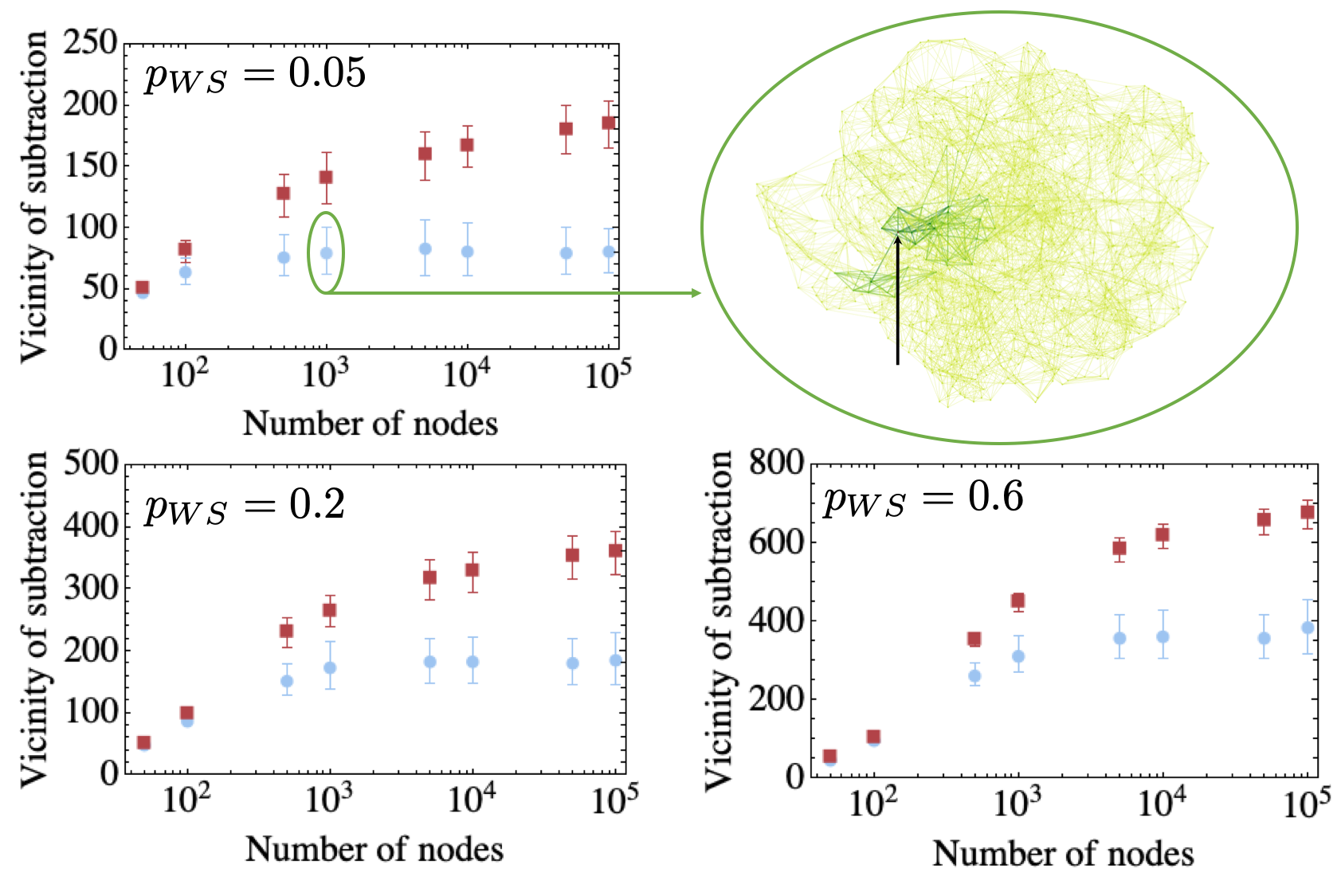}
\caption{The nearest and next-to-nearest neighbours of the photon-subtracted node in imprinted WS network are counted and represented as the {\em vicinity of subtraction}. The number of nodes in this vicinity is shown as a function of the total number of nodes in the network. Every point shows the average of 100 realizations of the network, and the error bar shows the standard deviation around the average. The WS networks are generated from a regular one-dimensional network in which every node is connected to its $k=5$ nearest neighbors with a varying rewiring probability $p_{WS}=0.05$, $p_{WS}=0.2$, and $p_{WS}=0.6$ for the different figures. Each plot shows two possible scenarios: one where the photon is subtracted in a random node (light blue) and one where it is subtracted in the node with the highest connectivity (dark red). Finally, we also show one explicit example of a distance-resolved imprinted network of 1000 nodes, with the vicinity of the photon-subtracted node highlighted in darker green. \label{WSvicinity}}
\end{figure}

We therefore have proven that effects of photon subtraction in such a multimode system can only affect a certain environment around the node of subtraction. Moreover, the number of nodes, in which the effect of photon subtraction is felt, is independent of the number of subtracted photons. Hence, to study the effect of photon subtraction, we can restrict ourselves to intermediate network sizes, that have a large fraction of the nodes that lie in the vicinity of the photon-subtracted node. 

Figs.~\ref{WSvicinity} and \ref{BAvicinity} show how the size of these neighborhoods increases with the size of the network. In almost all cases we consider, we find that this growth either stops or stagnates when the size of the neighborhoods  reaches around 100 nodes. The notable exception is the case of BA networks when we subtract a photon in the most highly connected node. This observation is consistent with the fact that BA networks can have very high connectivities, as is also shown by the power-law statistics in \ref{histo}.

Note that for WS networks the size of the vicinity of the photon-subtracted node could also be changed by increasing or decreasing the connectivity $k$ of the initial regular network that is rewired. However, throughout our text we choose to keep it constant at $k=5$. As such, for the types of networks and the parameter ranges we consider a choice of $\sim 100$ nodes guarantees that most correlations in the system are affected by photon subtraction in a single node.

\begin{figure}
\centering
\includegraphics[width=0.5\textwidth]{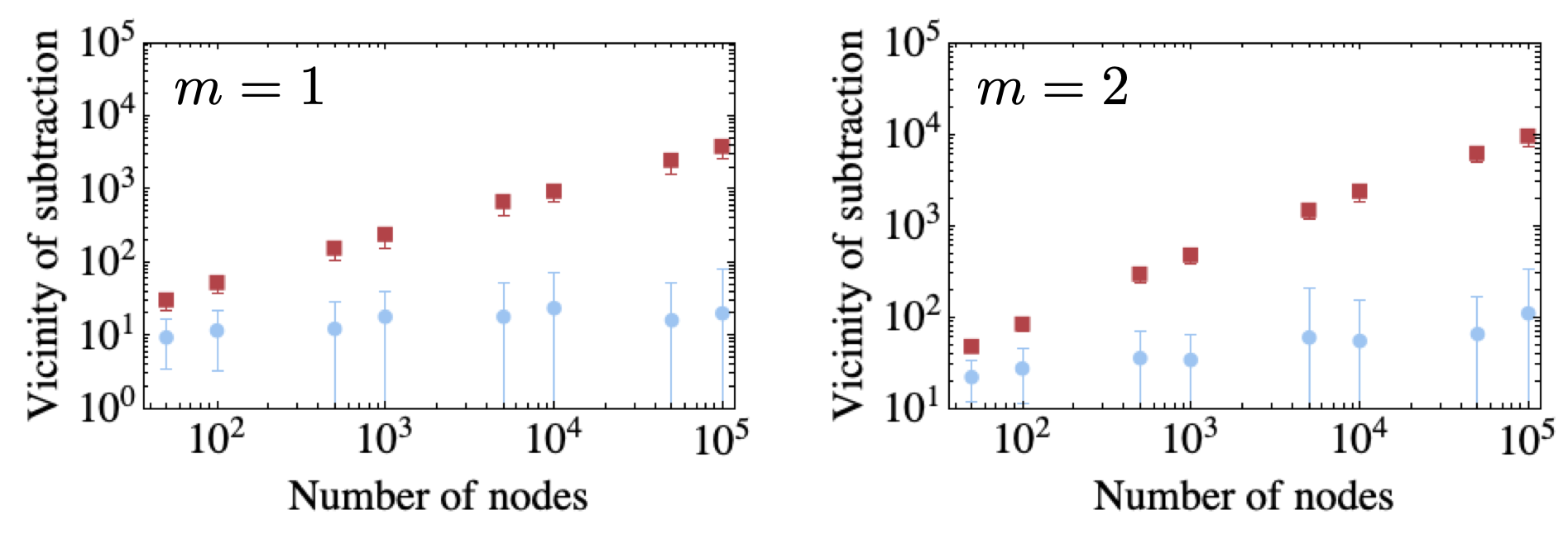}
\caption{The nearest and next-to-nearest neighbours of the photon-subtracted node in imprinted BA network are counted and represented as the {\em vicinity of subtraction}. The number of nodes in this vicinity is shown as a function of the total number of nodes in the network. Every point shows the average of 100 realizations of the network, and the error bar shows the standard deviation around the average. The BA networks are generated for parameters $m=1$ and $m=2$ in the different panels. Each plot shows two possible scenarios: one where the photon is subtracted in a random node (light blue), and one where it is subtracted in the node with the highest connectivity (dark red). \label{BAvicinity}}
\end{figure}

As we will show, the methods of complex network theory offer a new window to understand how photon subtraction influences the correlations in the relevant range.\\

We also arrive at another important conclusion: to induce non-Gaussian effects in vast cluster states, one must subtract photons in many different nodes. However, in Appendix \ref{compSP}, we argue how the complexity of this problem effectively makes it computationally hard to simulate. From a physical point of view, one would subtract these photons by coupling a tiny amount of light from the subtracting nodes,  into an auxiliary mode. Then we need photon detectors on these auxiliary modes to fire at the same time.  We can notice the connections to Gaussian boson sampling \cite{PhysRevLett.119.170501,PhysRevA.99.053816}, where it is shown that simulating the clicks of photon detectors mounted on a sufficiently complicated Gaussian states is computationally intractable. Similarly, there is also a direct connection to the hardness of sampling continuous variables on a photon subtracted state \cite{PhysRevA.96.062307}.

On a mathematical level, the problem at the basis of the computational complexity of these sampling problems is  finding perfect matchings \cite{deshpande2021quantum}. As we argue in detail in Appendix \ref{compSP}, the problem of finding all perfect matchings also appears when constructing the emergent network of photon-number correlations. Hence, fully simulating such networks in detail is only possible when many photons are subtracted in many modes.

Yet, when large states with photons subtracted in various modes are created in experiments, the measurement and analysis of emergent correlation networks may well turn out to be an important tool to characterise such states. Note that the same type of photon-number correlations have been used to benchmark computationally intractable Gaussian boson sampling experiments \cite{Zhong1460,zhong2021phaseprogrammable}.

\begin{figure*}
\includegraphics[width=0.99\textwidth]{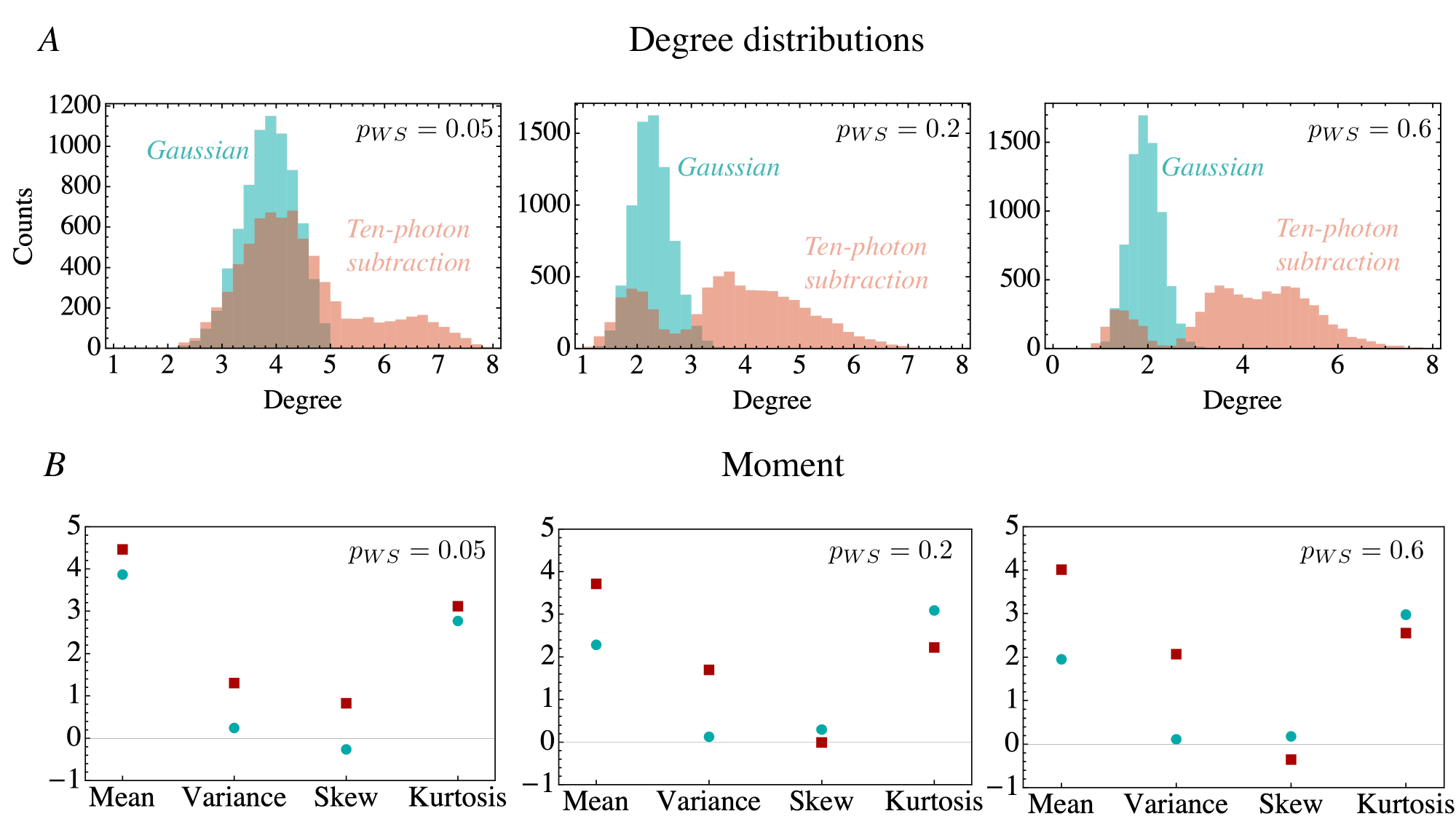}
\caption{Histograms (A) and moments (B) for the degree distributions of the emergent correlation network obtained from a WS imprinted structure with rewiring probabilities $p_{WS}=0.05,0.2$ and $0.6$. Colors indicate degree data prior to (cyan, dots) and after (red, squares) the subtraction of ten photons in the node with the highest connectivity. Data were each obtained by combining $74$ random realizations of a $100$-node network. The WS imprinted network is obtained by starting from a one-dimensional regular network where each node is connected to its $k=5$ nearest neighbors. Both moments and histograms show how photon subtraction changes the bulk of the distribution by increasing the mean degree and the width of the distribution (i.e., variance). The higher moments and histograms also show that the finer structure in the tails of the degree distribution depends strongly on the value of $p_{WS}$. \label{PhotonSub-Degree-WS}}
\end{figure*}

\subsection{Photon subtraction in Watts-Strogatz networks -- More randomness for larger effects } 
\label{ssec:ws}
Fig.~\ref{WSvicinity} suggests that we can maximize the effect of photon subtraction by subtracting the photons in the node with the highest connectivity in the imprinted network, i.e., the biggest hub. As such we probe network environments in the imprinted WS structure with the highest correlations. For the considered network size of 100 nodes, this choice has a small effect in the case of WS networks as most nodes have a similar connectivity, unlike BA networks where a few nodes serve as highly connected hubs.\\ 
In Fig.~\ref{PhotonSub-Degree-WS} we choose rewiring probabilities $p_{WS}=0.05,0.2,0.6$, as for the Gaussian case, to probe the effect of different imprinted network environments on the degree distribution in the emergent network of photon-number correlations. The data for each value of $p_{WS}$ are obtained by combining $74$ random realizations of a $100$-node network. 
The effect of photon subtraction is qualitatively similar in all cases.  A subset of nodes in the photon-subtracted cluster states retains degrees of the same order of magnitude as for the Gaussian network state, whereas a second subset finds its degree considerably increased, resulting in a bimodal distribution. This qualitative similarity translates to the moments in Fig.~\ref{PhotonSub-Degree-WS}(B), in the sense that photon subtraction shifts the distributions to higher means and variances, regardless of the value  of $p_{WS}$. However, photon subtraction causes stronger increases in the mean and variance for larger values of $p_{WS}$, and the higher moments behave differently depending on $p_{WS}$. These features are observed in Fig.~\ref{PhotonSub-Degree-WS}, where an increase in $p_{WS}$ lowers the overlaps between the histogram before and after photon subtraction.

\begin{figure*}
\includegraphics[width=0.99\textwidth]{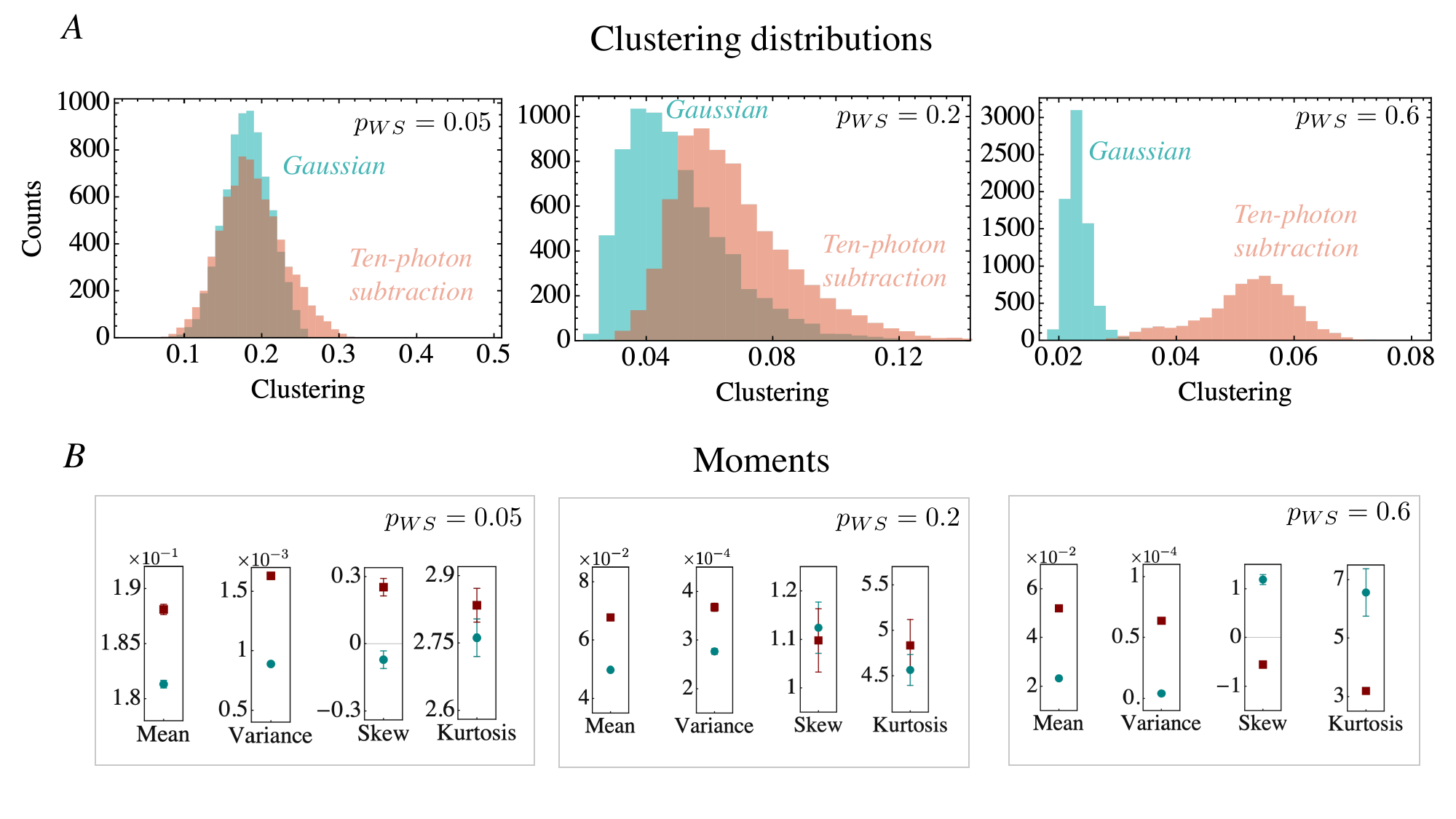}
\caption{Histograms (A) and moments (B) for the clustering distributions for the same networks as Fig.~\ref{PhotonSub-Degree-WS}. Colors indicate degree data prior to (cyan, dots) and after (red, squares) the subtraction of ten photons in the node with the highest connectivity. Photon subtraction shifts and widens the distribution, as shown by the histrograms and quantified by the mean and variance. The higher moments and histrograms indicate that the finer structure of these distributions depend strongly on the value of $p_{WS}$. \label{PhotonSub-Cluster-WS}}
\end{figure*}

In Fig.~\ref{PhotonSub-Cluster-WS}, we explore the role of photon subtraction on the clustering coefficients. The observed difference between different values of $p_{WS}$ is even more profound: the clustering coefficients are only weakly affected by photon subtraction for $p_{WS}=0.05$, whereas for $p_{WS}=0.6$ the histogram changes dramatically. These drastic changes are also seen when comparing the moments before and after photon subtraction in Fig.~\ref{PhotonSub-Cluster-WS} B, where photon subtraction increases the skewness and kurtosis for $p_{WS}=0.05$, but strongly decreases these moments for $p_{WS}=0.6$. Nevertheless, even though the clustering coefficients are not strongly affected by photon subtraction in imprinted WS structures with $p_{WS} = 0.05$, these clustering coefficients remain much higher than those of the imprinted networks with higher values of $p_{WS}$ (which one can also confirm in the moments).

These observations coincide with the intuition that photon subtraction generally increases the correlations in our system. However, it remains to understand which features of the network structure associated with the different values of $p_{WS}$ determine the extent of the effect of photon subtraction. \\


In Fig.~\ref{ScalingWS} we compare the effect of photon subtraction on WS networks with $k = 5$ and $p_{WS} = 0.05$ for different network sizes. In this particular case, we compare the histograms for the degrees and clustering coefficients obtained from 100 networks of 100 nodes to those for 10 networks of 1000 nodes. To limit computational times, the node for photon subtraction was chosen randomly (rather than the one with the highest connectivity). As can be deduced from Fig.~\ref{WSvicinity}, the fraction of nodes in the network that is in the vicinity of the photon-subtracted node decreases with the size of the network. This leads to a sort of dilution of the effect of photon subtraction when we study global properties of the emergent correlation networks. This is clearly seen, both in the degree distribution and in the distribution of clustering coefficients. 

The behaviour in Fig.~\ref{ScalingWS} is a manifestation of the limited size of the vicinity of the photon-subtracted node as compared to the global network. From Fig.~\ref{WSvicinity}, we see that this behaviour can be be observed in all considered cases. As long as we subtract all photons in the same mode, there is not much to be gained from increasing the system size beyond 100 nodes. Instead, it is much more interesting to try and understand the local properties of the emergent networks, i.e., what is happening in the vicinity of the photon-subtracted node. This analysis will be presented in Section \ref{sec:AllDistance}. First, however, we will explore how the characteristics of emergent correlation networks are globally affected for BA imprinted structures, where the imprinted networks have a very different degree distribution.

\begin{figure}
\includegraphics[width=0.49\textwidth]{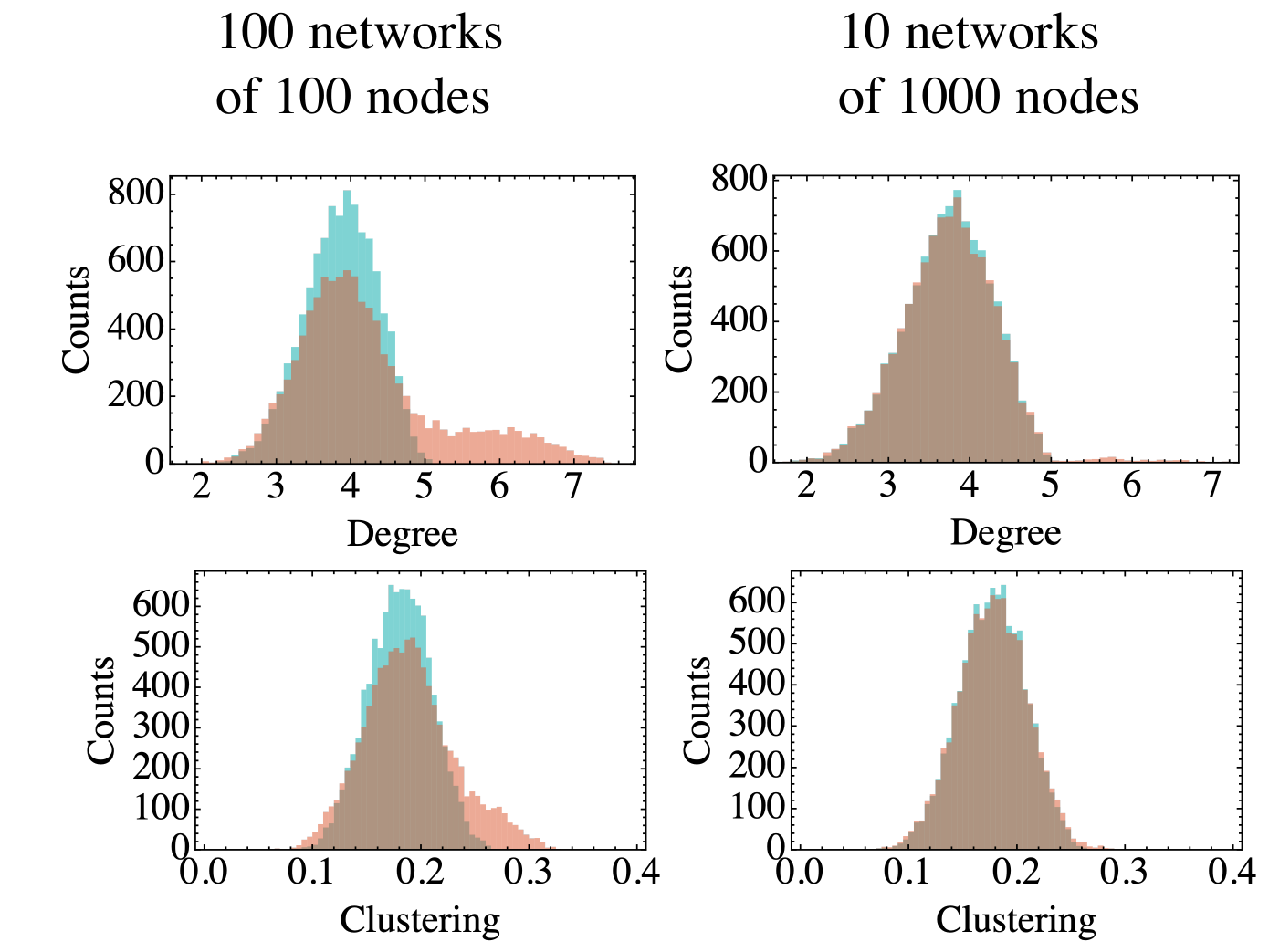}
\caption{Histograms of the degrees (top) and clustering coefficients (bottom) for the emergent correlation networks, after subtraction of ten photons, for 100 WS networks of 100 nodes (left) compared to 10 WS networks of 1000 nodes (right). The WS imprinted networks are obtained from a one-dimensional regular network where each node is connected to its $k=5$ nearest neighbors with a rewiring probability $p_{WS} = 0.05$. For each realisation, the ten photons are subtracted in one randomly chosen node.  \label{ScalingWS}}
\end{figure}

\subsection{Photon subtraction in Barab\'{a}si-Albert networks --Difference between random and highly connected subtraction node } 
\label{ssec:ba}

We now explore how photon subtraction affects the emergent network of a BA imprinted structure, both when we subtract always from the most important hub (i.e., the node with the highest connectivity in the imprinted network), and when we subtract in a randomly chosen node (likely a node with low connectivity).

\begin{figure}
\includegraphics[width=0.49\textwidth]{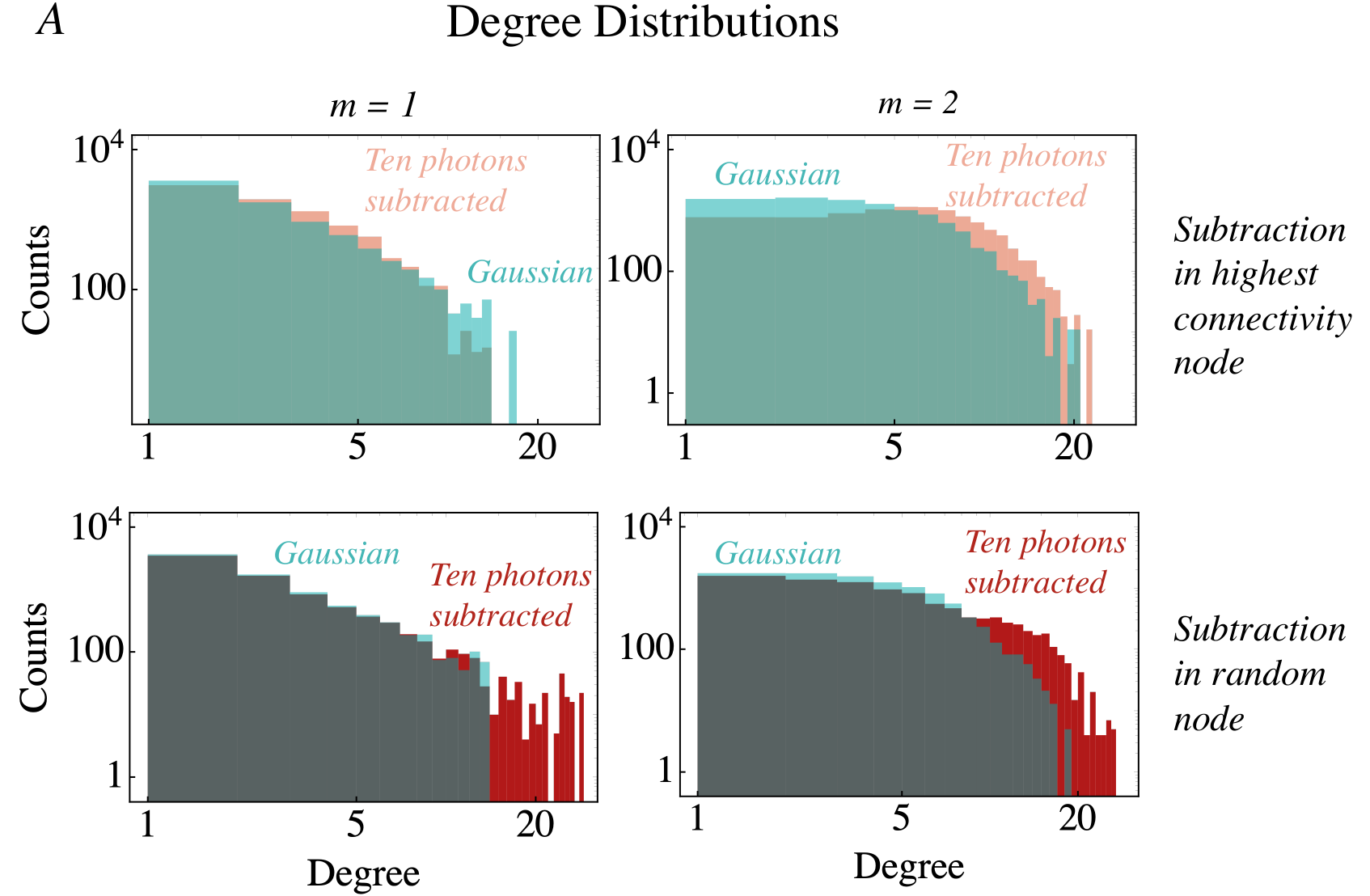}
\includegraphics[width=0.49\textwidth]{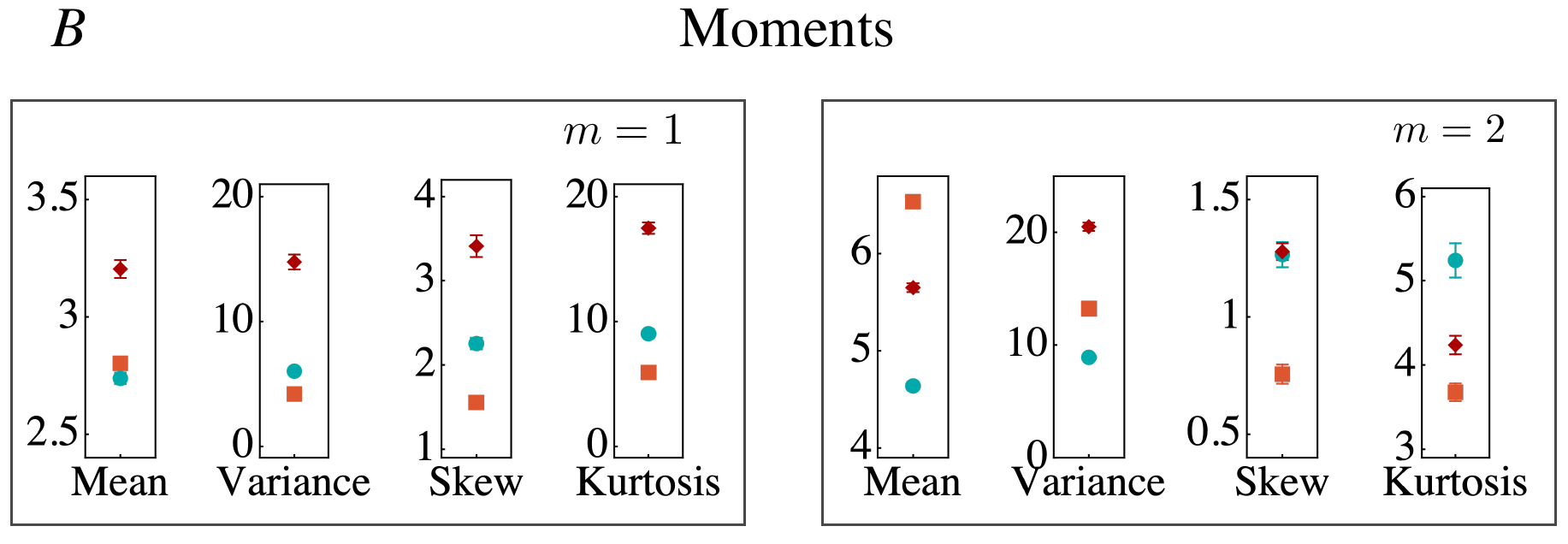}
\caption{Logarithmically-scaled histograms (A) and moments (B) for the degree distributions of the emergent correlation network obtained from a BA imprinted structure for networks generated with $m=1,2$. Colors indicate degree data prior to (cyan, dots) and after the subtraction of ten photons in the node with the highest connectivity (light red squares), or in a randomly chosen (dark red diamonds) node. These histograms were each obtained by combining $100$ random realizations of a $100$-node network. Photon subtraction mainly affects the tails of the distribution as seen in the histograms and reflected in the variance and kurtosis. The emergent networks for $m=1$ imprinted structures show power-law behaviour, which is reflected by high values of the kurtosis.\label{PhotonSub-Degree-BA}}
\end{figure}

Even before photon subtraction the moment analysis in Fig.~\ref{PhotonSub-Degree-BA}(B) shows that for imprinted BA structures the degree distributions of emergent correlation networks have a large variance and kurtosis, in particular for $m=1$. Hence, the emergent networks inherit some of the power-law features of the imprinted structures. In the top panels of Fig.~\ref{PhotonSub-Degree-BA} we therefore show the degree distribution on a log-log scale, for $m=1$ and $m=2$, before and after subtraction of ten photons. 

For $m=1$ imprinted structures, the effect of photon subtraction manifests within the tail of the distribution. We observe the power-law behaviour that is suggested by the moments, and we find that photon subtraction in a hub tends to reduce the weight in the tail. Thus, photon subtraction in the most important hub has a reasonably small effect on a large fraction of the network to make the degrees somewhat more homogeneous. In contrast, photon subtraction increases the weight in the tail if it occurs in a randomly picked node. This shows that, when the photons are subtracted in a node that is correlated to only a small number of other nodes, it can very significantly increase these correlations, thus causing larger values to appear in the tails. This behaviour is consistent with photon subtraction as a finite resource for entanglement distillation. Yet, it must me stressed that photon-number correlation are not necessarily quantum correlations. For nodes with a high connectivity, photon subtraction only weakly alters the individual correlations. As a final comment for the $m=1$ case, we must note that the bulk of the distribution remains largely unaffected, up to a point where the effect of photon subtraction is hardly visible when the histogram is plotted on a linear scale -- this is also reflected by a relatively small change in the mean degree.

For $m=2$ imprinted structures, the distribution does not show typical power-law behaviour, which is reflected in smaller values of kurtosis in Fig.~\ref{PhotonSub-Degree-BA}(B). These moments, nevertheless, show a profound change in the variance due to photon subtraction, which implies an overall widening of the distribution. Figure~\ref{PhotonSub-Degree-BA}(A) shows this feature, as now a larger fraction of the distribution grows to higher values of the degree. Hence, for $m=2$ we can conclude that photon subtraction predominantly affects the bulk of the distribution, which is qualitatively similar to what we saw for WS distributions.

\begin{figure}
\includegraphics[width=0.49\textwidth]{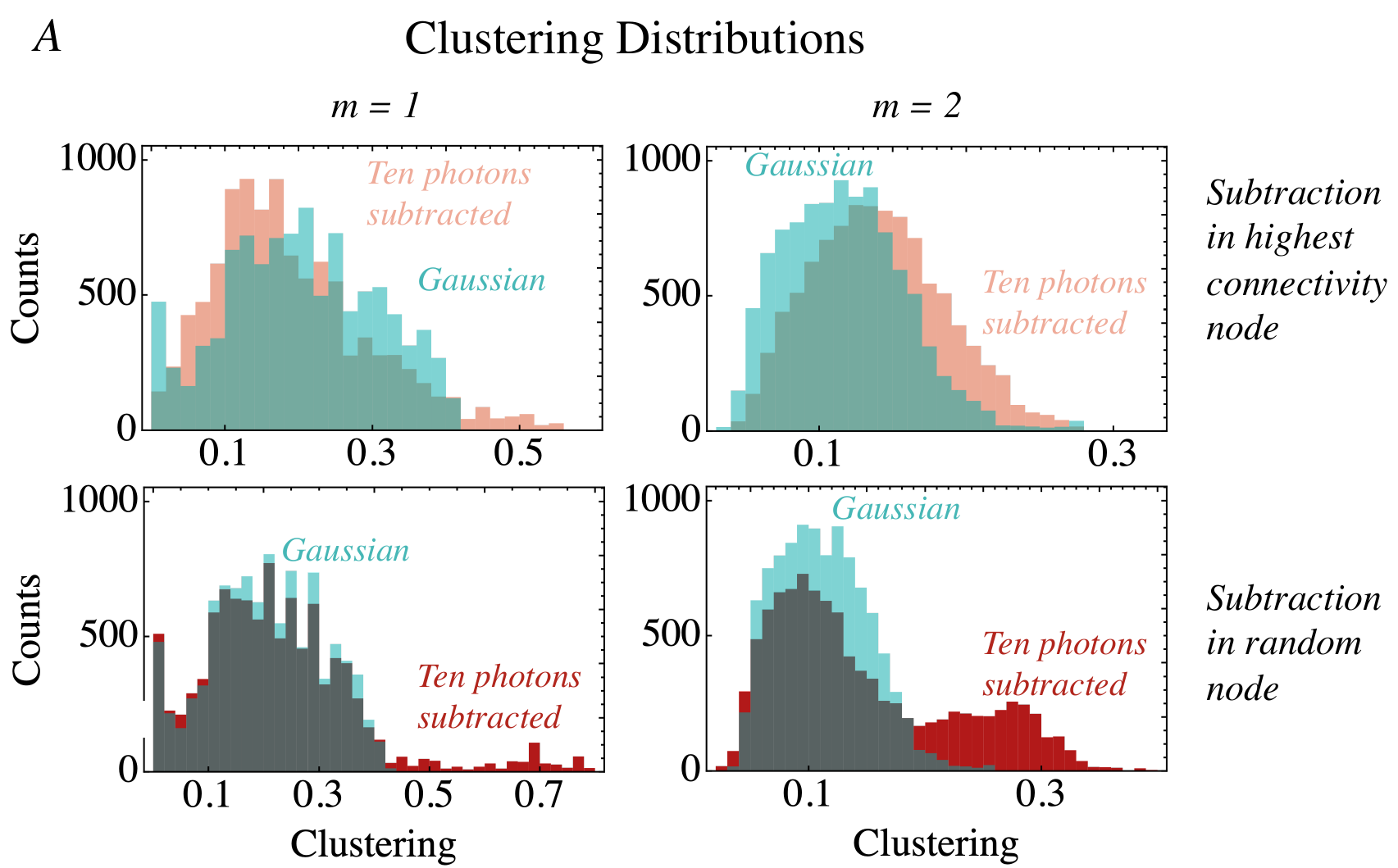}
\includegraphics[width=0.49\textwidth]{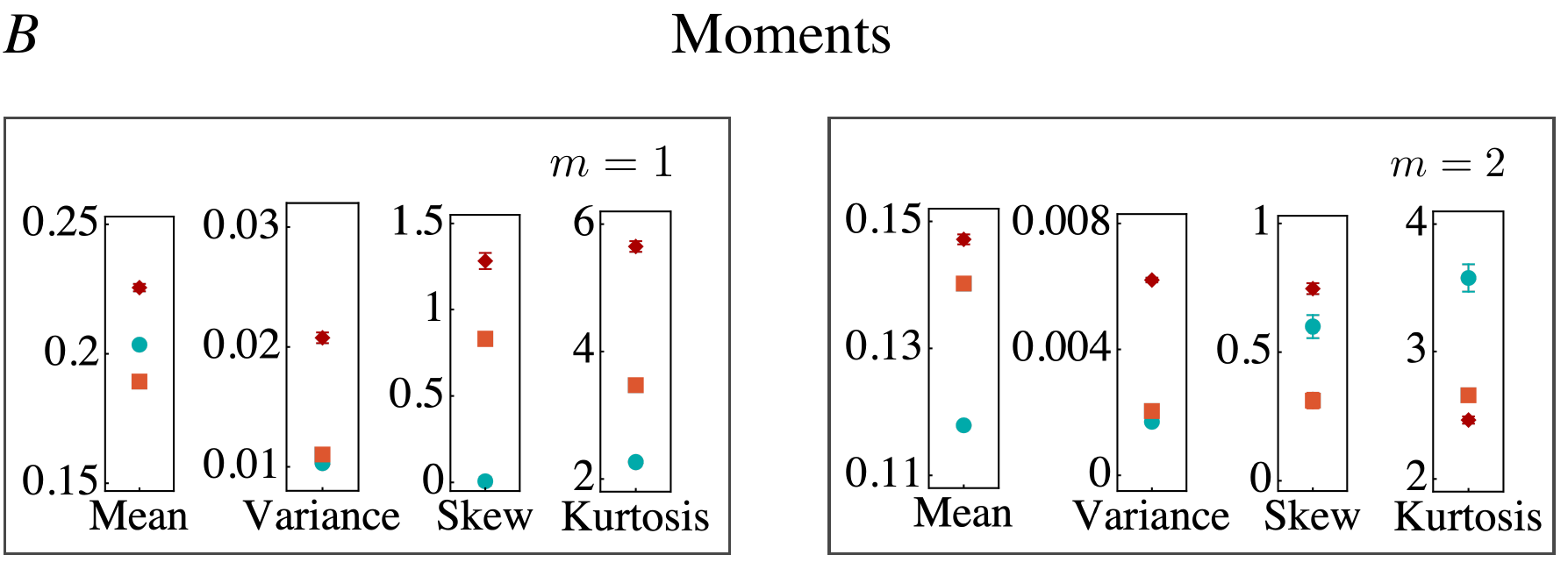}
\caption{Histograms (A) and moments (B) for the clustering distributions of the same networks as Fig.~\ref{PhotonSub-Degree-BA}, and the same color coding. Photon subtractions increase the tails of the distributions. Photon subtraction in a random node can create high clustering coefficients for a reasonably small number of modes.\label{PhotonSub-Cluster-BA}}
\end{figure}

In Fig.~\ref{PhotonSub-Cluster-BA} we observe that for $m=1$ the clustering coefficients in these networks can be increased up to $Cl=0.8$, though only for a small fraction of nodes. In other words, photon subtraction, again, predominantly affects the tails of the distribution for $m=1$ (which is confirmed by the moment analysis in Panel B of Fig.~\ref{PhotonSub-Cluster-BA}). 
Therefore  
random tree networks (i.e., BA with $m=1$) globally seem to be the most resilient networks to local photon subtraction operations, even though photon subtraction in nodes with few links can cause profound local changes in the correlations. For $m=2$ we again see a larger overall impact of photon subtraction, leading to more significant changes to the bulk of the distribution. This, too, is in line with the degree statistics.

These results suggest that the environment of the subtracted node in the imprinted network plays an important role in how the emergent network reacts to photon subtraction. To unravel this interplay between the imprinted structure and the emergent network, we will investigate the behaviour of nodes depending on their distance (in the imprinted network) to the node of photon subtraction.

Note that the power-law degree distribution of the imprinted networks can lead to very large vicinities of the photon-subtracted nodes (in particular when we subtract photons in a hub). This means that the number of non-Gaussian correlations in the networks can rapidly grow, as shown in Fig.~\ref{BAvicinity}. In our present implementation of the code to simulate the correlation networks \cite{Code}, this makes BA networks beyond 100 nodes numerically very challenging to treat.

\begin{figure*}[t]
\includegraphics[width=0.95\textwidth]{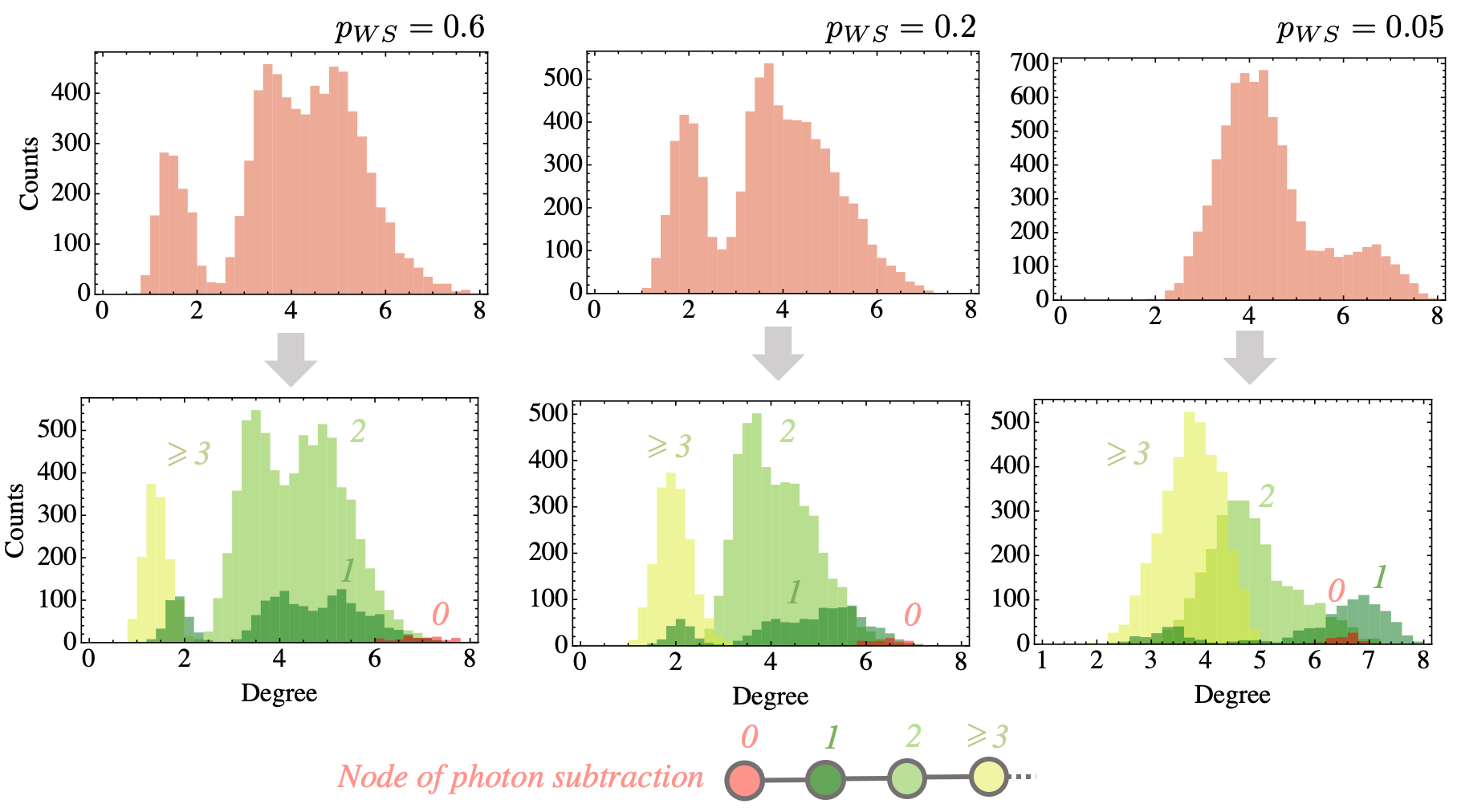}
\caption{Degree distribution of the emergent network of imprinted WS networks after the subtraction of ten photons, using the same data as Fig.~\ref{PhotonSub-Degree-WS}. Complete histograms (top) were each obtained by combining all the nodes of $74$ random realizations of a $100$-node network (see also Fig.~\ref{PhotonSub-Degree-WS}). Distance-resolved histograms (bottom) are obtained by grouping nodes based on their network distance (in the imprinted network) to the node of photon subtraction. Network distance is indicated by color code and labeled by a number (zero being the node of subtraction). Photon subtraction shifts the degree distribution to higher values for the nodes of photon subtraction (distance \textit{0}) and those at distance \textit{2}. At distance \textit{3} and beyond, the effects are negligible. At distance \textit{1} the distribution is affected in a non-trivial way depending on $p_{WS}$. The value of $p_{WS}$ also influences the relative importance of distance-induced features, e.g., for $p_{WS}=0.05$ we find a larger fraction of nodes at distance \textit{3} or beyond. \label{Distance-Degree-WS}}
\end{figure*}

\section{Imprinted structure guides non-Gaussian effects} \label{sec:AllDistance}

The results in Section \ref{sec:local} and the more general theorem presented in Appendix \ref{App:Analitics} show that the study of the non-Gaussian correlations induced by photon subtraction is actually a study of sub-networks rather than the study of the global state, as it was shown in Fig.~\ref{ScalingWS}. In this section we go beyond the simple separation of affected (i.e. the vicinity of the photon-subtracted node) and unaffected nodes, and explore more detailed sub-structures of the networks.\\  
In Subsection \ref{sec:Distance}, we first explore how distances for the photon-subtracted node in the imprinted network have an effect on the emergent correlations. This will notably highlight a different behaviour for nearest and next-to-nearest neighbours. In Subsection \ref{sec:nearestneighbor}, we will then explore in detail how the structure nearest-neighbour sub-networks of the imprinted networks have a profound influence of the non-Gaussian effects that manifest in the photon-number correlations.

\subsection{Distance-induced structure}\label{sec:Distance}
In Sec.~\ref{sec:NonGaussianNetworks}, we showed that photon subtraction induces additional structure in the emergent network. 
Here, we take the first step toward understanding how the emergent structure in photon-number correlations is influenced by the imprinted structure. We break up the statistics according to the imprinted distance between the node in which the photons were subtracted and the nodes under consideration. This distance between nodes is here understood to be the number of connections in the shortest path that connects the nodes in the imprinted structure. \\ 
In Sec.~\ref{sec:local} we emphasised that the quantity $\<\hat n_i \hat n_j\> - \<\hat n_i \>\< \hat n_j\>$ is only altered by photon subtraction when nodes $i$ and $j$ are both in the vicinity of the point of photon subtraction. When at least one of the vertices lies beyond, the features of its emergent correlations are only impacted via the denominator in Eq. \eqref{Cij}. For the degree statistics, this means that nodes at distances zero (point of subtraction), one (nearest neighbours), and two (next-to-nearest neighbours) are very differently affected by photon subtraction than the remaining nodes. This motivates the choice to separate the nodes into four groups: the nodes where the photons are subtracted (distance \textit{0}); their nearest neighbors (distance \textit{1}); the next-nearest neighbors (distance \textit{2}); and all the remaining nodes (distance \textit{3} or more).

\begin{figure}
\includegraphics[width=0.48\textwidth]{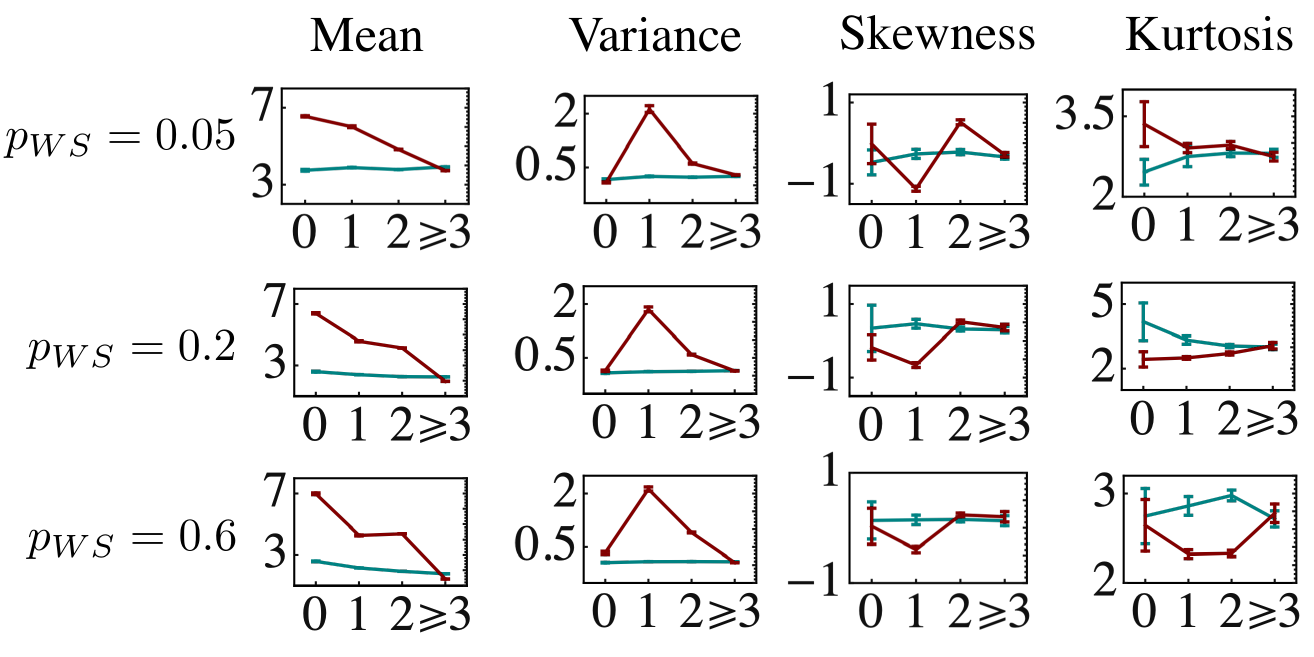}
\caption{Distance-resolved moments for the degree statistics in the emergent photon-number correlation network, resulting from imprinted WS networks with rewiring probabilities $p_{WS}=0.05$, $p_{WS}=0.2$, and $p_{WS}=0.6$. Gaussian states (cyan) and ten-photon subtracted states (red). Photon subtraction mean and variance are affected in the same way for all networks, showing that photon subtraction has the global tendency of increasing the degree and widening the distribution of nodes up to distance \textit{2}. The effect on higher moments depends on the value of $p_{WS}$, showing that photon subtraction also affects the fine structure of the degree distribution in a more subtle way that depends on the network topology. At distance \textit{3} and beyond we see no effect. \label{Distance-Moments-WS}}
\end{figure}

\begin{figure}
\includegraphics[width=0.48\textwidth]{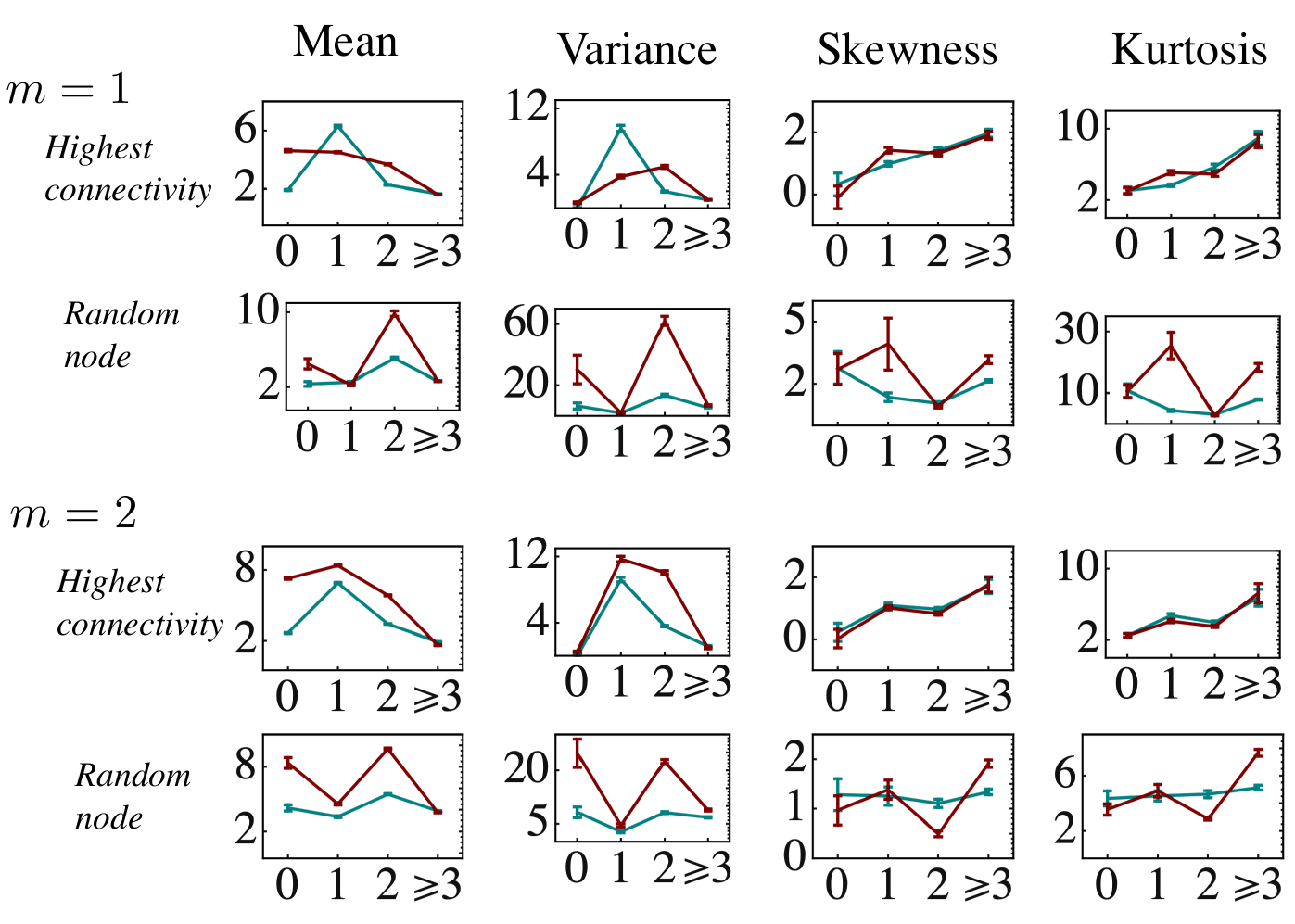}
\caption{Distance-resolved moments for the degree statistics in the emergent photon-number correlation network, resulting from imprinted BA networks with parameters $m=1$ and $m=2$. Gaussian states (cyan) and ten-photon subtracted states (red). The same observations hold as for the WS networks in Fig.~\ref{Distance-Moments-WS}, except for the distance \textit{1} nodes in imprinted structures with $m=1$. The latter is explained in detail in Sec.~\ref{sec:nearestneighbor}.  \label{Distance-Moments-BA}}
\end{figure}

As an example of such a distance analysis, in Fig.~\ref{Distance-Degree-WS} we show four histograms corresponding to our four chosen groups of nodes. A complementary quantitative view can be obtained by studying the moments of these distance-resolved histograms, as shown in the moments of the degree distribution in Fig.~\ref{Distance-Moments-WS} for imprinted WS structures and in Fig.~\ref{Distance-Moments-BA} for imprinted BA network. A completely analogous analysis can be carried out for the clustering coefficients.

The moments in Figs.~\ref{Distance-Moments-WS} and~\ref{Distance-Moments-BA} provide a range of important insights. First, we find that degree distribution of nodes that lie beyond the next-nearest neighbors ($\geqslant 3$) are generally unaltered by photon subtraction. A notable exception is found for the imprinted BA network with photon subtraction in a random node, where the higher moments, i.e., skewness and kurtosis, for these nodes are influenced. This is consistent with the idea that, for an imprinted BA network with photon subtraction in a random node, the non-Gaussian effects are confined to a smaller number of nodes, which in turn change more drastically. 

As a second observation, we find that the distance-dependent effects in the skewness and kurtosis depend strongly on the specific network-type and chosen parameters, in contrast to the mean and variance. This implies that photon subtraction induces some general effects on the bulk of degree distributions (as comprised by the first two moments), while the effect on the finer structure (as comprised by the higher moments) of the degree distributions depends more strongly on the precise topology of the imprinted networks. 

As an important general effect, we find that both for the nodes in which photons are subtracted (\textit{0}) and their next-nearest neighbors (\textit{2}) the mean and variance of the degree distribution always increase. The behaviour of the nearest neighbors (\textit{1}) is less systematic. For imprinted BA networks with photon subtraction in a random vertex, the mean and variance are essentially unaltered for the nearest neighbors. In contrast, for imprinted WS networks, and the $m=2$ BA network with photon subtraction in the node with highest connectivity, the mean and variance increase for these nodes. For the imprinted BA network with $m=1$ and photon subtraction in the node with highest connectivity, we find that the mean and variance decrease after photon subtraction. Hence, there must be other features in the topology of the imprinted network that influence the degree distribution of the nearest neighbors. These features will be laid out in the following subsection. 

\subsection{Nearest-neighbors (\textit{1}) subnetworks}\label{sec:nearestneighbor}

\begin{figure*}
\includegraphics[width=0.95\textwidth]{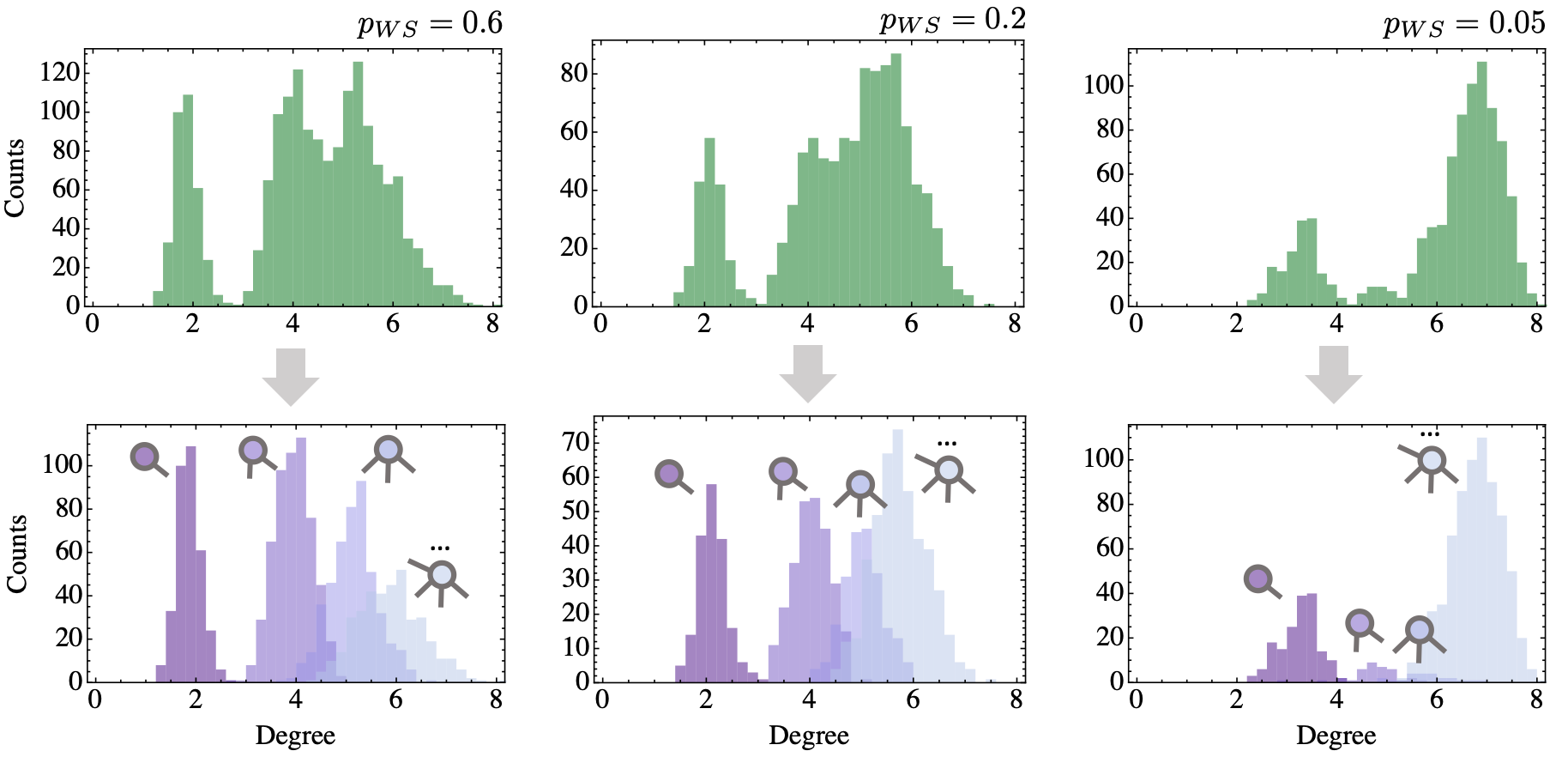}
\caption{Nearest-neighbor $(1)$ degree distribution of the emergent network after the subtraction of ten photons, using the same data as Fig.~\ref{Distance-Degree-WS}. Top row: complete nearest-neighbor histograms. Bottom row: histograms obtained by grouping nearest-neighbor nodes based on the number of other nearest-neighbor nodes they are connected to in the nearest-neighbour subnetwork. The connectivity in the nearest-neighbor sub-network is highlighted by the nodes represented next to the histogram. Nearest neighbors of the node of photon subtraction are more strongly affected by the non-Gaussian operation when they are connected to other nearest neighbors. Nodes that are not connected to any other nearest neighbors (darkest purple) are shifted to lower degrees as compared to the Gaussian distributions in Fig.~\ref{PhotonSub-Degree-WS}.
\label{Distance-One-Degree-WS}}
\end{figure*}

In Fig.~\ref{Distance-One-Degree-WS} we show that the effect on the degree of a nearest-neighbor node in the emergent correlation network is influenced by the number of other nearest neighbors it is connected to in the imprinted networks. This highlights the importance of the topology of the distance-\textit{1} sub-network, as compared to the total imprinted network. 
Quantitatively, this connectivity can be obtained by analyzing the nearest-neighbor sub-network, as highlighted in Fig.~\ref{Distance-One-Example-WS}. When we analyze all the nearest-neighbor sub-networks of our simulated WS networks, we obtain the result in Fig.~\ref{Distance-One-Degree-WS}. The bottom row of figures clearly shows that the degrees (in the emergent network) of nearest neighbors are more strongly affected by photon subtraction when these nodes have a \textit{higher number of connections to other nearest neighbors}. Thus, the different shapes of the nearest neighbor distributions (\textit{1}), for different values of the rewiring probability $p_{WS}$, can be fully understood from the nearest neighbor sub-network in the imprinted structure.

We note that photon subtraction shifts the histograms which group nearest-neighbor nodes according to their connectivity in the distance-\textit{1} sub-network to higher mean values for higher connectivity. 
However, for nodes that are not connected to other nearest neighbors, we witness a slight decrease in the average degree due to photon subtraction.

\begin{figure}
\includegraphics[width=0.39\textwidth]{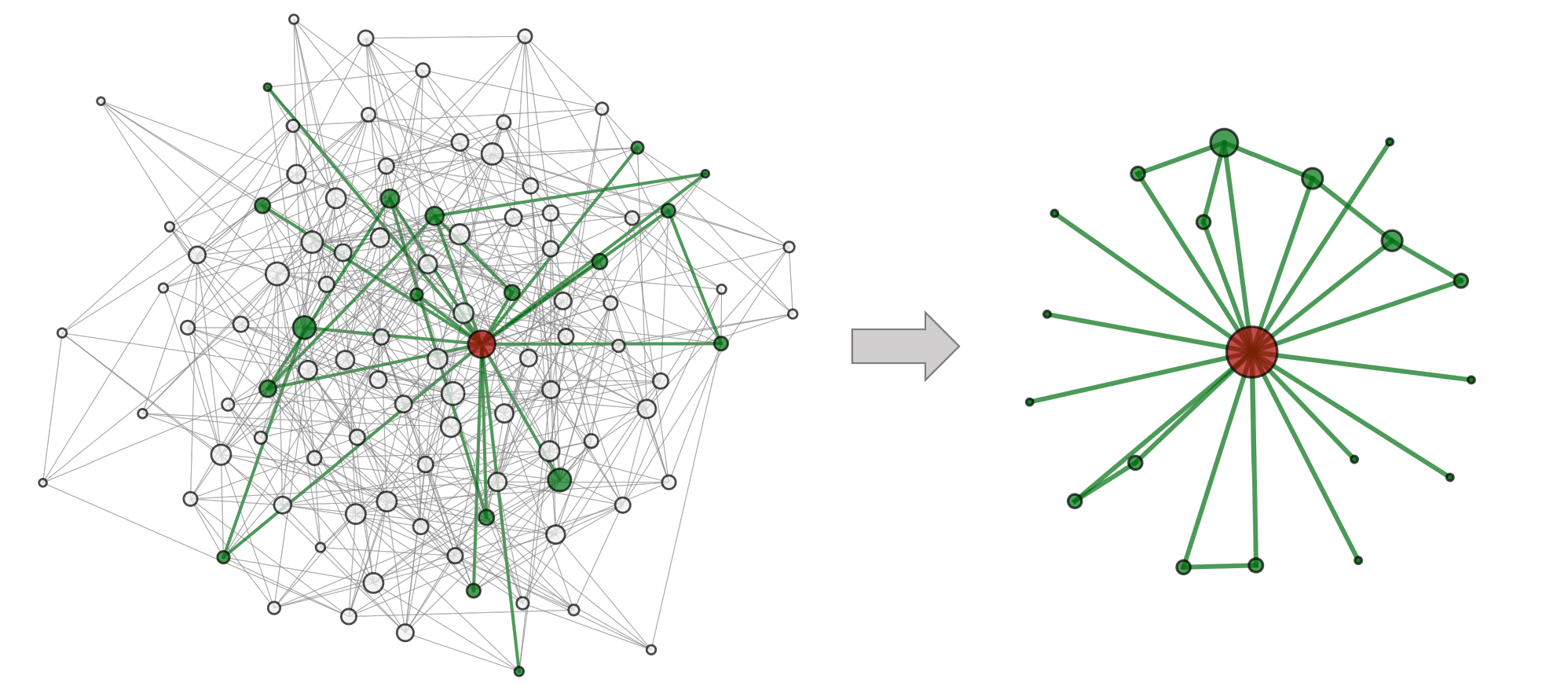}
\includegraphics[width=0.39\textwidth]{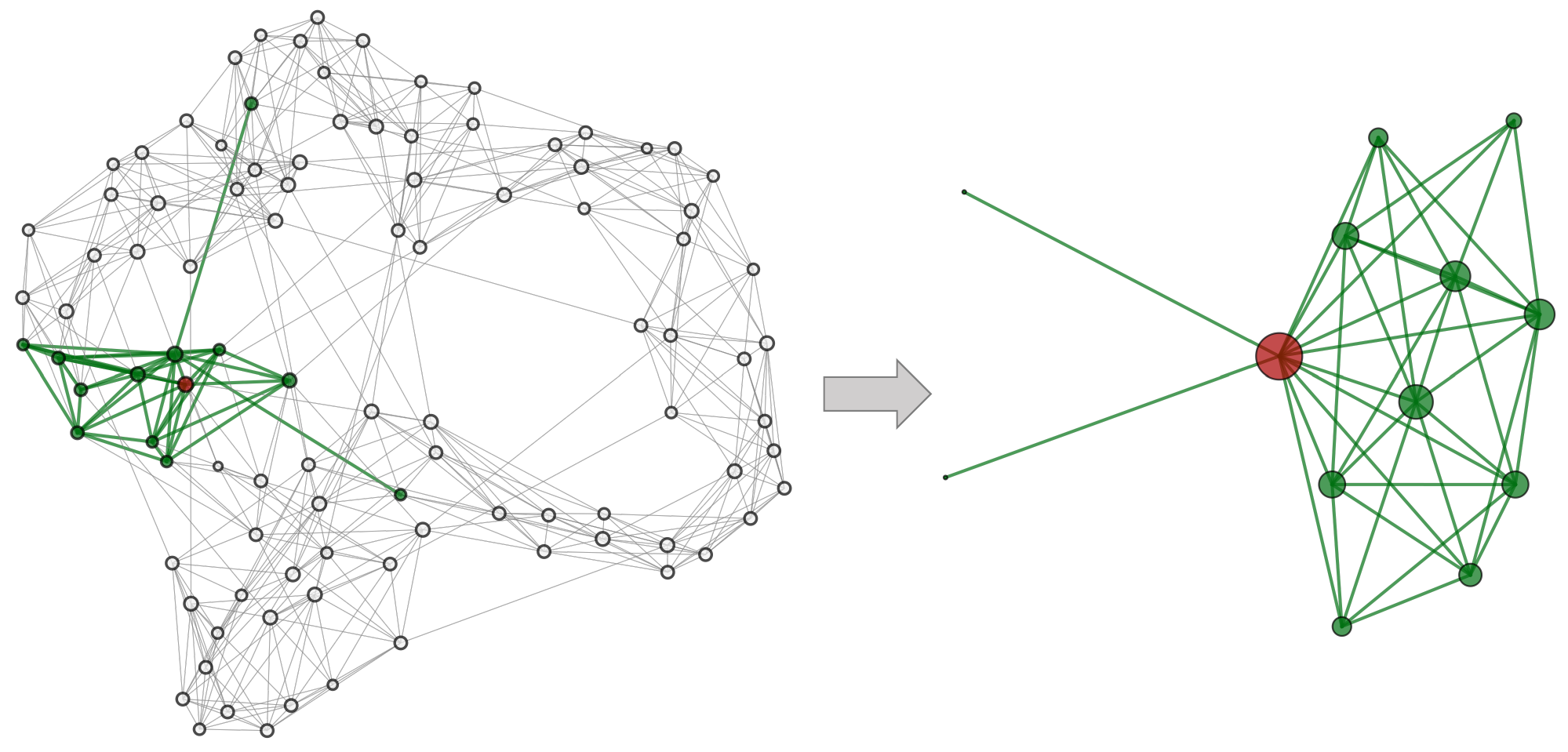}
\caption{WS networks obtained from rewiring a regular 1D network, where every node is connected to its $k=5$ nearest neighbors, with rewiring probability $p_{WS} = 0.6$ (top) and $p_{WS} = 0.05$ (bottom). The subtraction node is highlighted in red, whereas the nearest neighbor sub-network is shown in green. We thus illustrate that the value of $p_{WS}$ strongly influences the typical structure of the nearest-neighbor network, which in turn strongly influences the relative importance of the different histograms in the bottom row of Fig.~\ref{Distance-One-Degree-WS}. \label{Distance-One-Example-WS}}
\end{figure}

For $p_{WS}=0.6$ in Fig.~\ref{Distance-One-Example-WS}, we observe a significant fraction of nearest neighbors that are not connected to any other nearest neighbors. This provides a sharp contrast with the case for $p_{WS}=0.05$, where we observe that networks tend to form clusters, such that nearest neighbors are more likely to be connected to each other. From Fig.~\ref{Distance-One-Degree-WS}, we can understand how these features of the imprinted structure have a direct effect on the statistical features of the emergent network.

Similar analyses have been carried out for the nearest neighbor networks of all simulated classes of networks, leading to the same results. A particularly striking case is the BA network with $m=1$: because these networks are tree-like, nearest neighbors are never connected to one another. This explains why in Fig.~\ref{Distance-Moments-BA}, for $m=1$, the means for nearest neighbors (\textit{1}) are reduced by photon subtraction.
Moreover, we can now explain why in Fig.~\ref{Distance-Moments-BA} the contribution of the nearest-neighbor sub-network for photon subtraction from the highest degree node for BA with $m=2$ is more important than in the case of subtraction from a random node. In this latter case it is more likely to select one of the isolated nodes with a surrounding nearest-neighbor sub-network also characterized by low connectivity. However, the analysis is not sufficient to explain why the next-to-nearest neighbors (\textit{2}) are so strongly affected by photon subtraction from a random node in the BA network with $m=2$. 

\section{Summary, Discussion, and Outlook}\label{sec:Discussion}

We have addressed the question of how emergent complexity arises from imprinted complexity in continuous variable (CV) quantum networks.  In particular, we have studied  how localized non-Gaussian features, a key feature of CV quantum technologies and quantum networks such as a future quantum internet, spread through emergent network correlations.  Specifically, we used tools from complex network theory to analyze emergent networks of photon-number correlations that manifest in CV cluster states. These cluster states were initially taken as Gaussian states, constructed by applying an imprinted network of entangling $C_Z$ gates to a series of squeezed vacuum modes; non-Gaussian states were then created via photon subtraction. We focused on the particular case in which these imprinted networks of quantum gates were chosen to have a complex network structure, implementing either a Watts-Strogatz (WS) or Barab\'{a}si-Albert (BA) model.  

For Gaussian states created by imprinting WS networks, we found that increased probability of rewiring (and thus more randomness) in the imprinted networks decreases the typical degree and clustering coefficient in the emergent correlation network. When BA structures were imprinted, we found that emergent correlation networks inherit heavy tails with a structure that is strongly influenced by the number $m$ of connections added with every node in the BA preferential attachment process. For tree networks ($m=1$), the emergent networks were found to inherit a power-law tail in its degree statistics. In contrast, this power-law behavior was not observed for $m=2$.

We rendered the states non-Gaussian by subtracting ten photons in a specific node of the system. The short distance between the different nodes in the network guaranteed that the effect of photon subtraction spread far throughout the state. We showed that network theory methods are well suited to characterize the resulting non-Gaussian states, as photon subtraction was found to profoundly change the structure of the emergent networks (see Fig.~\ref{networks} for an example). We quantified this effect by comparing the distributions of the degrees and clustering coefficients before and after photon subtraction. Generally, we found that photon subtraction increased the typical degree and clustering coefficient, as seen from the mean, and increased the spread of these quantities as characterized by a growing variance. However, we also observed that the higher moments, and thus the finer structure in the tails of the distributions, depended more strongly on the network types and chosen parameters. For BA imprinted networks, we found a strong variation in behaviour of the emergent correlation network depending on whether the photons are subtracted from a hub or a random node. This highlights the importance of the local network topology in the vicinity of the node of photon-subtraction.

We proceeded to dissect the structure of the emergent correlation networks, based on the features of the imprinted network. As a first step, in Sec.~\ref{sec:Distance}, we filtered nodes in the emergent network based on their distance to the node of photon subtraction in the imprinted network. Generally, we found  that photon subtraction increases the variance and mean in the nodes of photon subtraction (distance \textit{0}), as well as for next-nearest neighbors (distance \textit{2}), while the first two moments of nodes that were further removed (distance $\geqslant 3$) remained unaffected. However, the nearest neighbors (distance \textit{1}) in imprinted tree networks defied this global trend. In Sec.~\ref{sec:nearestneighbor} we showed this is due to the particular structure of the nearest-neighbor sub-networks that are extracted from the imprinted network. When we further filtered the nodes at distance \textit{1} based on their connectivity to other nodes at distance \textit{1}, we found that higher connectivity led to a stronger increase in the typical degree due to photon subtraction. In other words, nearest-neighbor nodes were more strongly affected by photon subtraction if they were connected to many other nearest-neighbour nodes. In a tree network, nearest neighbours can never be connected to other nearest neighbors, which slightly decreased the average degree in the correlation network for the nodes at distance \textit{1}. The distance-resolved moment analysis furthermore showed that the fine structure in the tails of the distributions, as captured by higher moments, does depend strongly on the network topology.

The structural features we uncovered help us understand that the bulk features of emergent correlation networks of imprinted tree structures are only slightly impacted by photon subtraction. The effects that do manifest are mainly observed in the tails of the distribution, as seen in the moment analysis. The difference between subtraction in a random vs. a highly connected node is particularly interesting when interpreted in the light of the network structure. Photon subtraction is a finite resource to enhance correlations in the system. On the one hand, in a node that is correlated to only a small number of other nodes it significantly increases correlations between this small set of nodes, thus causing larger values to appear in the tails. On the other hand, when a photon is subtracted in a hub with very high number of nodes correlated to it, the increase in correlations is spread out.\\

This work represents a first test of complex network structures in the CV quantum regime under the applications of necessary non-Gaussian operations. We believe this investigation can be fruitful for making decisions about the structure of future quantum technologies on a large scale. 

We want to stress that the network models we have chosen are a very limited set of the ones that have been developed in complex network theory. The ones we selected have specific features - in particular: variable level of randomness, presence of hubs, power law distribution, and average short path distances - that have never been tested before in non-Gaussian cluster states. Of course other models and features can be treated in future works.

Moreover, complex network methods have proven useful for the theoretical investigation of such classes of states. 
Indeed, the characterization of highly multimode non-Gaussian states is generally an arduous task, where standard tools of CV quantum optics fall short. Typical experimental methods such as homodyne tomography lack the necessary scaling properties to study these systems, and theoretical constructs such as Wigner functions become hard to handle. To overcome this problem, one may look for global properties, e.g., Wigner negativity of the full multimode states~\cite{cimini2020neural}. Such global features have the disadvantage that they gloss over the local or neighborhood structures of the state, which are essential in multi-partite quantum platforms. Our results show that network theory offers effective statistical tools for studying these states. At present, we are unaware of any other method that allows us to describe the physical features that we deduced for these large non-Gaussian states. They offer us a road map for more detailed bottom-up studies of particular features such as the role of connections in the nearest-neighbor sub-network.

\begin{acknowledgments}
 This work was supported  by the European Research Council under the Consolidator Grant COQCOoN (Grant  No. 820079).  
This work was also performed in part with support by the U.S. National Science Foundation under grants
CCF-1839232, OAC-1740130, PHY-1806372, and PHY-1748958; and in conjunction with the QSUM
program, which is supported by the Engineering and Physical Sciences Research
Council grant EP/P01058X/1.  We thank the Complex Quantum Systems group at Laboratoire Kastler Brossel (and in particular D. Delande) for access to their computational facilities.

The software used to generate the complex networks in the Article is available on \cite{Code}
\end{acknowledgments}

\appendix

\section{Simulating cluster states and the evaluation of photon-number correlations}\label{app:corr}

\subsection{Perfect matching}
To execute the simulations presented in this article, we generated random complex networks using the ``python-inetwork'' library. From these randomly generated networks, we extracted the adjacency matrix ${\cal A}$ to generate the cluster state covariance matrices, as described in Eq.~(\ref{eq:V}). After generating the covariance matrix $V$ of a Gaussian network state, we used it to evaluate the photon-number correlations $[\mathbb{C}]_{ij}$ of Eq.~(\ref{Cij}) for the photon-subtracted states in Eq.~(\ref{subrho}).

The main technique used to evaluate these correlations relies on the properties of Gaussian quantum states. We previously used this method to fully characterize single-photon subtracted states in~\cite{WalschaersPRA17}. We illustrate this method by highlighting the evaluation of the element $\<\hat{n}_i \hat{n}_j\>$ in Eq.~(\ref{Cij}). For a photon-subtracted state we find
\begin{equation}\label{eq:NNCalc}
    \<\hat{n}_i \hat{n}_j\> = \frac{\tr[\hat a^{\dag}_{S_n} \dots \hat a^{\dag}_{S_1}\hat a^{\dag}_i \hat a^{\dag}_j\hat a_j \hat a_i\hat a_{S_1}\dots \hat a_{S_n} \rho]}{\tr[\hat a^{\dag}_{S_n} \dots \hat a^{\dag}_{S_1}\hat a_{S_1}\dots \hat a_{S_n} \rho]},
\end{equation}
where $\rho$ denotes the density matrix of the Gaussian network state and $i \neq j$. We then use a general property for Gaussian states, that allows us to express
\begin{align}\label{eq:perfMatch}
&\tr[\hat a^{\dag}_{S_n} \dots \hat a^{\dag}_{S_1}\hat a^{\dag}_i \hat a^{\dag}_j\hat a_j \hat a_i\hat a_{S_1}\dots \hat a_{S_n} \rho]\\
&=\sum_{\cal P} \prod_{\{p_1, p_2\} \in {\cal P}} \tr[\hat a^{\#}_{p_1}\hat a^{\#}_{p_2}\rho],\nonumber
\end{align}
where we introduce the label ${\cal P}$ to denote a ``perfect matching.'' A perfect matching means any way of dividing a set up into pairs, while maintaining the order. When we consider, for example the set $\{1,2,3,4\}$, one possible perfect matching would be $\{\{1,3\},\{2,4\}\}$. In this example, the notation $\{p_1, p_2\} \in {\cal P}$ refers to $\{1,3\}$ and $\{2,4\}$. In Eq.~(\ref{eq:perfMatch}), these perfect matchings are used to split the set of creation and annihilation operators, $\hat a^{\dag}_{S_n} \dots \hat a^{\dag}_{S_1}\hat a^{\dag}_i \hat a^{\dag}_j\hat a_j \hat a_i\hat a_{S_1}\dots \hat a_{S_n}$, in pairs: In total, we have $2n + 4$ creation and annihilation operators, which we can associate with a set of indices $\{1, \dots, 2n+4\}$. The sum over ${\cal P}$ runs over all possible perfect matchings of this index set. For every given perfect matching ${\cal P}$, we then multiply all the quantities $\tr[\hat a^{\#}_{p_1}\hat a^{\#}_{p_2}\rho]$ for the different paired indices $\{p_1, p_2\} \in {\cal P}$. The quantities $p_1$ and $p_2$ are indices in the index set, and the quantity $\hat a^{\#}_{p_j}$ denotes the creation or annihilation operator that occurs at the $p_j$th position in the product $\hat a^{\dag}_{S_n} \dots \hat a^{\dag}_{S_1}\hat a^{\dag}_i \hat a^{\dag}_j\hat a_j \hat a_i\hat a_{S_1}\dots \hat a_{S_n}$. Let us list some examples: $\tr[\hat a^{\#}_{1}\hat a^{\#}_{2}\rho] = \tr[\hat a^{\dag}_{S_1}\hat a^{\dag}_{S_2}\rho]$, $\tr[\hat a^{\#}_{2}\hat a^{\#}_{n+3}\rho] = \tr[\hat a^{\dag}_{S_2}\hat a_j\rho]$, and $\tr[\hat a^{\#}_{1}\hat a^{\#}_{n+5}\rho] = \tr[\hat a^{\dag}_{S_1}\hat a_{S_1}\rho]$.

What remains is now to evaluate the quantities $\tr[\hat a^{\#}_{p_1}\hat a^{\#}_{p_2}\rho]$, and this can be done directly via the covariance matrix $V$, by expressing the creation and annihilation operators in terms of quadrature operators (see also~\cite{WalschaersPRA17,PhysRevA.99.023836} for more details). We find the following identities:
\begin{align}\label{eq:identity}
&\tr[\hat a^{\dag}_j\hat a^{\dag}_k\rho] = \frac{1}{4}[V_{jk} - V_{j+N\,,k+N} -i (V_{j\,,k+N} + V_{j+N\,,k})],\\
&\tr[\hat a_j\hat a_k\rho] = \frac{1}{4}[V_{jk} - V_{j+N\,,k+N} + i (V_{j\,,k+N} + V_{j+N\,,k})],\\
&\tr[\hat a^{\dag}_j\hat a_k\rho] = \frac{1}{4}[V_{jk} + V_{j+N\,,k+N} + i (V_{j\,,k+N} - V_{j+N\,,k})-2\delta_{jk}] .\label{eq:identityfinal}
\end{align}
These identities are expressed in the mode basis that corresponds to the nodes of the network state.

Using Eqs.~\eqref{eq:V} and~\eqref{eq:identity}-\eqref{eq:identityfinal}, we can calculate the weighted adjacency matrix of the emergent network, $\mathbb{A}_{ij}$. In principle, it is possible to calculate $\mathbb{A}_{ij}$ for both the cluster state as well as the photon-subtracted state. However, as we will show, it is exponentially difficult to write a closed-form expression for $\mathbb{A}_{ij}$ in the photon-subtracted state. Below, we will first calculate $\mathbb{A}_{ij}$ in the cluster state. Then we discuss the photon-subtracted case.

\subsection{Gaussian state}\label{App:Gaussian}
Since the cluster state $\rho$ is Gaussian, one can evaluate the connected correlation $ c_{ij} \equiv \trho{ \hat{n}_i\hat{n}_j } - \trho{ \hat{n}_i}\trho{ \hat{n}_j }$ by applying Eq.~\eqref{eq:perfMatch}. Only two terms remain, giving
\begin{equation}\label{eqn: connected correlation}
c_{ij} = \trho{\hat{a}_i^\dagger \hat{a}_j} \trho{\hat{a}_i\hat{a}_j^\dagger} +  \trho{\hat{a}_i^\dagger \hat{a}_j^\dagger} \trho{\hat{a}_i\hat{a}_j}.
\end{equation}
Each term on the right hand side of Eq.~\eqref{eqn: connected correlation} can be evaluated using Eqs.~\eqref{eq:V} and~\eqref{eq:identity}-\eqref{eq:identityfinal}. 

For $i\neq j$, we have
\begin{align}
& \trho{\hat{a}_i^\dagger \hat{a}_j} = s (\mathcal{A}^2)_{ij}/4,\nonumber\\
& \trho{\hat{a}_i \hat{a}_j^\dagger} = s (\mathcal{A}^2)_{ij}/4,\nonumber\\
& \trho{\hat{a}_i^\dagger \hat{a}_j^\dagger} = -(s (\mathcal{A}^2)_{ij} - 2is \mathcal{A}_{ij})/4,\nonumber\\
& \trho{\hat{a}_i \hat{a}_j} = -(s (\mathcal{A}^2)_{ij} + 2is \mathcal{A}_{ij})/4.
\end{align}
Therefore
\begin{equation} \label{eqn: cij}
c_{ij} = \frac{s^2}{8} \left( (\mathcal{A}^2)_{ij}^2 + 2\mathcal{A}_{ij} \right)
\end{equation}
where we used that $(\mathcal{A}_{ij})^2 = \mathcal{A}_{ij}$.

For $i=j$, we have
\begin{align}
& \trho{\hat{a}_i^\dagger \hat{a}_i} = \left( s + 1/s + s{\cal D}_i - 2\right)/4,\nonumber\\
& \trho{\hat{a}_i \hat{a}_i^\dagger} = \left( s + 1/s + s{\cal D}_i + 2\right)/4,\nonumber\\
& \trho{\hat{a}_i^\dagger \hat{a}_i^\dagger} = \left( s - 1/s - s{\cal D}_i\right)/4,\nonumber\\
& \trho{\hat{a}_i \hat{a}_i} = \left( s - 1/s - s{\cal D}_i\right)/4,
\end{align}
where we used that $\mathcal{A}_{ii} = 0$ and $(\mathcal{A}^2)_{ii} = {\cal D}_i$. Therefore
\begin{equation} \label{eqn: ni}
c_{ii} = \frac{1}{8} \left( s^2 + \frac{1}{s^2} + s^2({\cal D}_i)^2 + 2{\cal D}_i - 2\right).
\end{equation}
One obtains Eq.~\eqref{Cij-Gaus} in the main text from Eqs.~\eqref{eqn: cij} and~\eqref{eqn: ni}, where $c_{ii}$ is denoted in the main text as $\mathfrak{N}(s,{\cal D}_i)$.

\subsection{Photon-subtracted states}\label{compSP}

For photon-subtracted states, the correlations quickly become hard to evaluate. At the basis of this complexity lies the appearance of perfect matchings in Eq. \eqref{eq:perfMatch}. Finding all possible perfect matchings is a computationally hard problem that belongs to the complexity class $\# P$. It is also the problem which lies at the basis of the hardness of Gaussian Boson Sampling. Hence, when the number of subtracted photons grows, correlation functions quickly become practically impossible to evaluate.

Generally speaking, the best algorithms for evaluating Eq. \eqref{eq:perfMatch} use recursive techniques. In our work, we greatly simplify this computational problem by subtracting all the photons in the same mode, i.e., $S_1=\dots=S_n=S$. In this case, expression \eqref{eq:NNCalc} for $\<\hat n_i \hat n_j\>$ only contains creation and annihilation operators in three different modes. This greatly limits the different possible factors $\tr[\hat a^{\#}_{p_1}\hat a^{\#}_{p_2}\rho]$ that can appear in Eq. \eqref{eq:perfMatch}. Many different partitions will lead to equivalent contributions. The problem thus reduces to that of identifying all the different classes of partitions, evaluating the contribution, and counting the multiplicity. 

Once we subtract more than three photons, the total number of different classes of terms remains fixed. We evaluated these by hand and counted a total of 43 classes, each appearing with a certain multiplicity. The correlation networks are then calculated by evaluating the contribution by multiplying relevant quantities given by Eqs. \eqref{eq:identity} - \eqref{eq:identityfinal} for each of these 43 classes. Then we multiply each contribution with the right multiplicity, which depends on the number of subtractions and can be calculated through combinatorics. The quantities $\<\hat n_i\>$ are evaluated using the same method. For more details, we refer to the code that was used to carry out the simulations \cite{Code}.

In Appendix \ref{App:Analitics}, we prove analytically that the photon number correlations $\<\hat n_i\hat n_j\> - \<\hat n_i\>\<\hat n_j\>$ only change when nodes $i$ and $j$ are in the vicinity of of the photon-subtracted node (see Appendix \ref{App:Analitics} for details). This implies that we can first generate a full correlation network of Gaussian correlations by relying on the analytical formula (\ref{eqn: cij}) and subsequently we can use (\ref{eq:perfMatch}) to update only the affected correlations. This method is implemented in the second version of our code for simulating WS networks \cite{Code}.

\section{Correlations unaffected by photon subtraction}\label{App:Analitics}

In this section of the appendix, we prove a general result for correlations in $n$-photon subtracted states: photon subtraction can only change the covariance between observables, if both observables are initially correlated to the mode in which the photons are subtracted.

Assume that we subtract $n$ photons from Gaussian state $\rho$ in a mode with label $S$ and associated annihilation operator $\hat a_S$. We denote the algebra of observables ${\cal A}_{\rm near}$ as those observables which are ``near to $S$'' in the sense that ${\cal A}_{\rm near}$ is generated by observables $\hat a_k$ and $\hat a^{\dag}_k$ for which either $\trho{\hat a^{\dag}_k \hat a_S} \neq 0$ or $\trho{\hat a_k \hat a_S} \neq 0$. Following Eqs.~(\ref{eq:identity})-(\ref{eq:identityfinal}) we can equivalently define ${\cal A}_{\rm near}$ as the algebra of observables restricted to the modes with labels $k$ for which the matrix 
\begin{equation}
    \begin{pmatrix}
    V_{Sk} & V_{S\, k+N} \\ V_{S+N \, k} & V_{S+N \, k+N}
    \end{pmatrix} \neq \begin{pmatrix}
    0 & 0\\
    0 & 0 
    \end{pmatrix}
\end{equation}
We then define ${\cal A}_{\rm far}$ as the complement of ${\cal A}_{\rm near}$ in the sense that ${\cal A}_{\rm near} \otimes {\cal A}_{\rm far} = {\cal A}$, where ${\cal A}$ is the full algebra of observables on the $N$-mode Fock space that describes the entire system.

\begin{theorem}\label{Theorem1}
For any observable $\hat X \in {\cal A}_{\rm near}$ and another arbitrary observable $\hat Y \in {\cal A}_{\rm far}$, it holds that
\begin{equation}
    \<\hat X \hat Y\> - \<\hat X \>\< \hat Y\> = \tr[\hat X \hat Y \rho] - \tr[\hat X \rho]\tr[\hat Y \rho],
\end{equation}
where $\<\dots\>$ denotes the expectation value in the $n$-photon subtracted state and $\tr[\dots \rho]$ is the expectation value in the initial Gaussian state.
\end{theorem}

\begin{proof}
We first of all use that every observable $\hat X \in {\cal A}_{\rm near}$ can be arbitrarily well approximated by a polynomial in creation and annihilation operators in ${\cal A}_{\rm near}$. This implies that
\begin{equation}
\hat X = \sum_{j=0}^{\infty} \sum_{k} c_{j,k} \hat a^{\#}_{1,k} \dots \hat a^{\#}_{j,k},
\end{equation}
where $ \hat a^{\#}_{1,k}$ can either be a creation or an annihilation operator. The sum over $k$ takes into account that there are many possible products of creation and annihilation operators, also known as Wick monomials, of length $j$.   And similarly for $\hat Y$ we find
\begin{equation}
\hat Y = \sum_{j=0}^{\infty} \sum_{k} c_{j,k} \hat b^{\#}_{1,k} \dots \hat b^{\#}_{j,k},
\end{equation}
To highlight that the creation and annihilation operators that build $\hat X$ and $\hat Y$ have different supports, we have noted the the creation and annihilation operators for ${\cal A}_{\rm far}$ as $\{\hat b^{\#}_{j,k}\}$.

The fact that any observable can be written as a series expansion of creation and annihilation operators implies that our theorem can be proven by proving that 
\begin{equation}
\begin{split}\label{toprove}
   & \< \hat a^{\#}_{1} \dots \hat a^{\#}_{j}  \hat b^{\#}_{1} \dots \hat b^{\#}_{j'}\> - \<\hat a^{\#}_{1} \dots \hat a^{\#}_{j}  \>\<  \hat b^{\#}_{1} \dots \hat b^{\#}_{j'}\> \\
    &= \tr[\hat a^{\#}_{1} \dots \hat a^{\#}_{j}  \hat b^{\#}_{1} \dots \hat b^{\#}_{j'} \rho] - \tr[\hat a^{\#}_{1} \dots \hat a^{\#}_{j}  \rho]\tr[ \hat b^{\#}_{1} \dots \hat b^{\#}_{j'} \rho],
\end{split}
\end{equation}
First of all, let us consider the term
\begin{equation}\label{eq:<Y>}
\<  \hat b^{\#}_{1} \dots \hat b^{\#}_{j'}\> = \frac{\tr\left[\left(\hat a^{\dag}_S\right)^n\hat b^{\#}_{1} \dots \hat b^{\#}_{j'}\left(\hat a_S\right)^n \rho \right]}{\tr\left[\left(\hat a^{\dag}_S\right)^n\left(\hat a_S\right)^n \rho \right]}
\end{equation}
Because $\hat b^{\#}_{1}, \dots, \hat b^{\#}_{j'} \in {\cal A}_{\rm far}$, we find that $\trho[\hat b^{\dag}_k \hat a_S] = 0$ and $\trho[\hat b_k \hat a_S] = 0$. An application of Eq.(\ref{eq:perfMatch}) than shows that
\begin{equation}
\begin{split}
&\tr\left[\left(\hat a^{\dag}_S\right)^n\hat b^{\#}_{1} \dots \hat b^{\#}_{j'}\left(\hat a_S\right)^n \rho \right] \\
&= \tr\left[\left(\hat a^{\dag}_S\right)^n\left(\hat a_S\right)^n \rho \right]\tr\left[\hat b^{\#}_{1} \dots \hat b^{\#}_{j'} \rho \right]
\end{split}
\end{equation}
When we then take into account the denominator in Eq.(\ref{eq:<Y>}), we find
\begin{equation}\label{eq:expB}
\<  \hat b^{\#}_{1} \dots \hat b^{\#}_{j'}\> = \tr\left[\hat b^{\#}_{1} \dots \hat b^{\#}_{j'} \rho \right],
\end{equation}
it automatically follows that $\<\hat Y\> = \tr[\hat Y \rho]$.

The expectation value $ \<\hat a^{\#}_{1} \dots \hat a^{\#}_{j}  \>$ is much more intricate to evaluate since by construction either $\trho{\hat a^{\dag}_k \hat a_S} \neq 0$ or $\trho{\hat a_k \hat a_S} \neq 0$. We then find 
\begin{align}\label{eq:perfMatch2}
&\tr[\left(\hat a^{\dag}_S\right)^n\hat a^{\#}_{1} \dots \hat a^{\#}_{j}\left(\hat a_S\right)^n \rho ]=\sum_{\cal P} \prod_{\{p_1, p_2\} \in {\cal P}} \tr[\hat a^{\#}_{p_1}\hat a^{\#}_{p_2}\rho]\nonumber\\
&= \tr\left[\left(\hat a^{\dag}_S\right)^n\left(\hat a_S\right)^n \rho \right]\tr\left[\hat a^{\#}_{1} \dots \hat a^{\#}_{j} \rho \right] + \text{cross terms}.
\end{align}
We therefore find that
\begin{equation}
\begin{split}\label{eq:expA}
 \<\hat a^{\#}_{1} \dots \hat a^{\#}_{j}  \> = &\tr\left[\hat a^{\#}_{1} \dots \hat a^{\#}_{j} \rho \right]+  {\cal T}.
\end{split}
\end{equation}
The cross terms ${\cal T}$ contain expectation values that combine creation or annihilation operators of the Wick monomial $\hat a^{\#}_{1} \dots \hat a^{\#}_{j}$ with either $\hat a_S$ or $\hat a^{\dag}_S$ as obtained from the perfect matching [Eq. (\ref{eq:perfMatch2})]. For what follows, it is useful to explicitly identify these cross terms as
\begin{equation}
\begin{split}\label{eq:crossTerms}
{\cal T} = &\<\hat a^{\#}_{1} \dots \hat a^{\#}_{j}  \>-\tr\left[\hat a^{\#}_{1} \dots \hat a^{\#}_{j} \rho \right].
\end{split}
\end{equation}

Finally, we can now consider $ \< \hat a^{\#}_{1} \dots \hat a^{\#}_{j}  \hat b^{\#}_{1} \dots \hat b^{\#}_{j'}\> $ and apply many of the same lines of reasoning. We express
\begin{align}\label{eq:perfMatch3}
&\tr[\left(\hat a^{\dag}_S\right)^n\hat a^{\#}_{1} \dots \hat a^{\#}_{j}\hat b^{\#}_{1} \dots \hat b^{\#}_{j'}\left(\hat a_S\right)^n \rho]\\
&=\sum_{\cal P} \prod_{\{p_1, p_2\} \in {\cal P}} \tr[\hat a^{\#}_{p_1}\hat a^{\#}_{p_2}\rho],\nonumber
\end{align}
where $\hat a^{\#}_{p_i}$ can either be a creation/annihilation operator of the type $\hat a$ or of the type $\hat b$. A wide variety of terms will appear in the set of perfect matchings [Eq. (\ref{eq:perfMatch3})]. Generally speaking, we have the terms related to $\hat a^{\#}_{1} \dots \hat a^{\#}_{j}$, those related to photon subtraction, i.e., $\hat a_S$ or $\hat a^{\dag}_S$, and those related to $\hat b^{\#}_{1} \dots \hat b^{\#}_{j'}$. The crucial element is that $\trho[\hat b^{\dag}_k \hat a_S] = 0$ and $\trho[\hat b_k \hat a_S] = 0$, which implies that any perfect matching that matches a subtraction operator $\hat a_S$ or $\hat a^{\dag}_S$ with a creation or annihilation operator that originates from $\hat b^{\#}_{1} \dots \hat b^{\#}_{j'}$ vanishes. In other words, the subtraction operators $\hat a_S$ or $\hat a^{\dag}_S$ can be matched only with creation and annihilation operators that live on ${\cal A}_{\rm near}$. This implies that we can rewrite
\begin{equation}\label{eq:perfMatch4}
\begin{split}
&\frac{\tr[\left(\hat a^{\dag}_S\right)^n\hat a^{\#}_{1} \dots \hat a^{\#}_{j}\hat b^{\#}_{1} \dots \hat b^{\#}_{j'}\left(\hat a_S\right)^n \rho]}{\tr\left[\left(\hat a^{\dag}_S\right)^n\left(\hat a_S\right)^n \rho \right]}\\
=&\tr[\hat a^{\#}_{1} \dots \hat a^{\#}_{j}\hat b^{\#}_{1} \dots \hat b^{\#}_{j'}\rho] \\
&+\frac{\tr[\hat b^{\#}_{1} \dots \hat b^{\#}_{j'}\rho] \tr[\left(\hat a^{\dag}_S\right)^n\hat a^{\#}_{1} \dots \hat a^{\#}_{j}\left(\hat a_S\right)^n\rho]}{\tr\left[\left(\hat a^{\dag}_S\right)^n\left(\hat a_S\right)^n \rho \right]}\\
&-\tr[\hat a^{\#}_{1} \dots \hat a^{\#}_{j}\rho] \tr[\hat b^{\#}_{1} \dots \hat b^{\#}_{j'}\rho]\\
=&\tr[\hat a^{\#}_{1} \dots \hat a^{\#}_{j}\hat b^{\#}_{1} \dots \hat b^{\#}_{j'}\rho] +\tr[\hat b^{\#}_{1} \dots \hat b^{\#}_{j'}\rho] \<\hat a^{\#}_{1} \dots \hat a^{\#}_{j}\>\\
&-\tr[\hat a^{\#}_{1} \dots \hat a^{\#}_{j}\rho] \tr[\hat b^{\#}_{1} \dots \hat b^{\#}_{j'}\rho].
\end{split}
\end{equation}
where the term $\tr[\hat a^{\#}_{1} \dots \hat a^{\#}_{j}\rho] \tr[\hat b^{\#}_{1} \dots \hat b^{\#}_{j'}\rho]$ is subtracted to avoid double counting. We can then use the definition of the cross terms \eqref{eq:crossTerms} to write
\begin{equation}\begin{split}\label{eq:expAB}
 \< \hat a^{\#}_{1} \dots \hat a^{\#}_{j}  \hat b^{\#}_{1} \dots \hat b^{\#}_{j'}\> = &\tr[\hat a^{\#}_{1} \dots \hat a^{\#}_{j}\hat b^{\#}_{1} \dots \hat b^{\#}_{j'}\rho] \\
 &+ {\cal T}\,  \tr[\hat b^{\#}_{1} \dots \hat b^{\#}_{j'}\rho]. 
 \end{split}
\end{equation}
This provides us with the final ingredient we need to complete the proof.

When the results \eqref{eq:expB}, \eqref{eq:expA}, and \eqref{eq:expAB} are combined, we find that all terms proportional to ${\cal T}$ drop out. As such, we find indeed that the identity \eqref{toprove} holds. Because the identity holds for all possible Wick monomials, it follows that 
\begin{equation}
    \<\hat X \hat Y\> - \<\hat X \>\< \hat Y\> = \tr[\hat X \hat Y \rho] - \tr[\hat X \rho]\tr[\hat Y \rho],
\end{equation} 
for any par of observables $\hat X \in {\cal A}_{\rm near}$ and $\hat Y \in {\cal A}_{\rm far}$.\end{proof}

In full analogy, one can prove a second theorem
\begin{theorem}
For any observables $\hat X, \hat Y \in {\cal A}_{\rm far}$, it holds that
\begin{equation}
    \<\hat X \hat Y\> - \<\hat X \>\< \hat Y\> = \tr[\hat X \hat Y \rho] - \tr[\hat X \rho]\tr[\hat Y \rho],
\end{equation}
where $\<\dots\>$ denotes the expectation value in the $n$-photon subtracted state and $\tr[\dots \rho]$ is the expectation value in the initial Gaussian state.
\end{theorem}

\begin{proof}
The proof is fully analogous to Theorem \ref{Theorem1}. The only modification it that in the present case the cross terms all vanish, such that ${\cal T}=0$.
\end{proof}

\section{Moment analysis}\label{sec:Moments}

To provide a complementary quantitative grasp on the statistics of the degrees and clustering coefficients of the emergent networks of photon-number correlations as defined in Eq.~(\ref{Cij}), we analyzed the moments of these distributions.  More specifically, we considered the first four non-trivial moments: mean, variance, skew, and kurtosis. For an arbitrary stochastic variable $X$, these quantities are defined as
\begin{align}
    {\rm Mean} &= \mathbb{E}[X],\\
    {\rm Variance} &= \mathbb{E}[(X-\mathbb{E}[X])^2],\\
    {\rm Skewness} &= \frac{\mathbb{E}[(X-\mathbb{E}[X])^3]}{\mathbb{E}[(X-\mathbb{E}[X])^2]^{3/2}},\\
    {\rm Kurtosis} &= \frac{\mathbb{E}[(X-\mathbb{E}[X])^4]}{\mathbb{E}[(X-\mathbb{E}[X])^2]^{2}},
\end{align}
and thus they can be estimated from the data. Furthermore, we used error propagation methods to estimate the standard statistical error on each of these quantities.

\bibliography{biblio.bib}

\begin{thebibliography}{95}%
\makeatletter
\providecommand \@ifxundefined [1]{%
 \@ifx{#1\undefined}
}%
\providecommand \@ifnum [1]{%
 \ifnum #1\expandafter \@firstoftwo
 \else \expandafter \@secondoftwo
 \fi
}%
\providecommand \@ifx [1]{%
 \ifx #1\expandafter \@firstoftwo
 \else \expandafter \@secondoftwo
 \fi
}%
\providecommand \natexlab [1]{#1}%
\providecommand \enquote  [1]{``#1''}%
\providecommand \bibnamefont  [1]{#1}%
\providecommand \bibfnamefont [1]{#1}%
\providecommand \citenamefont [1]{#1}%
\providecommand \href@noop [0]{\@secondoftwo}%
\providecommand \href [0]{\begingroup \@sanitize@url \@href}%
\providecommand \@href[1]{\@@startlink{#1}\@@href}%
\providecommand \@@href[1]{\endgroup#1\@@endlink}%
\providecommand \@sanitize@url [0]{\catcode `\\12\catcode `\$12\catcode
  `\&12\catcode `\#12\catcode `\^12\catcode `\_12\catcode `\%12\relax}%
\providecommand \@@startlink[1]{}%
\providecommand \@@endlink[0]{}%
\providecommand \url  [0]{\begingroup\@sanitize@url \@url }%
\providecommand \@url [1]{\endgroup\@href {#1}{\urlprefix }}%
\providecommand \urlprefix  [0]{URL }%
\providecommand \Eprint [0]{\href }%
\providecommand \doibase [0]{https://doi.org/}%
\providecommand \selectlanguage [0]{\@gobble}%
\providecommand \bibinfo  [0]{\@secondoftwo}%
\providecommand \bibfield  [0]{\@secondoftwo}%
\providecommand \translation [1]{[#1]}%
\providecommand \BibitemOpen [0]{}%
\providecommand \bibitemStop [0]{}%
\providecommand \bibitemNoStop [0]{.\EOS\space}%
\providecommand \EOS [0]{\spacefactor3000\relax}%
\providecommand \BibitemShut  [1]{\csname bibitem#1\endcsname}%
\let\auto@bib@innerbib\@empty
\bibitem [{\citenamefont {Islam}\ \emph {et~al.}(2015)\citenamefont {Islam},
  \citenamefont {Ma}, \citenamefont {Preiss}, \citenamefont {Eric~Tai},
  \citenamefont {Lukin}, \citenamefont {Rispoli},\ and\ \citenamefont
  {Greiner}}]{Islam15}%
  \BibitemOpen
  \bibfield  {author} {\bibinfo {author} {\bibfnamefont {R.}~\bibnamefont
  {Islam}}, \bibinfo {author} {\bibfnamefont {R.}~\bibnamefont {Ma}}, \bibinfo
  {author} {\bibfnamefont {P.~M.}\ \bibnamefont {Preiss}}, \bibinfo {author}
  {\bibfnamefont {M.}~\bibnamefont {Eric~Tai}}, \bibinfo {author}
  {\bibfnamefont {A.}~\bibnamefont {Lukin}}, \bibinfo {author} {\bibfnamefont
  {M.}~\bibnamefont {Rispoli}},\ and\ \bibinfo {author} {\bibfnamefont
  {M.}~\bibnamefont {Greiner}},\ }\bibfield  {title} {\bibinfo {title}
  {Measuring entanglement entropy in a quantum many-body system},\ }\href
  {https://doi.org/10.1038/nature15750} {\bibfield  {journal} {\bibinfo
  {journal} {Nature}\ }\textbf {\bibinfo {volume} {528}},\ \bibinfo {pages}
  {77} (\bibinfo {year} {2015})}\BibitemShut {NoStop}%
\bibitem [{\citenamefont {Laflorencie}(2016)}]{Laflorencie16}%
  \BibitemOpen
  \bibfield  {author} {\bibinfo {author} {\bibfnamefont {N.}~\bibnamefont
  {Laflorencie}},\ }\bibfield  {title} {\bibinfo {title} {Quantum entanglement
  in condensed matter systems},\ }\href
  {https://doi.org/https://doi.org/10.1016/j.physrep.2016.06.008} {\bibfield
  {journal} {\bibinfo  {journal} {Physics Reports}\ }\textbf {\bibinfo {volume}
  {646}},\ \bibinfo {pages} {1 } (\bibinfo {year} {2016})},\ \bibinfo {note}
  {quantum entanglement in condensed matter systems}\BibitemShut {NoStop}%
\bibitem [{\citenamefont {Or{\'u}s}(2019)}]{Orus19}%
  \BibitemOpen
  \bibfield  {author} {\bibinfo {author} {\bibfnamefont {R.}~\bibnamefont
  {Or{\'u}s}},\ }\bibfield  {title} {\bibinfo {title} {Tensor networks for
  complex quantum systems},\ }\href {https://doi.org/10.1038/s42254-019-0086-7}
  {\bibfield  {journal} {\bibinfo  {journal} {Nature Reviews Physics}\ }\textbf
  {\bibinfo {volume} {1}},\ \bibinfo {pages} {538} (\bibinfo {year}
  {2019})}\BibitemShut {NoStop}%
\bibitem [{\citenamefont {Chitambar}\ and\ \citenamefont
  {Gour}(2019)}]{Chitambar19}%
  \BibitemOpen
  \bibfield  {author} {\bibinfo {author} {\bibfnamefont {E.}~\bibnamefont
  {Chitambar}}\ and\ \bibinfo {author} {\bibfnamefont {G.}~\bibnamefont
  {Gour}},\ }\bibfield  {title} {\bibinfo {title} {Quantum resource theories},\
  }\href {https://doi.org/10.1103/RevModPhys.91.025001} {\bibfield  {journal}
  {\bibinfo  {journal} {Rev. Mod. Phys.}\ }\textbf {\bibinfo {volume} {91}},\
  \bibinfo {pages} {025001} (\bibinfo {year} {2019})}\BibitemShut {NoStop}%
\bibitem [{\citenamefont {Carleo}\ \emph {et~al.}(2019)\citenamefont {Carleo},
  \citenamefont {Cirac}, \citenamefont {Cranmer}, \citenamefont {Daudet},
  \citenamefont {Schuld}, \citenamefont {Tishby}, \citenamefont
  {Vogt-Maranto},\ and\ \citenamefont {Zdeborov\'a}}]{Carleo19}%
  \BibitemOpen
  \bibfield  {author} {\bibinfo {author} {\bibfnamefont {G.}~\bibnamefont
  {Carleo}}, \bibinfo {author} {\bibfnamefont {I.}~\bibnamefont {Cirac}},
  \bibinfo {author} {\bibfnamefont {K.}~\bibnamefont {Cranmer}}, \bibinfo
  {author} {\bibfnamefont {L.}~\bibnamefont {Daudet}}, \bibinfo {author}
  {\bibfnamefont {M.}~\bibnamefont {Schuld}}, \bibinfo {author} {\bibfnamefont
  {N.}~\bibnamefont {Tishby}}, \bibinfo {author} {\bibfnamefont
  {L.}~\bibnamefont {Vogt-Maranto}},\ and\ \bibinfo {author} {\bibfnamefont
  {L.}~\bibnamefont {Zdeborov\'a}},\ }\bibfield  {title} {\bibinfo {title}
  {Machine learning and the physical sciences},\ }\href
  {https://doi.org/10.1103/RevModPhys.91.045002} {\bibfield  {journal}
  {\bibinfo  {journal} {Rev. Mod. Phys.}\ }\textbf {\bibinfo {volume} {91}},\
  \bibinfo {pages} {045002} (\bibinfo {year} {2019})}\BibitemShut {NoStop}%
\bibitem [{\citenamefont {Dunjko}\ and\ \citenamefont
  {Briegel}(2018)}]{Dunjko18}%
  \BibitemOpen
  \bibfield  {author} {\bibinfo {author} {\bibfnamefont {V.}~\bibnamefont
  {Dunjko}}\ and\ \bibinfo {author} {\bibfnamefont {H.~J.}\ \bibnamefont
  {Briegel}},\ }\bibfield  {title} {\bibinfo {title} {Machine learning {\&}
  artificial intelligence in the quantum domain: a review of recent progress},\
  }\href {https://doi.org/10.1088/1361-6633/aab406} {\bibfield  {journal}
  {\bibinfo  {journal} {Reports on Progress in Physics}\ }\textbf {\bibinfo
  {volume} {81}},\ \bibinfo {pages} {074001} (\bibinfo {year}
  {2018})}\BibitemShut {NoStop}%
\bibitem [{\citenamefont {Deng}\ \emph {et~al.}(2017)\citenamefont {Deng},
  \citenamefont {Li},\ and\ \citenamefont {Das~Sarma}}]{Deng17}%
  \BibitemOpen
  \bibfield  {author} {\bibinfo {author} {\bibfnamefont {D.-L.}\ \bibnamefont
  {Deng}}, \bibinfo {author} {\bibfnamefont {X.}~\bibnamefont {Li}},\ and\
  \bibinfo {author} {\bibfnamefont {S.}~\bibnamefont {Das~Sarma}},\ }\bibfield
  {title} {\bibinfo {title} {Quantum entanglement in neural network states},\
  }\href {https://doi.org/10.1103/PhysRevX.7.021021} {\bibfield  {journal}
  {\bibinfo  {journal} {Phys. Rev. X}\ }\textbf {\bibinfo {volume} {7}},\
  \bibinfo {pages} {021021} (\bibinfo {year} {2017})}\BibitemShut {NoStop}%
\bibitem [{\citenamefont {Carleo}\ and\ \citenamefont
  {Troyer}(2017)}]{Carleo17}%
  \BibitemOpen
  \bibfield  {author} {\bibinfo {author} {\bibfnamefont {G.}~\bibnamefont
  {Carleo}}\ and\ \bibinfo {author} {\bibfnamefont {M.}~\bibnamefont
  {Troyer}},\ }\bibfield  {title} {\bibinfo {title} {Solving the quantum
  many-body problem with artificial neural networks},\ }\href
  {https://doi.org/10.1126/science.aag2302} {\bibfield  {journal} {\bibinfo
  {journal} {Science}\ }\textbf {\bibinfo {volume} {355}},\ \bibinfo {pages}
  {602} (\bibinfo {year} {2017})}\BibitemShut {NoStop}%
\bibitem [{\citenamefont {Carrasquilla}()}]{Carrasquilla20}%
  \BibitemOpen
  \bibfield  {author} {\bibinfo {author} {\bibfnamefont {J.}~\bibnamefont
  {Carrasquilla}},\ }\bibfield  {title} {\bibinfo {title} {{Machine Learning
  for Quantum Matter}},\ }\bibinfo {note} {arXiv: 2003.11040}\BibitemShut
  {NoStop}%
\bibitem [{\citenamefont {Guo}\ \emph {et~al.}(2019)\citenamefont {Guo},
  \citenamefont {Barrett}, \citenamefont {Wang},\ and\ \citenamefont
  {Lvovsky}}]{guo2019backpropagation}%
  \BibitemOpen
  \bibfield  {author} {\bibinfo {author} {\bibfnamefont {X.}~\bibnamefont
  {Guo}}, \bibinfo {author} {\bibfnamefont {T.~D.}\ \bibnamefont {Barrett}},
  \bibinfo {author} {\bibfnamefont {Z.~M.}\ \bibnamefont {Wang}},\ and\
  \bibinfo {author} {\bibfnamefont {A.~I.}\ \bibnamefont {Lvovsky}},\
  }\href@noop {} {\bibinfo {title} {Backpropagation through nonlinear units for
  all-optical training of neural networks}} (\bibinfo {year} {2019}),\ \Eprint
  {https://arxiv.org/abs/1912.12256} {arXiv:1912.12256 [cs.ET]} \BibitemShut
  {NoStop}%
\bibitem [{\citenamefont {Asavanant}\ \emph {et~al.}(2019)\citenamefont
  {Asavanant}, \citenamefont {Shiozawa}, \citenamefont {Yokoyama},
  \citenamefont {Charoensombutamon}, \citenamefont {Emura}, \citenamefont
  {Alexander}, \citenamefont {Takeda}, \citenamefont {Yoshikawa}, \citenamefont
  {Menicucci}, \citenamefont {Yonezawa},\ and\ \citenamefont
  {Furusawa}}]{Asavanant19}%
  \BibitemOpen
  \bibfield  {author} {\bibinfo {author} {\bibfnamefont {W.}~\bibnamefont
  {Asavanant}}, \bibinfo {author} {\bibfnamefont {Y.}~\bibnamefont {Shiozawa}},
  \bibinfo {author} {\bibfnamefont {S.}~\bibnamefont {Yokoyama}}, \bibinfo
  {author} {\bibfnamefont {B.}~\bibnamefont {Charoensombutamon}}, \bibinfo
  {author} {\bibfnamefont {H.}~\bibnamefont {Emura}}, \bibinfo {author}
  {\bibfnamefont {R.~N.}\ \bibnamefont {Alexander}}, \bibinfo {author}
  {\bibfnamefont {S.}~\bibnamefont {Takeda}}, \bibinfo {author} {\bibfnamefont
  {J.-i.}\ \bibnamefont {Yoshikawa}}, \bibinfo {author} {\bibfnamefont {N.~C.}\
  \bibnamefont {Menicucci}}, \bibinfo {author} {\bibfnamefont {H.}~\bibnamefont
  {Yonezawa}},\ and\ \bibinfo {author} {\bibfnamefont {A.}~\bibnamefont
  {Furusawa}},\ }\bibfield  {title} {\bibinfo {title} {Generation of
  time-domain-multiplexed two-dimensional cluster state},\ }\href
  {https://doi.org/10.1126/science.aay2645} {\bibfield  {journal} {\bibinfo
  {journal} {Science}\ }\textbf {\bibinfo {volume} {366}},\ \bibinfo {pages}
  {373} (\bibinfo {year} {2019})}\BibitemShut {NoStop}%
\bibitem [{\citenamefont {Larsen}\ \emph {et~al.}(2019)\citenamefont {Larsen},
  \citenamefont {Guo}, \citenamefont {Breum}, \citenamefont
  {Neergaard-Nielsen},\ and\ \citenamefont {Andersen}}]{Larsen19}%
  \BibitemOpen
  \bibfield  {author} {\bibinfo {author} {\bibfnamefont {M.~V.}\ \bibnamefont
  {Larsen}}, \bibinfo {author} {\bibfnamefont {X.}~\bibnamefont {Guo}},
  \bibinfo {author} {\bibfnamefont {C.~R.}\ \bibnamefont {Breum}}, \bibinfo
  {author} {\bibfnamefont {J.~S.}\ \bibnamefont {Neergaard-Nielsen}},\ and\
  \bibinfo {author} {\bibfnamefont {U.~L.}\ \bibnamefont {Andersen}},\
  }\bibfield  {title} {\bibinfo {title} {Deterministic generation of a
  two-dimensional cluster state},\ }\href
  {https://doi.org/10.1126/science.aay4354} {\bibfield  {journal} {\bibinfo
  {journal} {Science}\ }\textbf {\bibinfo {volume} {366}},\ \bibinfo {pages}
  {369} (\bibinfo {year} {2019})},\ \Eprint
  {https://arxiv.org/abs/https://science.sciencemag.org/content/366/6463/369.full.pdf}
  {https://science.sciencemag.org/content/366/6463/369.full.pdf} \BibitemShut
  {NoStop}%
\bibitem [{\citenamefont {Chen}\ \emph {et~al.}(2014)\citenamefont {Chen},
  \citenamefont {Menicucci},\ and\ \citenamefont {Pfister}}]{Chen14}%
  \BibitemOpen
  \bibfield  {author} {\bibinfo {author} {\bibfnamefont {M.}~\bibnamefont
  {Chen}}, \bibinfo {author} {\bibfnamefont {N.~C.}\ \bibnamefont
  {Menicucci}},\ and\ \bibinfo {author} {\bibfnamefont {O.}~\bibnamefont
  {Pfister}},\ }\bibfield  {title} {\bibinfo {title} {Experimental realization
  of multipartite entanglement of 60 modes of a quantum optical frequency
  comb},\ }\href {https://doi.org/10.1103/PhysRevLett.112.120505} {\bibfield
  {journal} {\bibinfo  {journal} {Phys. Rev. Lett.}\ }\textbf {\bibinfo
  {volume} {112}},\ \bibinfo {pages} {120505} (\bibinfo {year}
  {2014})}\BibitemShut {NoStop}%
\bibitem [{\citenamefont {Ra}\ \emph {et~al.}(2020)\citenamefont {Ra},
  \citenamefont {Dufour}, \citenamefont {Walschaers}, \citenamefont {Jacquard},
  \citenamefont {Michel}, \citenamefont {Fabre},\ and\ \citenamefont
  {Treps}}]{Ra19}%
  \BibitemOpen
  \bibfield  {author} {\bibinfo {author} {\bibfnamefont {Y.-S.}\ \bibnamefont
  {Ra}}, \bibinfo {author} {\bibfnamefont {A.}~\bibnamefont {Dufour}}, \bibinfo
  {author} {\bibfnamefont {M.}~\bibnamefont {Walschaers}}, \bibinfo {author}
  {\bibfnamefont {C.}~\bibnamefont {Jacquard}}, \bibinfo {author}
  {\bibfnamefont {T.}~\bibnamefont {Michel}}, \bibinfo {author} {\bibfnamefont
  {C.}~\bibnamefont {Fabre}},\ and\ \bibinfo {author} {\bibfnamefont
  {N.}~\bibnamefont {Treps}},\ }\bibfield  {title} {\bibinfo {title}
  {Non-gaussian quantum states of a multimode light field},\ }\href
  {https://doi.org/10.1038/s41567-019-0726-y} {\bibfield  {journal} {\bibinfo
  {journal} {Nature Physics}\ }\textbf {\bibinfo {volume} {16}},\ \bibinfo
  {pages} {144} (\bibinfo {year} {2020})}\BibitemShut {NoStop}%
\bibitem [{\citenamefont {Biagi}\ \emph {et~al.}(2020)\citenamefont {Biagi},
  \citenamefont {Costanzo}, \citenamefont {Bellini},\ and\ \citenamefont
  {Zavatta}}]{Biagi20}%
  \BibitemOpen
  \bibfield  {author} {\bibinfo {author} {\bibfnamefont {N.}~\bibnamefont
  {Biagi}}, \bibinfo {author} {\bibfnamefont {L.~S.}\ \bibnamefont {Costanzo}},
  \bibinfo {author} {\bibfnamefont {M.}~\bibnamefont {Bellini}},\ and\ \bibinfo
  {author} {\bibfnamefont {A.}~\bibnamefont {Zavatta}},\ }\bibfield  {title}
  {\bibinfo {title} {Entangling macroscopic light states by delocalized photon
  addition},\ }\href {https://doi.org/10.1103/PhysRevLett.124.033604}
  {\bibfield  {journal} {\bibinfo  {journal} {Phys. Rev. Lett.}\ }\textbf
  {\bibinfo {volume} {124}},\ \bibinfo {pages} {033604} (\bibinfo {year}
  {2020})}\BibitemShut {NoStop}%
\bibitem [{\citenamefont {Walschaers}\ \emph
  {et~al.}(2017{\natexlab{a}})\citenamefont {Walschaers}, \citenamefont
  {Fabre}, \citenamefont {Parigi},\ and\ \citenamefont
  {Treps}}]{WalschaersPRA17}%
  \BibitemOpen
  \bibfield  {author} {\bibinfo {author} {\bibfnamefont {M.}~\bibnamefont
  {Walschaers}}, \bibinfo {author} {\bibfnamefont {C.}~\bibnamefont {Fabre}},
  \bibinfo {author} {\bibfnamefont {V.}~\bibnamefont {Parigi}},\ and\ \bibinfo
  {author} {\bibfnamefont {N.}~\bibnamefont {Treps}},\ }\bibfield  {title}
  {\bibinfo {title} {Statistical signatures of multimode single-photon-added
  and -subtracted states of light},\ }\href
  {https://doi.org/10.1103/PhysRevA.96.053835} {\bibfield  {journal} {\bibinfo
  {journal} {Phys. Rev. A}\ }\textbf {\bibinfo {volume} {96}},\ \bibinfo
  {pages} {053835} (\bibinfo {year} {2017}{\natexlab{a}})}\BibitemShut
  {NoStop}%
\bibitem [{\citenamefont {Zhuang}\ \emph {et~al.}(2018)\citenamefont {Zhuang},
  \citenamefont {Shor},\ and\ \citenamefont {Shapiro}}]{Zhuang18}%
  \BibitemOpen
  \bibfield  {author} {\bibinfo {author} {\bibfnamefont {Q.}~\bibnamefont
  {Zhuang}}, \bibinfo {author} {\bibfnamefont {P.~W.}\ \bibnamefont {Shor}},\
  and\ \bibinfo {author} {\bibfnamefont {J.~H.}\ \bibnamefont {Shapiro}},\
  }\bibfield  {title} {\bibinfo {title} {Resource theory of non-gaussian
  operations},\ }\href {https://doi.org/10.1103/PhysRevA.97.052317} {\bibfield
  {journal} {\bibinfo  {journal} {Phys. Rev. A}\ }\textbf {\bibinfo {volume}
  {97}},\ \bibinfo {pages} {052317} (\bibinfo {year} {2018})}\BibitemShut
  {NoStop}%
\bibitem [{\citenamefont {Albarelli}\ \emph {et~al.}(2018)\citenamefont
  {Albarelli}, \citenamefont {Genoni}, \citenamefont {Paris},\ and\
  \citenamefont {Ferraro}}]{Albarelli18}%
  \BibitemOpen
  \bibfield  {author} {\bibinfo {author} {\bibfnamefont {F.}~\bibnamefont
  {Albarelli}}, \bibinfo {author} {\bibfnamefont {M.~G.}\ \bibnamefont
  {Genoni}}, \bibinfo {author} {\bibfnamefont {M.~G.~A.}\ \bibnamefont
  {Paris}},\ and\ \bibinfo {author} {\bibfnamefont {A.}~\bibnamefont
  {Ferraro}},\ }\bibfield  {title} {\bibinfo {title} {Resource theory of
  quantum non-gaussianity and wigner negativity},\ }\href
  {https://doi.org/10.1103/PhysRevA.98.052350} {\bibfield  {journal} {\bibinfo
  {journal} {Phys. Rev. A}\ }\textbf {\bibinfo {volume} {98}},\ \bibinfo
  {pages} {052350} (\bibinfo {year} {2018})}\BibitemShut {NoStop}%
\bibitem [{\citenamefont {Takagi}\ and\ \citenamefont
  {Zhuang}(2018)}]{Takagi18}%
  \BibitemOpen
  \bibfield  {author} {\bibinfo {author} {\bibfnamefont {R.}~\bibnamefont
  {Takagi}}\ and\ \bibinfo {author} {\bibfnamefont {Q.}~\bibnamefont
  {Zhuang}},\ }\bibfield  {title} {\bibinfo {title} {Convex resource theory of
  non-gaussianity},\ }\href {https://doi.org/10.1103/PhysRevA.97.062337}
  {\bibfield  {journal} {\bibinfo  {journal} {Phys. Rev. A}\ }\textbf {\bibinfo
  {volume} {97}},\ \bibinfo {pages} {062337} (\bibinfo {year}
  {2018})}\BibitemShut {NoStop}%
\bibitem [{\citenamefont {Cimini}\ \emph {et~al.}(2020)\citenamefont {Cimini},
  \citenamefont {Barbieri}, \citenamefont {Treps}, \citenamefont {Walschaers},\
  and\ \citenamefont {Parigi}}]{cimini2020neural}%
  \BibitemOpen
  \bibfield  {author} {\bibinfo {author} {\bibfnamefont {V.}~\bibnamefont
  {Cimini}}, \bibinfo {author} {\bibfnamefont {M.}~\bibnamefont {Barbieri}},
  \bibinfo {author} {\bibfnamefont {N.}~\bibnamefont {Treps}}, \bibinfo
  {author} {\bibfnamefont {M.}~\bibnamefont {Walschaers}},\ and\ \bibinfo
  {author} {\bibfnamefont {V.}~\bibnamefont {Parigi}},\ }\bibfield  {title}
  {\bibinfo {title} {Neural networks for detecting multimode wigner
  negativity},\ }\href {https://doi.org/10.1103/PhysRevLett.125.160504}
  {\bibfield  {journal} {\bibinfo  {journal} {Phys. Rev. Lett.}\ }\textbf
  {\bibinfo {volume} {125}},\ \bibinfo {pages} {160504} (\bibinfo {year}
  {2020})}\BibitemShut {NoStop}%
\bibitem [{\citenamefont {Walschaers}\ \emph {et~al.}(2020)\citenamefont
  {Walschaers}, \citenamefont {Parigi},\ and\ \citenamefont
  {Treps}}]{walschaers2020conditional}%
  \BibitemOpen
  \bibfield  {author} {\bibinfo {author} {\bibfnamefont {M.}~\bibnamefont
  {Walschaers}}, \bibinfo {author} {\bibfnamefont {V.}~\bibnamefont {Parigi}},\
  and\ \bibinfo {author} {\bibfnamefont {N.}~\bibnamefont {Treps}},\ }\bibfield
   {title} {\bibinfo {title} {Practical framework for conditional non-gaussian
  quantum state preparation},\ }\href
  {https://doi.org/10.1103/PRXQuantum.1.020305} {\bibfield  {journal} {\bibinfo
   {journal} {PRX Quantum}\ }\textbf {\bibinfo {volume} {1}},\ \bibinfo {pages}
  {020305} (\bibinfo {year} {2020})}\BibitemShut {NoStop}%
\bibitem [{\citenamefont {Chabaud}\ \emph
  {et~al.}(2017{\natexlab{a}})\citenamefont {Chabaud}, \citenamefont {Douce},
  \citenamefont {Markham}, \citenamefont {van Loock}, \citenamefont {Kashefi},\
  and\ \citenamefont {Ferrini}}]{Chabaud17}%
  \BibitemOpen
  \bibfield  {author} {\bibinfo {author} {\bibfnamefont {U.}~\bibnamefont
  {Chabaud}}, \bibinfo {author} {\bibfnamefont {T.}~\bibnamefont {Douce}},
  \bibinfo {author} {\bibfnamefont {D.}~\bibnamefont {Markham}}, \bibinfo
  {author} {\bibfnamefont {P.}~\bibnamefont {van Loock}}, \bibinfo {author}
  {\bibfnamefont {E.}~\bibnamefont {Kashefi}},\ and\ \bibinfo {author}
  {\bibfnamefont {G.}~\bibnamefont {Ferrini}},\ }\bibfield  {title} {\bibinfo
  {title} {Continuous-variable sampling from photon-added or photon-subtracted
  squeezed states},\ }\href {https://doi.org/10.1103/PhysRevA.96.062307}
  {\bibfield  {journal} {\bibinfo  {journal} {Phys. Rev. A}\ }\textbf {\bibinfo
  {volume} {96}},\ \bibinfo {pages} {062307} (\bibinfo {year}
  {2017}{\natexlab{a}})}\BibitemShut {NoStop}%
\bibitem [{\citenamefont {Hamilton}\ \emph {et~al.}(2017)\citenamefont
  {Hamilton}, \citenamefont {Kruse}, \citenamefont {Sansoni}, \citenamefont
  {Barkhofen}, \citenamefont {Silberhorn},\ and\ \citenamefont
  {Jex}}]{PhysRevLett.119.170501}%
  \BibitemOpen
  \bibfield  {author} {\bibinfo {author} {\bibfnamefont {C.~S.}\ \bibnamefont
  {Hamilton}}, \bibinfo {author} {\bibfnamefont {R.}~\bibnamefont {Kruse}},
  \bibinfo {author} {\bibfnamefont {L.}~\bibnamefont {Sansoni}}, \bibinfo
  {author} {\bibfnamefont {S.}~\bibnamefont {Barkhofen}}, \bibinfo {author}
  {\bibfnamefont {C.}~\bibnamefont {Silberhorn}},\ and\ \bibinfo {author}
  {\bibfnamefont {I.}~\bibnamefont {Jex}},\ }\bibfield  {title} {\bibinfo
  {title} {Gaussian boson sampling},\ }\href
  {https://doi.org/10.1103/PhysRevLett.119.170501} {\bibfield  {journal}
  {\bibinfo  {journal} {Phys. Rev. Lett.}\ }\textbf {\bibinfo {volume} {119}},\
  \bibinfo {pages} {170501} (\bibinfo {year} {2017})}\BibitemShut {NoStop}%
\bibitem [{\citenamefont {Aaronson}\ and\ \citenamefont
  {Arkhipov}(2011)}]{10.1145/1993636.1993682}%
  \BibitemOpen
  \bibfield  {author} {\bibinfo {author} {\bibfnamefont {S.}~\bibnamefont
  {Aaronson}}\ and\ \bibinfo {author} {\bibfnamefont {A.}~\bibnamefont
  {Arkhipov}},\ }\bibfield  {title} {\bibinfo {title} {The computational
  complexity of linear optics},\ }in\ \href
  {https://doi.org/10.1145/1993636.1993682} {\emph {\bibinfo {booktitle}
  {Proceedings of the Forty-Third Annual ACM Symposium on Theory of
  Computing}}},\ \bibinfo {series and number} {STOC '11}\ (\bibinfo
  {publisher} {Association for Computing Machinery},\ \bibinfo {address} {New
  York, NY, USA},\ \bibinfo {year} {2011})\ p.\ \bibinfo {pages}
  {333–342}\BibitemShut {NoStop}%
\bibitem [{\citenamefont {Walschaers}\ \emph
  {et~al.}(2016{\natexlab{a}})\citenamefont {Walschaers}, \citenamefont
  {Kuipers}, \citenamefont {Urbina}, \citenamefont {Mayer}, \citenamefont
  {Tichy}, \citenamefont {Richter},\ and\ \citenamefont
  {Buchleitner}}]{Walschaers_2016}%
  \BibitemOpen
  \bibfield  {author} {\bibinfo {author} {\bibfnamefont {M.}~\bibnamefont
  {Walschaers}}, \bibinfo {author} {\bibfnamefont {J.}~\bibnamefont {Kuipers}},
  \bibinfo {author} {\bibfnamefont {J.-D.}\ \bibnamefont {Urbina}}, \bibinfo
  {author} {\bibfnamefont {K.}~\bibnamefont {Mayer}}, \bibinfo {author}
  {\bibfnamefont {M.~C.}\ \bibnamefont {Tichy}}, \bibinfo {author}
  {\bibfnamefont {K.}~\bibnamefont {Richter}},\ and\ \bibinfo {author}
  {\bibfnamefont {A.}~\bibnamefont {Buchleitner}},\ }\bibfield  {title}
  {\bibinfo {title} {Statistical benchmark for {BosonSampling}},\ }\href
  {https://doi.org/10.1088/1367-2630/18/3/032001} {\bibfield  {journal}
  {\bibinfo  {journal} {New Journal of Physics}\ }\textbf {\bibinfo {volume}
  {18}},\ \bibinfo {pages} {032001} (\bibinfo {year}
  {2016}{\natexlab{a}})}\BibitemShut {NoStop}%
\bibitem [{\citenamefont {Giordani}\ \emph {et~al.}(2018)\citenamefont
  {Giordani}, \citenamefont {Flamini}, \citenamefont {Pompili}, \citenamefont
  {Viggianiello}, \citenamefont {Spagnolo}, \citenamefont {Crespi},
  \citenamefont {Osellame}, \citenamefont {Wiebe}, \citenamefont {Walschaers},
  \citenamefont {Buchleitner},\ and\ \citenamefont {Sciarrino}}]{Giordani}%
  \BibitemOpen
  \bibfield  {author} {\bibinfo {author} {\bibfnamefont {T.}~\bibnamefont
  {Giordani}}, \bibinfo {author} {\bibfnamefont {F.}~\bibnamefont {Flamini}},
  \bibinfo {author} {\bibfnamefont {M.}~\bibnamefont {Pompili}}, \bibinfo
  {author} {\bibfnamefont {N.}~\bibnamefont {Viggianiello}}, \bibinfo {author}
  {\bibfnamefont {N.}~\bibnamefont {Spagnolo}}, \bibinfo {author}
  {\bibfnamefont {A.}~\bibnamefont {Crespi}}, \bibinfo {author} {\bibfnamefont
  {R.}~\bibnamefont {Osellame}}, \bibinfo {author} {\bibfnamefont
  {N.}~\bibnamefont {Wiebe}}, \bibinfo {author} {\bibfnamefont
  {M.}~\bibnamefont {Walschaers}}, \bibinfo {author} {\bibfnamefont
  {A.}~\bibnamefont {Buchleitner}},\ and\ \bibinfo {author} {\bibfnamefont
  {F.}~\bibnamefont {Sciarrino}},\ }\bibfield  {title} {\bibinfo {title}
  {Experimental statistical signature of many-body quantum interference},\
  }\href@noop {} {\bibfield  {journal} {\bibinfo  {journal} {Nature Photonics}\
  }\textbf {\bibinfo {volume} {12}},\ \bibinfo {pages} {173} (\bibinfo {year}
  {2018})}\BibitemShut {NoStop}%
\bibitem [{\citenamefont {Phillips}\ \emph {et~al.}(2019)\citenamefont
  {Phillips}, \citenamefont {Walschaers}, \citenamefont {Renema}, \citenamefont
  {Walmsley}, \citenamefont {Treps},\ and\ \citenamefont
  {Sperling}}]{PhysRevA.99.023836}%
  \BibitemOpen
  \bibfield  {author} {\bibinfo {author} {\bibfnamefont {D.~S.}\ \bibnamefont
  {Phillips}}, \bibinfo {author} {\bibfnamefont {M.}~\bibnamefont
  {Walschaers}}, \bibinfo {author} {\bibfnamefont {J.~J.}\ \bibnamefont
  {Renema}}, \bibinfo {author} {\bibfnamefont {I.~A.}\ \bibnamefont
  {Walmsley}}, \bibinfo {author} {\bibfnamefont {N.}~\bibnamefont {Treps}},\
  and\ \bibinfo {author} {\bibfnamefont {J.}~\bibnamefont {Sperling}},\
  }\bibfield  {title} {\bibinfo {title} {Benchmarking of gaussian boson
  sampling using two-point correlators},\ }\href
  {https://doi.org/10.1103/PhysRevA.99.023836} {\bibfield  {journal} {\bibinfo
  {journal} {Phys. Rev. A}\ }\textbf {\bibinfo {volume} {99}},\ \bibinfo
  {pages} {023836} (\bibinfo {year} {2019})}\BibitemShut {NoStop}%
\bibitem [{\citenamefont {Zhong}\ \emph {et~al.}(2020)\citenamefont {Zhong},
  \citenamefont {Wang}, \citenamefont {Deng}, \citenamefont {Chen},
  \citenamefont {Peng}, \citenamefont {Luo}, \citenamefont {Qin}, \citenamefont
  {Wu}, \citenamefont {Ding}, \citenamefont {Hu}, \citenamefont {Hu},
  \citenamefont {Yang}, \citenamefont {Zhang}, \citenamefont {Li},
  \citenamefont {Li}, \citenamefont {Jiang}, \citenamefont {Gan}, \citenamefont
  {Yang}, \citenamefont {You}, \citenamefont {Wang}, \citenamefont {Li},
  \citenamefont {Liu}, \citenamefont {Lu},\ and\ \citenamefont
  {Pan}}]{Zhong1460}%
  \BibitemOpen
  \bibfield  {author} {\bibinfo {author} {\bibfnamefont {H.-S.}\ \bibnamefont
  {Zhong}}, \bibinfo {author} {\bibfnamefont {H.}~\bibnamefont {Wang}},
  \bibinfo {author} {\bibfnamefont {Y.-H.}\ \bibnamefont {Deng}}, \bibinfo
  {author} {\bibfnamefont {M.-C.}\ \bibnamefont {Chen}}, \bibinfo {author}
  {\bibfnamefont {L.-C.}\ \bibnamefont {Peng}}, \bibinfo {author}
  {\bibfnamefont {Y.-H.}\ \bibnamefont {Luo}}, \bibinfo {author} {\bibfnamefont
  {J.}~\bibnamefont {Qin}}, \bibinfo {author} {\bibfnamefont {D.}~\bibnamefont
  {Wu}}, \bibinfo {author} {\bibfnamefont {X.}~\bibnamefont {Ding}}, \bibinfo
  {author} {\bibfnamefont {Y.}~\bibnamefont {Hu}}, \bibinfo {author}
  {\bibfnamefont {P.}~\bibnamefont {Hu}}, \bibinfo {author} {\bibfnamefont
  {X.-Y.}\ \bibnamefont {Yang}}, \bibinfo {author} {\bibfnamefont {W.-J.}\
  \bibnamefont {Zhang}}, \bibinfo {author} {\bibfnamefont {H.}~\bibnamefont
  {Li}}, \bibinfo {author} {\bibfnamefont {Y.}~\bibnamefont {Li}}, \bibinfo
  {author} {\bibfnamefont {X.}~\bibnamefont {Jiang}}, \bibinfo {author}
  {\bibfnamefont {L.}~\bibnamefont {Gan}}, \bibinfo {author} {\bibfnamefont
  {G.}~\bibnamefont {Yang}}, \bibinfo {author} {\bibfnamefont {L.}~\bibnamefont
  {You}}, \bibinfo {author} {\bibfnamefont {Z.}~\bibnamefont {Wang}}, \bibinfo
  {author} {\bibfnamefont {L.}~\bibnamefont {Li}}, \bibinfo {author}
  {\bibfnamefont {N.-L.}\ \bibnamefont {Liu}}, \bibinfo {author} {\bibfnamefont
  {C.-Y.}\ \bibnamefont {Lu}},\ and\ \bibinfo {author} {\bibfnamefont {J.-W.}\
  \bibnamefont {Pan}},\ }\bibfield  {title} {\bibinfo {title} {Quantum
  computational advantage using photons},\ }\href
  {https://doi.org/10.1126/science.abe8770} {\bibfield  {journal} {\bibinfo
  {journal} {Science}\ }\textbf {\bibinfo {volume} {370}},\ \bibinfo {pages}
  {1460} (\bibinfo {year} {2020})}\BibitemShut {NoStop}%
\bibitem [{\citenamefont {Zhong}\ \emph {et~al.}(2021)\citenamefont {Zhong},
  \citenamefont {Deng}, \citenamefont {Qin}, \citenamefont {Wang},
  \citenamefont {Chen}, \citenamefont {Peng}, \citenamefont {Luo},
  \citenamefont {Wu}, \citenamefont {Gong}, \citenamefont {Su}, \citenamefont
  {Hu}, \citenamefont {Hu}, \citenamefont {Yang}, \citenamefont {Zhang},
  \citenamefont {Li}, \citenamefont {Li}, \citenamefont {Jiang}, \citenamefont
  {Gan}, \citenamefont {Yang}, \citenamefont {You}, \citenamefont {Wang},
  \citenamefont {Li}, \citenamefont {Liu}, \citenamefont {Renema},
  \citenamefont {Lu},\ and\ \citenamefont {Pan}}]{zhong2021phaseprogrammable}%
  \BibitemOpen
  \bibfield  {author} {\bibinfo {author} {\bibfnamefont {H.-S.}\ \bibnamefont
  {Zhong}}, \bibinfo {author} {\bibfnamefont {Y.-H.}\ \bibnamefont {Deng}},
  \bibinfo {author} {\bibfnamefont {J.}~\bibnamefont {Qin}}, \bibinfo {author}
  {\bibfnamefont {H.}~\bibnamefont {Wang}}, \bibinfo {author} {\bibfnamefont
  {M.-C.}\ \bibnamefont {Chen}}, \bibinfo {author} {\bibfnamefont {L.-C.}\
  \bibnamefont {Peng}}, \bibinfo {author} {\bibfnamefont {Y.-H.}\ \bibnamefont
  {Luo}}, \bibinfo {author} {\bibfnamefont {D.}~\bibnamefont {Wu}}, \bibinfo
  {author} {\bibfnamefont {S.-Q.}\ \bibnamefont {Gong}}, \bibinfo {author}
  {\bibfnamefont {H.}~\bibnamefont {Su}}, \bibinfo {author} {\bibfnamefont
  {Y.}~\bibnamefont {Hu}}, \bibinfo {author} {\bibfnamefont {P.}~\bibnamefont
  {Hu}}, \bibinfo {author} {\bibfnamefont {X.-Y.}\ \bibnamefont {Yang}},
  \bibinfo {author} {\bibfnamefont {W.-J.}\ \bibnamefont {Zhang}}, \bibinfo
  {author} {\bibfnamefont {H.}~\bibnamefont {Li}}, \bibinfo {author}
  {\bibfnamefont {Y.}~\bibnamefont {Li}}, \bibinfo {author} {\bibfnamefont
  {X.}~\bibnamefont {Jiang}}, \bibinfo {author} {\bibfnamefont
  {L.}~\bibnamefont {Gan}}, \bibinfo {author} {\bibfnamefont {G.}~\bibnamefont
  {Yang}}, \bibinfo {author} {\bibfnamefont {L.}~\bibnamefont {You}}, \bibinfo
  {author} {\bibfnamefont {Z.}~\bibnamefont {Wang}}, \bibinfo {author}
  {\bibfnamefont {L.}~\bibnamefont {Li}}, \bibinfo {author} {\bibfnamefont
  {N.-L.}\ \bibnamefont {Liu}}, \bibinfo {author} {\bibfnamefont {J.~J.}\
  \bibnamefont {Renema}}, \bibinfo {author} {\bibfnamefont {C.-Y.}\
  \bibnamefont {Lu}},\ and\ \bibinfo {author} {\bibfnamefont {J.-W.}\
  \bibnamefont {Pan}},\ }\bibfield  {title} {\bibinfo {title}
  {Phase-programmable gaussian boson sampling using stimulated squeezed
  light},\ }\href {https://doi.org/10.1103/PhysRevLett.127.180502} {\bibfield
  {journal} {\bibinfo  {journal} {Phys. Rev. Lett.}\ }\textbf {\bibinfo
  {volume} {127}},\ \bibinfo {pages} {180502} (\bibinfo {year}
  {2021})}\BibitemShut {NoStop}%
\bibitem [{\citenamefont {Madsen}\ \emph {et~al.}(2022)\citenamefont {Madsen},
  \citenamefont {Laudenbach}, \citenamefont {Askarani}, \citenamefont
  {Rortais}, \citenamefont {Vincent}, \citenamefont {Bulmer}, \citenamefont
  {Miatto}, \citenamefont {Neuhaus}, \citenamefont {Helt}, \citenamefont
  {Collins}, \citenamefont {Lita}, \citenamefont {Gerrits}, \citenamefont
  {Nam}, \citenamefont {Vaidya}, \citenamefont {Menotti}, \citenamefont
  {Dhand}, \citenamefont {Vernon}, \citenamefont {Quesada},\ and\ \citenamefont
  {Lavoie}}]{Xanadu-Advantage}%
  \BibitemOpen
  \bibfield  {author} {\bibinfo {author} {\bibfnamefont {L.~S.}\ \bibnamefont
  {Madsen}}, \bibinfo {author} {\bibfnamefont {F.}~\bibnamefont {Laudenbach}},
  \bibinfo {author} {\bibfnamefont {M.~F.}\ \bibnamefont {Askarani}}, \bibinfo
  {author} {\bibfnamefont {F.}~\bibnamefont {Rortais}}, \bibinfo {author}
  {\bibfnamefont {T.}~\bibnamefont {Vincent}}, \bibinfo {author} {\bibfnamefont
  {J.~F.~F.}\ \bibnamefont {Bulmer}}, \bibinfo {author} {\bibfnamefont {F.~M.}\
  \bibnamefont {Miatto}}, \bibinfo {author} {\bibfnamefont {L.}~\bibnamefont
  {Neuhaus}}, \bibinfo {author} {\bibfnamefont {L.~G.}\ \bibnamefont {Helt}},
  \bibinfo {author} {\bibfnamefont {M.~J.}\ \bibnamefont {Collins}}, \bibinfo
  {author} {\bibfnamefont {A.~E.}\ \bibnamefont {Lita}}, \bibinfo {author}
  {\bibfnamefont {T.}~\bibnamefont {Gerrits}}, \bibinfo {author} {\bibfnamefont
  {S.~W.}\ \bibnamefont {Nam}}, \bibinfo {author} {\bibfnamefont {V.~D.}\
  \bibnamefont {Vaidya}}, \bibinfo {author} {\bibfnamefont {M.}~\bibnamefont
  {Menotti}}, \bibinfo {author} {\bibfnamefont {I.}~\bibnamefont {Dhand}},
  \bibinfo {author} {\bibfnamefont {Z.}~\bibnamefont {Vernon}}, \bibinfo
  {author} {\bibfnamefont {N.}~\bibnamefont {Quesada}},\ and\ \bibinfo {author}
  {\bibfnamefont {J.}~\bibnamefont {Lavoie}},\ }\bibfield  {title} {\bibinfo
  {title} {Quantum computational advantage with a programmable photonic
  processor},\ }\href {https://doi.org/10.1038/s41586-022-04725-x} {\bibfield
  {journal} {\bibinfo  {journal} {Nature}\ }\textbf {\bibinfo {volume} {606}},\
  \bibinfo {pages} {75} (\bibinfo {year} {2022})}\BibitemShut {NoStop}%
\bibitem [{\citenamefont {Cai}\ \emph {et~al.}(2017)\citenamefont {Cai},
  \citenamefont {Roslund}, \citenamefont {Ferrini}, \citenamefont {Arzani},
  \citenamefont {Xu}, \citenamefont {Fabre},\ and\ \citenamefont
  {Treps}}]{Cai17}%
  \BibitemOpen
  \bibfield  {author} {\bibinfo {author} {\bibfnamefont {Y.}~\bibnamefont
  {Cai}}, \bibinfo {author} {\bibfnamefont {J.}~\bibnamefont {Roslund}},
  \bibinfo {author} {\bibfnamefont {G.}~\bibnamefont {Ferrini}}, \bibinfo
  {author} {\bibfnamefont {F.}~\bibnamefont {Arzani}}, \bibinfo {author}
  {\bibfnamefont {X.}~\bibnamefont {Xu}}, \bibinfo {author} {\bibfnamefont
  {C.}~\bibnamefont {Fabre}},\ and\ \bibinfo {author} {\bibfnamefont
  {N.}~\bibnamefont {Treps}},\ }\bibfield  {title} {\bibinfo {title} {Multimode
  entanglement in reconfigurable graph states using optical frequency combs},\
  }\href {https://doi.org/10.1038/ncomms15645} {\bibfield  {journal} {\bibinfo
  {journal} {Nat. Commun.}\ }\textbf {\bibinfo {volume} {8}},\ \bibinfo {pages}
  {15645} (\bibinfo {year} {2017})}\BibitemShut {NoStop}%
\bibitem [{\citenamefont {Nokkala}\ \emph
  {et~al.}(2018{\natexlab{a}})\citenamefont {Nokkala}, \citenamefont {Arzani},
  \citenamefont {Galve}, \citenamefont {Zambrini}, \citenamefont {Maniscalco},
  \citenamefont {Piilo}, \citenamefont {Treps},\ and\ \citenamefont
  {Parigi}}]{Nokkala18a}%
  \BibitemOpen
  \bibfield  {author} {\bibinfo {author} {\bibfnamefont {J.}~\bibnamefont
  {Nokkala}}, \bibinfo {author} {\bibfnamefont {F.}~\bibnamefont {Arzani}},
  \bibinfo {author} {\bibfnamefont {F.}~\bibnamefont {Galve}}, \bibinfo
  {author} {\bibfnamefont {R.}~\bibnamefont {Zambrini}}, \bibinfo {author}
  {\bibfnamefont {S.}~\bibnamefont {Maniscalco}}, \bibinfo {author}
  {\bibfnamefont {J.}~\bibnamefont {Piilo}}, \bibinfo {author} {\bibfnamefont
  {N.}~\bibnamefont {Treps}},\ and\ \bibinfo {author} {\bibfnamefont
  {V.}~\bibnamefont {Parigi}},\ }\bibfield  {title} {\bibinfo {title}
  {Reconfigurable optical implementation of quantum complex networks},\ }\href
  {https://doi.org/10.1088/1367-2630/aabc77} {\bibfield  {journal} {\bibinfo
  {journal} {New Journal of Physics}\ }\textbf {\bibinfo {volume} {20}},\
  \bibinfo {pages} {053024} (\bibinfo {year} {2018}{\natexlab{a}})}\BibitemShut
  {NoStop}%
\bibitem [{\citenamefont {Sansavini}\ and\ \citenamefont
  {Parigi}(2020)}]{Sansavini20}%
  \BibitemOpen
  \bibfield  {author} {\bibinfo {author} {\bibfnamefont {F.}~\bibnamefont
  {Sansavini}}\ and\ \bibinfo {author} {\bibfnamefont {V.}~\bibnamefont
  {Parigi}},\ }\bibfield  {title} {\bibinfo {title} {Continuous variables graph
  states shaped as complex networks: Optimization and manipulation},\ }\href
  {https://doi.org/10.3390/e22010026} {\bibfield  {journal} {\bibinfo
  {journal} {Entropy}\ }\textbf {\bibinfo {volume} {22}},\ \bibinfo {pages}
  {26} (\bibinfo {year} {2020})}\BibitemShut {NoStop}%
\bibitem [{\citenamefont {Arzani}\ \emph {et~al.}(2019)\citenamefont {Arzani},
  \citenamefont {Ferrini}, \citenamefont {Grosshans},\ and\ \citenamefont
  {Markham}}]{Arzani19}%
  \BibitemOpen
  \bibfield  {author} {\bibinfo {author} {\bibfnamefont {F.}~\bibnamefont
  {Arzani}}, \bibinfo {author} {\bibfnamefont {G.}~\bibnamefont {Ferrini}},
  \bibinfo {author} {\bibfnamefont {F.}~\bibnamefont {Grosshans}},\ and\
  \bibinfo {author} {\bibfnamefont {D.}~\bibnamefont {Markham}},\ }\bibfield
  {title} {\bibinfo {title} {Random coding for sharing bosonic quantum
  secrets},\ }\href {https://doi.org/10.1103/PhysRevA.100.022303} {\bibfield
  {journal} {\bibinfo  {journal} {Phys. Rev. A}\ }\textbf {\bibinfo {volume}
  {100}},\ \bibinfo {pages} {022303} (\bibinfo {year} {2019})}\BibitemShut
  {NoStop}%
\bibitem [{\citenamefont {Roslund}\ \emph {et~al.}(2014)\citenamefont
  {Roslund}, \citenamefont {Medeiros~de Ara\'ujo}, \citenamefont {Jiang},
  \citenamefont {Fabre},\ and\ \citenamefont {Treps}}]{Roslund14}%
  \BibitemOpen
  \bibfield  {author} {\bibinfo {author} {\bibfnamefont {J.}~\bibnamefont
  {Roslund}}, \bibinfo {author} {\bibfnamefont {R.}~\bibnamefont {Medeiros~de
  Ara\'ujo}}, \bibinfo {author} {\bibfnamefont {S.}~\bibnamefont {Jiang}},
  \bibinfo {author} {\bibfnamefont {C.}~\bibnamefont {Fabre}},\ and\ \bibinfo
  {author} {\bibfnamefont {N.}~\bibnamefont {Treps}},\ }\bibfield  {title}
  {\bibinfo {title} {Wavelength-multiplexed quantum networks with ultrafast
  frequency combs},\ }\href {https://doi.org/10.1038/nphoton.2013.340}
  {\bibfield  {journal} {\bibinfo  {journal} {Nat. Photon.}\ }\textbf {\bibinfo
  {volume} {8}},\ \bibinfo {pages} {109} (\bibinfo {year} {2014})}\BibitemShut
  {NoStop}%
\bibitem [{\citenamefont {Yokoyama}\ \emph {et~al.}(2013)\citenamefont
  {Yokoyama}, \citenamefont {Ukai}, \citenamefont {Armstrong}, \citenamefont
  {Sornphiphatphong}, \citenamefont {Kaji}, \citenamefont {Suzuki},
  \citenamefont {Yoshikawa}, \citenamefont {Yonezawa}, \citenamefont
  {Menicucci},\ and\ \citenamefont {Furusawa}}]{Yokoyama13}%
  \BibitemOpen
  \bibfield  {author} {\bibinfo {author} {\bibfnamefont {S.}~\bibnamefont
  {Yokoyama}}, \bibinfo {author} {\bibfnamefont {R.}~\bibnamefont {Ukai}},
  \bibinfo {author} {\bibfnamefont {S.~C.}\ \bibnamefont {Armstrong}}, \bibinfo
  {author} {\bibfnamefont {C.}~\bibnamefont {Sornphiphatphong}}, \bibinfo
  {author} {\bibfnamefont {T.}~\bibnamefont {Kaji}}, \bibinfo {author}
  {\bibfnamefont {S.}~\bibnamefont {Suzuki}}, \bibinfo {author} {\bibfnamefont
  {J.-i.}\ \bibnamefont {Yoshikawa}}, \bibinfo {author} {\bibfnamefont
  {H.}~\bibnamefont {Yonezawa}}, \bibinfo {author} {\bibfnamefont {N.~C.}\
  \bibnamefont {Menicucci}},\ and\ \bibinfo {author} {\bibfnamefont
  {A.}~\bibnamefont {Furusawa}},\ }\bibfield  {title} {\bibinfo {title}
  {Ultra-large-scale continuous-variable cluster states multiplexed in the time
  domain},\ }\href@noop {} {\bibfield  {journal} {\bibinfo  {journal} {Nature
  Photonics}\ }\textbf {\bibinfo {volume} {7}},\ \bibinfo {pages} {982}
  (\bibinfo {year} {2013})}\BibitemShut {NoStop}%
\bibitem [{\citenamefont {Menicucci}\ \emph {et~al.}(2006)\citenamefont
  {Menicucci}, \citenamefont {van Loock}, \citenamefont {Gu}, \citenamefont
  {Weedbrook}, \citenamefont {Ralph},\ and\ \citenamefont
  {Nielsen}}]{Menicucci06}%
  \BibitemOpen
  \bibfield  {author} {\bibinfo {author} {\bibfnamefont {N.~C.}\ \bibnamefont
  {Menicucci}}, \bibinfo {author} {\bibfnamefont {P.}~\bibnamefont {van
  Loock}}, \bibinfo {author} {\bibfnamefont {M.}~\bibnamefont {Gu}}, \bibinfo
  {author} {\bibfnamefont {C.}~\bibnamefont {Weedbrook}}, \bibinfo {author}
  {\bibfnamefont {T.~C.}\ \bibnamefont {Ralph}},\ and\ \bibinfo {author}
  {\bibfnamefont {M.~A.}\ \bibnamefont {Nielsen}},\ }\bibfield  {title}
  {\bibinfo {title} {Universal quantum computation with continuous-variable
  cluster states},\ }\href {https://doi.org/10.1103/PhysRevLett.97.110501}
  {\bibfield  {journal} {\bibinfo  {journal} {Phys. Rev. Lett.}\ }\textbf
  {\bibinfo {volume} {97}},\ \bibinfo {pages} {110501} (\bibinfo {year}
  {2006})}\BibitemShut {NoStop}%
\bibitem [{\citenamefont {Gu}\ \emph {et~al.}(2009)\citenamefont {Gu},
  \citenamefont {Weedbrook}, \citenamefont {Menicucci}, \citenamefont {Ralph},\
  and\ \citenamefont {van Loock}}]{Gu09}%
  \BibitemOpen
  \bibfield  {author} {\bibinfo {author} {\bibfnamefont {M.}~\bibnamefont
  {Gu}}, \bibinfo {author} {\bibfnamefont {C.}~\bibnamefont {Weedbrook}},
  \bibinfo {author} {\bibfnamefont {N.~C.}\ \bibnamefont {Menicucci}}, \bibinfo
  {author} {\bibfnamefont {T.~C.}\ \bibnamefont {Ralph}},\ and\ \bibinfo
  {author} {\bibfnamefont {P.}~\bibnamefont {van Loock}},\ }\bibfield  {title}
  {\bibinfo {title} {Quantum computing with continuous-variable clusters},\
  }\href {https://doi.org/10.1103/PhysRevA.79.062318} {\bibfield  {journal}
  {\bibinfo  {journal} {Phys. Rev. A}\ }\textbf {\bibinfo {volume} {79}},\
  \bibinfo {pages} {062318} (\bibinfo {year} {2009})}\BibitemShut {NoStop}%
\bibitem [{\citenamefont {Walschaers}\ \emph {et~al.}(2018)\citenamefont
  {Walschaers}, \citenamefont {Sarkar}, \citenamefont {Parigi},\ and\
  \citenamefont {Treps}}]{Walschaers18}%
  \BibitemOpen
  \bibfield  {author} {\bibinfo {author} {\bibfnamefont {M.}~\bibnamefont
  {Walschaers}}, \bibinfo {author} {\bibfnamefont {S.}~\bibnamefont {Sarkar}},
  \bibinfo {author} {\bibfnamefont {V.}~\bibnamefont {Parigi}},\ and\ \bibinfo
  {author} {\bibfnamefont {N.}~\bibnamefont {Treps}},\ }\bibfield  {title}
  {\bibinfo {title} {Tailoring non-gaussian continuous-variable graph states},\
  }\href {https://doi.org/10.1103/PhysRevLett.121.220501} {\bibfield  {journal}
  {\bibinfo  {journal} {Phys. Rev. Lett.}\ }\textbf {\bibinfo {volume} {121}},\
  \bibinfo {pages} {220501} (\bibinfo {year} {2018})}\BibitemShut {NoStop}%
\bibitem [{\citenamefont {Newman}(2018)}]{Newman18}%
  \BibitemOpen
  \bibfield  {author} {\bibinfo {author} {\bibfnamefont {M.~E.~J.}\
  \bibnamefont {Newman}},\ }\href@noop {} {\emph {\bibinfo {title} {Networks,
  second edition}}}\ (\bibinfo  {publisher} {Oxford University Press},\
  \bibinfo {year} {2018})\BibitemShut {NoStop}%
\bibitem [{\citenamefont {Barab\'{a}si}(2016)}]{Barabasi16}%
  \BibitemOpen
  \bibfield  {author} {\bibinfo {author} {\bibfnamefont {A.~L.}\ \bibnamefont
  {Barab\'{a}si}},\ }\href@noop {} {\emph {\bibinfo {title} {Networks
  science}}}\ (\bibinfo  {publisher} {Cambridge University Press},\ \bibinfo
  {year} {2016})\BibitemShut {NoStop}%
\bibitem [{\citenamefont {Albert}\ and\ \citenamefont
  {Barab\'asi}(2002)}]{AlbertRMP02}%
  \BibitemOpen
  \bibfield  {author} {\bibinfo {author} {\bibfnamefont {R.}~\bibnamefont
  {Albert}}\ and\ \bibinfo {author} {\bibfnamefont {A.-L.}\ \bibnamefont
  {Barab\'asi}},\ }\bibfield  {title} {\bibinfo {title} {Statistical mechanics
  of complex networks},\ }\href {https://doi.org/10.1103/RevModPhys.74.47}
  {\bibfield  {journal} {\bibinfo  {journal} {Rev. Mod. Phys.}\ }\textbf
  {\bibinfo {volume} {74}},\ \bibinfo {pages} {47} (\bibinfo {year}
  {2002})}\BibitemShut {NoStop}%
\bibitem [{\citenamefont {Dorogovtsev}\ \emph {et~al.}(2008)\citenamefont
  {Dorogovtsev}, \citenamefont {Goltsev},\ and\ \citenamefont
  {Mendes}}]{DorogovtsevRMP08}%
  \BibitemOpen
  \bibfield  {author} {\bibinfo {author} {\bibfnamefont {S.~N.}\ \bibnamefont
  {Dorogovtsev}}, \bibinfo {author} {\bibfnamefont {A.~V.}\ \bibnamefont
  {Goltsev}},\ and\ \bibinfo {author} {\bibfnamefont {J.~F.~F.}\ \bibnamefont
  {Mendes}},\ }\bibfield  {title} {\bibinfo {title} {Critical phenomena in
  complex networks},\ }\href {https://doi.org/10.1103/RevModPhys.80.1275}
  {\bibfield  {journal} {\bibinfo  {journal} {Rev. Mod. Phys.}\ }\textbf
  {\bibinfo {volume} {80}},\ \bibinfo {pages} {1275} (\bibinfo {year}
  {2008})}\BibitemShut {NoStop}%
\bibitem [{\citenamefont {Bianconi}(2015)}]{Bianconi15}%
  \BibitemOpen
  \bibfield  {author} {\bibinfo {author} {\bibfnamefont {G.}~\bibnamefont
  {Bianconi}},\ }\bibfield  {title} {\bibinfo {title} {Interdisciplinary and
  physics challenges of network theory},\ }\href
  {https://doi.org/10.1209/0295-5075/111/56001} {\bibfield  {journal} {\bibinfo
   {journal} {{EPL} (Europhysics Letters)}\ }\textbf {\bibinfo {volume}
  {111}},\ \bibinfo {pages} {56001} (\bibinfo {year} {2015})}\BibitemShut
  {NoStop}%
\bibitem [{\citenamefont {Halu}\ \emph {et~al.}(2013)\citenamefont {Halu},
  \citenamefont {Garnerone}, \citenamefont {Vezzani},\ and\ \citenamefont
  {Bianconi}}]{Halu13}%
  \BibitemOpen
  \bibfield  {author} {\bibinfo {author} {\bibfnamefont {A.}~\bibnamefont
  {Halu}}, \bibinfo {author} {\bibfnamefont {S.}~\bibnamefont {Garnerone}},
  \bibinfo {author} {\bibfnamefont {A.}~\bibnamefont {Vezzani}},\ and\ \bibinfo
  {author} {\bibfnamefont {G.}~\bibnamefont {Bianconi}},\ }\bibfield  {title}
  {\bibinfo {title} {Phase transition of light on complex quantum networks},\
  }\href {https://doi.org/10.1103/PhysRevE.87.022104} {\bibfield  {journal}
  {\bibinfo  {journal} {Phys. Rev. E}\ }\textbf {\bibinfo {volume} {87}},\
  \bibinfo {pages} {022104} (\bibinfo {year} {2013})}\BibitemShut {NoStop}%
\bibitem [{\citenamefont {Jahnke}\ \emph {et~al.}(2008)\citenamefont {Jahnke},
  \citenamefont {Kantelhardt}, \citenamefont {Berkovits},\ and\ \citenamefont
  {Havlin}}]{Jahnke08}%
  \BibitemOpen
  \bibfield  {author} {\bibinfo {author} {\bibfnamefont {L.}~\bibnamefont
  {Jahnke}}, \bibinfo {author} {\bibfnamefont {J.~W.}\ \bibnamefont
  {Kantelhardt}}, \bibinfo {author} {\bibfnamefont {R.}~\bibnamefont
  {Berkovits}},\ and\ \bibinfo {author} {\bibfnamefont {S.}~\bibnamefont
  {Havlin}},\ }\bibfield  {title} {\bibinfo {title} {Wave localization in
  complex networks with high clustering},\ }\href
  {https://doi.org/10.1103/PhysRevLett.101.175702} {\bibfield  {journal}
  {\bibinfo  {journal} {Phys. Rev. Lett.}\ }\textbf {\bibinfo {volume} {101}},\
  \bibinfo {pages} {175702} (\bibinfo {year} {2008})}\BibitemShut {NoStop}%
\bibitem [{\citenamefont {Burioni}\ \emph {et~al.}(2001)\citenamefont
  {Burioni}, \citenamefont {Cassi}, \citenamefont {Rasetti}, \citenamefont
  {Sodano},\ and\ \citenamefont {Vezzani}}]{Burioni01}%
  \BibitemOpen
  \bibfield  {author} {\bibinfo {author} {\bibfnamefont {R.}~\bibnamefont
  {Burioni}}, \bibinfo {author} {\bibfnamefont {D.}~\bibnamefont {Cassi}},
  \bibinfo {author} {\bibfnamefont {M.}~\bibnamefont {Rasetti}}, \bibinfo
  {author} {\bibfnamefont {P.}~\bibnamefont {Sodano}},\ and\ \bibinfo {author}
  {\bibfnamefont {A.}~\bibnamefont {Vezzani}},\ }\bibfield  {title} {\bibinfo
  {title} {Bose-einstein condensation on inhomogeneous complex networks},\
  }\href {https://doi.org/10.1088/0953-4075/34/23/314} {\bibfield  {journal}
  {\bibinfo  {journal} {Journal of Physics B: Atomic, Molecular and Optical
  Physics}\ }\textbf {\bibinfo {volume} {34}},\ \bibinfo {pages} {4697}
  (\bibinfo {year} {2001})}\BibitemShut {NoStop}%
\bibitem [{\citenamefont {M\"{u}lken}\ and\ \citenamefont
  {Blumen}(2011)}]{Mulken11}%
  \BibitemOpen
  \bibfield  {author} {\bibinfo {author} {\bibfnamefont {O.}~\bibnamefont
  {M\"{u}lken}}\ and\ \bibinfo {author} {\bibfnamefont {A.}~\bibnamefont
  {Blumen}},\ }\bibfield  {title} {\bibinfo {title} {Continuous-time quantum
  walks: Models for coherent transport on complex networks},\ }\href
  {https://doi.org/https://doi.org/10.1016/j.physrep.2011.01.002} {\bibfield
  {journal} {\bibinfo  {journal} {Physics Reports}\ }\textbf {\bibinfo {volume}
  {502}},\ \bibinfo {pages} {37 } (\bibinfo {year} {2011})}\BibitemShut
  {NoStop}%
\bibitem [{\citenamefont {Valdez}\ \emph {et~al.}(2017)\citenamefont {Valdez},
  \citenamefont {Jaschke}, \citenamefont {Vargas},\ and\ \citenamefont
  {Carr}}]{Valdez17}%
  \BibitemOpen
  \bibfield  {author} {\bibinfo {author} {\bibfnamefont {M.~A.}\ \bibnamefont
  {Valdez}}, \bibinfo {author} {\bibfnamefont {D.}~\bibnamefont {Jaschke}},
  \bibinfo {author} {\bibfnamefont {D.~L.}\ \bibnamefont {Vargas}},\ and\
  \bibinfo {author} {\bibfnamefont {L.~D.}\ \bibnamefont {Carr}},\ }\bibfield
  {title} {\bibinfo {title} {Quantifying complexity in quantum phase
  transitions via mutual information complex networks},\ }\href
  {https://doi.org/10.1103/PhysRevLett.119.225301} {\bibfield  {journal}
  {\bibinfo  {journal} {Phys. Rev. Lett.}\ }\textbf {\bibinfo {volume} {119}},\
  \bibinfo {pages} {225301} (\bibinfo {year} {2017})}\BibitemShut {NoStop}%
\bibitem [{\citenamefont {Nokkala}\ \emph
  {et~al.}(2018{\natexlab{b}})\citenamefont {Nokkala}, \citenamefont
  {Maniscalco},\ and\ \citenamefont {Piilo}}]{Nokkala18}%
  \BibitemOpen
  \bibfield  {author} {\bibinfo {author} {\bibfnamefont {J.}~\bibnamefont
  {Nokkala}}, \bibinfo {author} {\bibfnamefont {S.}~\bibnamefont
  {Maniscalco}},\ and\ \bibinfo {author} {\bibfnamefont {J.}~\bibnamefont
  {Piilo}},\ }\bibfield  {title} {\bibinfo {title} {Local probe for
  connectivity and coupling strength in quantum complex networks},\ }\href
  {https://doi.org/10.1038/s41598-018-30863-2} {\bibfield  {journal} {\bibinfo
  {journal} {Scientific Reports}\ }\textbf {\bibinfo {volume} {8}},\ \bibinfo
  {pages} {13010} (\bibinfo {year} {2018}{\natexlab{b}})}\BibitemShut {NoStop}%
\bibitem [{\citenamefont {Biamonte}\ \emph {et~al.}(2019)\citenamefont
  {Biamonte}, \citenamefont {Faccin},\ and\ \citenamefont
  {De~Domenico}}]{Biamonte19}%
  \BibitemOpen
  \bibfield  {author} {\bibinfo {author} {\bibfnamefont {J.}~\bibnamefont
  {Biamonte}}, \bibinfo {author} {\bibfnamefont {M.}~\bibnamefont {Faccin}},\
  and\ \bibinfo {author} {\bibfnamefont {M.}~\bibnamefont {De~Domenico}},\
  }\bibfield  {title} {\bibinfo {title} {Complex networks from classical to
  quantum},\ }\href {https://doi.org/10.1038/s42005-019-0152-6} {\bibfield
  {journal} {\bibinfo  {journal} {Communications Physics}\ }\textbf {\bibinfo
  {volume} {2}},\ \bibinfo {pages} {53} (\bibinfo {year} {2019})}\BibitemShut
  {NoStop}%
\bibitem [{\citenamefont {Cabot}\ \emph {et~al.}(2018)\citenamefont {Cabot},
  \citenamefont {Galve}, \citenamefont {Egu{\'\i}luz}, \citenamefont {Klemm},
  \citenamefont {Maniscalco},\ and\ \citenamefont {Zambrini}}]{Cabot18}%
  \BibitemOpen
  \bibfield  {author} {\bibinfo {author} {\bibfnamefont {A.}~\bibnamefont
  {Cabot}}, \bibinfo {author} {\bibfnamefont {F.}~\bibnamefont {Galve}},
  \bibinfo {author} {\bibfnamefont {V.~M.}\ \bibnamefont {Egu{\'\i}luz}},
  \bibinfo {author} {\bibfnamefont {K.}~\bibnamefont {Klemm}}, \bibinfo
  {author} {\bibfnamefont {S.}~\bibnamefont {Maniscalco}},\ and\ \bibinfo
  {author} {\bibfnamefont {R.}~\bibnamefont {Zambrini}},\ }\bibfield  {title}
  {\bibinfo {title} {Unveiling noiseless clusters in complex quantum
  networks},\ }\href {https://doi.org/10.1038/s41534-018-0108-9} {\bibfield
  {journal} {\bibinfo  {journal} {npj Quantum Information}\ }\textbf {\bibinfo
  {volume} {4}},\ \bibinfo {pages} {57} (\bibinfo {year} {2018})}\BibitemShut
  {NoStop}%
\bibitem [{\citenamefont {Chakraborty}\ \emph {et~al.}(2016)\citenamefont
  {Chakraborty}, \citenamefont {Novo}, \citenamefont {Ambainis},\ and\
  \citenamefont {Omar}}]{Chakraborty16}%
  \BibitemOpen
  \bibfield  {author} {\bibinfo {author} {\bibfnamefont {S.}~\bibnamefont
  {Chakraborty}}, \bibinfo {author} {\bibfnamefont {L.}~\bibnamefont {Novo}},
  \bibinfo {author} {\bibfnamefont {A.}~\bibnamefont {Ambainis}},\ and\
  \bibinfo {author} {\bibfnamefont {Y.}~\bibnamefont {Omar}},\ }\bibfield
  {title} {\bibinfo {title} {Spatial search by quantum walk is optimal for
  almost all graphs},\ }\href {https://doi.org/10.1103/PhysRevLett.116.100501}
  {\bibfield  {journal} {\bibinfo  {journal} {Phys. Rev. Lett.}\ }\textbf
  {\bibinfo {volume} {116}},\ \bibinfo {pages} {100501} (\bibinfo {year}
  {2016})}\BibitemShut {NoStop}%
\bibitem [{\citenamefont {Cuquet}\ and\ \citenamefont
  {Calsamiglia}(2009)}]{Cuquet09}%
  \BibitemOpen
  \bibfield  {author} {\bibinfo {author} {\bibfnamefont {M.}~\bibnamefont
  {Cuquet}}\ and\ \bibinfo {author} {\bibfnamefont {J.}~\bibnamefont
  {Calsamiglia}},\ }\bibfield  {title} {\bibinfo {title} {Entanglement
  percolation in quantum complex networks},\ }\href
  {https://doi.org/10.1103/PhysRevLett.103.240503} {\bibfield  {journal}
  {\bibinfo  {journal} {Phys. Rev. Lett.}\ }\textbf {\bibinfo {volume} {103}},\
  \bibinfo {pages} {240503} (\bibinfo {year} {2009})}\BibitemShut {NoStop}%
\bibitem [{\citenamefont {Faccin}\ \emph {et~al.}(2014)\citenamefont {Faccin},
  \citenamefont {Migda\l{}}, \citenamefont {Johnson}, \citenamefont
  {Bergholm},\ and\ \citenamefont {Biamonte}}]{Faccin14}%
  \BibitemOpen
  \bibfield  {author} {\bibinfo {author} {\bibfnamefont {M.}~\bibnamefont
  {Faccin}}, \bibinfo {author} {\bibfnamefont {P.}~\bibnamefont {Migda\l{}}},
  \bibinfo {author} {\bibfnamefont {T.~H.}\ \bibnamefont {Johnson}}, \bibinfo
  {author} {\bibfnamefont {V.}~\bibnamefont {Bergholm}},\ and\ \bibinfo
  {author} {\bibfnamefont {J.~D.}\ \bibnamefont {Biamonte}},\ }\bibfield
  {title} {\bibinfo {title} {Community detection in quantum complex networks},\
  }\href {https://doi.org/10.1103/PhysRevX.4.041012} {\bibfield  {journal}
  {\bibinfo  {journal} {Phys. Rev. X}\ }\textbf {\bibinfo {volume} {4}},\
  \bibinfo {pages} {041012} (\bibinfo {year} {2014})}\BibitemShut {NoStop}%
\bibitem [{\citenamefont {Sundar}\ \emph {et~al.}(2018)\citenamefont {Sundar},
  \citenamefont {Valdez}, \citenamefont {Carr},\ and\ \citenamefont
  {Hazzard}}]{sundar2018complex}%
  \BibitemOpen
  \bibfield  {author} {\bibinfo {author} {\bibfnamefont {B.}~\bibnamefont
  {Sundar}}, \bibinfo {author} {\bibfnamefont {M.~A.}\ \bibnamefont {Valdez}},
  \bibinfo {author} {\bibfnamefont {L.~D.}\ \bibnamefont {Carr}},\ and\
  \bibinfo {author} {\bibfnamefont {K.~R.}\ \bibnamefont {Hazzard}},\
  }\bibfield  {title} {\bibinfo {title} {Complex-network description of thermal
  quantum states in the ising spin chain},\ }\href@noop {} {\bibfield
  {journal} {\bibinfo  {journal} {Physical Review A}\ }\textbf {\bibinfo
  {volume} {97}},\ \bibinfo {pages} {052320} (\bibinfo {year}
  {2018})}\BibitemShut {NoStop}%
\bibitem [{\citenamefont {Bu{\v c}a}\ \emph {et~al.}(2019)\citenamefont {Bu{\v
  c}a}, \citenamefont {Tindall},\ and\ \citenamefont {Jaksch}}]{Buca19}%
  \BibitemOpen
  \bibfield  {author} {\bibinfo {author} {\bibfnamefont {B.}~\bibnamefont
  {Bu{\v c}a}}, \bibinfo {author} {\bibfnamefont {J.}~\bibnamefont {Tindall}},\
  and\ \bibinfo {author} {\bibfnamefont {D.}~\bibnamefont {Jaksch}},\
  }\bibfield  {title} {\bibinfo {title} {Non-stationary coherent quantum
  many-body dynamics through dissipation},\ }\href@noop {} {\bibfield
  {journal} {\bibinfo  {journal} {Nature Communications}\ }\textbf {\bibinfo
  {volume} {10}},\ \bibinfo {pages} {1730} (\bibinfo {year}
  {2019})}\BibitemShut {NoStop}%
\bibitem [{\citenamefont {Sokolov}\ \emph {et~al.}(2020)\citenamefont
  {Sokolov}, \citenamefont {Rossi}, \citenamefont {Garc\'{\i}a-P\'erez},\ and\
  \citenamefont {Maniscalco}}]{Sokolov2020}%
  \BibitemOpen
  \bibfield  {author} {\bibinfo {author} {\bibfnamefont {B.}~\bibnamefont
  {Sokolov}}, \bibinfo {author} {\bibfnamefont {M.~A.~C.}\ \bibnamefont
  {Rossi}}, \bibinfo {author} {\bibfnamefont {G.}~\bibnamefont
  {Garc\'{\i}a-P\'erez}},\ and\ \bibinfo {author} {\bibfnamefont
  {S.}~\bibnamefont {Maniscalco}},\ }\href@noop {} {\bibinfo {title} {Emergent
  entanglement structures and self-similarity in quantum spin chains}}
  (\bibinfo {year} {2020}),\ \Eprint {https://arxiv.org/abs/2007.06989}
  {arXiv:2007.06989 [quant-ph]} \BibitemShut {NoStop}%
\bibitem [{\citenamefont {Hillberry}\ \emph {et~al.}(2020)\citenamefont
  {Hillberry}, \citenamefont {Jones}, \citenamefont {Vargas}, \citenamefont
  {Rall}, \citenamefont {Halpern}, \citenamefont {Bao}, \citenamefont
  {Notarnicola}, \citenamefont {Montangero},\ and\ \citenamefont
  {Carr}}]{hillberry2020entangled}%
  \BibitemOpen
  \bibfield  {author} {\bibinfo {author} {\bibfnamefont {L.~E.}\ \bibnamefont
  {Hillberry}}, \bibinfo {author} {\bibfnamefont {M.~T.}\ \bibnamefont
  {Jones}}, \bibinfo {author} {\bibfnamefont {D.~L.}\ \bibnamefont {Vargas}},
  \bibinfo {author} {\bibfnamefont {P.}~\bibnamefont {Rall}}, \bibinfo {author}
  {\bibfnamefont {N.~Y.}\ \bibnamefont {Halpern}}, \bibinfo {author}
  {\bibfnamefont {N.}~\bibnamefont {Bao}}, \bibinfo {author} {\bibfnamefont
  {S.}~\bibnamefont {Notarnicola}}, \bibinfo {author} {\bibfnamefont
  {S.}~\bibnamefont {Montangero}},\ and\ \bibinfo {author} {\bibfnamefont
  {L.~D.}\ \bibnamefont {Carr}},\ }\href@noop {} {\bibinfo {title} {Entangled
  quantum cellular automata, physical complexity, and goldilocks rules}}
  (\bibinfo {year} {2020}),\ \Eprint {https://arxiv.org/abs/2005.01763}
  {arXiv:2005.01763 [quant-ph]} \BibitemShut {NoStop}%
\bibitem [{\citenamefont {Kimble}(2008)}]{Kimble08}%
  \BibitemOpen
  \bibfield  {author} {\bibinfo {author} {\bibfnamefont {H.~J.}\ \bibnamefont
  {Kimble}},\ }\bibfield  {title} {\bibinfo {title} {The quantum internet},\
  }\href {https://doi.org/10.1038/nature07127} {\bibfield  {journal} {\bibinfo
  {journal} {Nature}\ }\textbf {\bibinfo {volume} {453}},\ \bibinfo {pages}
  {1023} (\bibinfo {year} {2008})}\BibitemShut {NoStop}%
\bibitem [{\citenamefont {Childs}(2009)}]{PhysRevLett.102.180501}%
  \BibitemOpen
  \bibfield  {author} {\bibinfo {author} {\bibfnamefont {A.~M.}\ \bibnamefont
  {Childs}},\ }\bibfield  {title} {\bibinfo {title} {Universal computation by
  quantum walk},\ }\href {https://doi.org/10.1103/PhysRevLett.102.180501}
  {\bibfield  {journal} {\bibinfo  {journal} {Phys. Rev. Lett.}\ }\textbf
  {\bibinfo {volume} {102}},\ \bibinfo {pages} {180501} (\bibinfo {year}
  {2009})}\BibitemShut {NoStop}%
\bibitem [{\citenamefont {Mohseni}\ \emph {et~al.}(2008)\citenamefont
  {Mohseni}, \citenamefont {Rebentrost}, \citenamefont {Lloyd},\ and\
  \citenamefont {Aspuru-Guzik}}]{Mohseni2008}%
  \BibitemOpen
  \bibfield  {author} {\bibinfo {author} {\bibfnamefont {M.}~\bibnamefont
  {Mohseni}}, \bibinfo {author} {\bibfnamefont {P.}~\bibnamefont {Rebentrost}},
  \bibinfo {author} {\bibfnamefont {S.}~\bibnamefont {Lloyd}},\ and\ \bibinfo
  {author} {\bibfnamefont {A.}~\bibnamefont {Aspuru-Guzik}},\ }\bibfield
  {title} {\bibinfo {title} {Environment-assisted quantum walks in
  photosynthetic energy transfer},\ }\href {https://doi.org/10.1063/1.3002335}
  {\bibfield  {journal} {\bibinfo  {journal} {The Journal of Chemical Physics}\
  }\textbf {\bibinfo {volume} {129}},\ \bibinfo {pages} {174106} (\bibinfo
  {year} {2008})},\ \Eprint
  {https://arxiv.org/abs/https://doi.org/10.1063/1.3002335}
  {https://doi.org/10.1063/1.3002335} \BibitemShut {NoStop}%
\bibitem [{\citenamefont {Plenio}\ and\ \citenamefont
  {Huelga}(2008)}]{Plenio_2008}%
  \BibitemOpen
  \bibfield  {author} {\bibinfo {author} {\bibfnamefont {M.~B.}\ \bibnamefont
  {Plenio}}\ and\ \bibinfo {author} {\bibfnamefont {S.~F.}\ \bibnamefont
  {Huelga}},\ }\bibfield  {title} {\bibinfo {title} {Dephasing-assisted
  transport: quantum networks and biomolecules},\ }\href
  {https://doi.org/10.1088/1367-2630/10/11/113019} {\bibfield  {journal}
  {\bibinfo  {journal} {New Journal of Physics}\ }\textbf {\bibinfo {volume}
  {10}},\ \bibinfo {pages} {113019} (\bibinfo {year} {2008})}\BibitemShut
  {NoStop}%
\bibitem [{\citenamefont {Walschaers}\ \emph
  {et~al.}(2016{\natexlab{b}})\citenamefont {Walschaers}, \citenamefont
  {Schlawin}, \citenamefont {Wellens},\ and\ \citenamefont
  {Buchleitner}}]{annurev2016}%
  \BibitemOpen
  \bibfield  {author} {\bibinfo {author} {\bibfnamefont {M.}~\bibnamefont
  {Walschaers}}, \bibinfo {author} {\bibfnamefont {F.}~\bibnamefont
  {Schlawin}}, \bibinfo {author} {\bibfnamefont {T.}~\bibnamefont {Wellens}},\
  and\ \bibinfo {author} {\bibfnamefont {A.}~\bibnamefont {Buchleitner}},\
  }\bibfield  {title} {\bibinfo {title} {Quantum transport on disordered and
  noisy networks: An interplay of structural complexity and uncertainty},\
  }\href {https://doi.org/10.1146/annurev-conmatphys-031115-011327} {\bibfield
  {journal} {\bibinfo  {journal} {Annual Review of Condensed Matter Physics}\
  }\textbf {\bibinfo {volume} {7}},\ \bibinfo {pages} {223} (\bibinfo {year}
  {2016}{\natexlab{b}})},\ \Eprint
  {https://arxiv.org/abs/https://doi.org/10.1146/annurev-conmatphys-031115-011327}
  {https://doi.org/10.1146/annurev-conmatphys-031115-011327} \BibitemShut
  {NoStop}%
\bibitem [{\citenamefont {Awschalom}\ \emph {et~al.}(2019)\citenamefont
  {Awschalom}, \citenamefont {Berggren}, \citenamefont {Bernien}, \citenamefont
  {Bhave}, \citenamefont {Carr}, \citenamefont {Davids}, \citenamefont
  {Economou}, \citenamefont {Englund}, \citenamefont {Faraon}, \citenamefont
  {Fejer}, \citenamefont {Guha}, \citenamefont {Gustafsson}, \citenamefont
  {Hu}, \citenamefont {Jiang}, \citenamefont {Kim}, \citenamefont {Korzh},
  \citenamefont {Kumar}, \citenamefont {Kwiat}, \citenamefont {Lončar},
  \citenamefont {Lukin}, \citenamefont {Miller}, \citenamefont {Monroe},
  \citenamefont {Nam}, \citenamefont {Narang}, \citenamefont {Orcutt},
  \citenamefont {Raymer}, \citenamefont {Safavi-Naeini}, \citenamefont
  {Spiropulu}, \citenamefont {Srinivasan}, \citenamefont {Sun}, \citenamefont
  {Vučković}, \citenamefont {Waks}, \citenamefont {Walsworth}, \citenamefont
  {Weiner},\ and\ \citenamefont {Zhang}}]{awschalom2019development}%
  \BibitemOpen
  \bibfield  {author} {\bibinfo {author} {\bibfnamefont {D.}~\bibnamefont
  {Awschalom}}, \bibinfo {author} {\bibfnamefont {K.~K.}\ \bibnamefont
  {Berggren}}, \bibinfo {author} {\bibfnamefont {H.}~\bibnamefont {Bernien}},
  \bibinfo {author} {\bibfnamefont {S.}~\bibnamefont {Bhave}}, \bibinfo
  {author} {\bibfnamefont {L.~D.}\ \bibnamefont {Carr}}, \bibinfo {author}
  {\bibfnamefont {P.}~\bibnamefont {Davids}}, \bibinfo {author} {\bibfnamefont
  {S.~E.}\ \bibnamefont {Economou}}, \bibinfo {author} {\bibfnamefont
  {D.}~\bibnamefont {Englund}}, \bibinfo {author} {\bibfnamefont
  {A.}~\bibnamefont {Faraon}}, \bibinfo {author} {\bibfnamefont
  {M.}~\bibnamefont {Fejer}}, \bibinfo {author} {\bibfnamefont
  {S.}~\bibnamefont {Guha}}, \bibinfo {author} {\bibfnamefont {M.~V.}\
  \bibnamefont {Gustafsson}}, \bibinfo {author} {\bibfnamefont
  {E.}~\bibnamefont {Hu}}, \bibinfo {author} {\bibfnamefont {L.}~\bibnamefont
  {Jiang}}, \bibinfo {author} {\bibfnamefont {J.}~\bibnamefont {Kim}}, \bibinfo
  {author} {\bibfnamefont {B.}~\bibnamefont {Korzh}}, \bibinfo {author}
  {\bibfnamefont {P.}~\bibnamefont {Kumar}}, \bibinfo {author} {\bibfnamefont
  {P.~G.}\ \bibnamefont {Kwiat}}, \bibinfo {author} {\bibfnamefont
  {M.}~\bibnamefont {Lončar}}, \bibinfo {author} {\bibfnamefont {M.~D.}\
  \bibnamefont {Lukin}}, \bibinfo {author} {\bibfnamefont {D.~A.~B.}\
  \bibnamefont {Miller}}, \bibinfo {author} {\bibfnamefont {C.}~\bibnamefont
  {Monroe}}, \bibinfo {author} {\bibfnamefont {S.~W.}\ \bibnamefont {Nam}},
  \bibinfo {author} {\bibfnamefont {P.}~\bibnamefont {Narang}}, \bibinfo
  {author} {\bibfnamefont {J.~S.}\ \bibnamefont {Orcutt}}, \bibinfo {author}
  {\bibfnamefont {M.~G.}\ \bibnamefont {Raymer}}, \bibinfo {author}
  {\bibfnamefont {A.~H.}\ \bibnamefont {Safavi-Naeini}}, \bibinfo {author}
  {\bibfnamefont {M.}~\bibnamefont {Spiropulu}}, \bibinfo {author}
  {\bibfnamefont {K.}~\bibnamefont {Srinivasan}}, \bibinfo {author}
  {\bibfnamefont {S.}~\bibnamefont {Sun}}, \bibinfo {author} {\bibfnamefont
  {J.}~\bibnamefont {Vučković}}, \bibinfo {author} {\bibfnamefont
  {E.}~\bibnamefont {Waks}}, \bibinfo {author} {\bibfnamefont {R.}~\bibnamefont
  {Walsworth}}, \bibinfo {author} {\bibfnamefont {A.~M.}\ \bibnamefont
  {Weiner}},\ and\ \bibinfo {author} {\bibfnamefont {Z.}~\bibnamefont
  {Zhang}},\ }\href@noop {} {\bibinfo {title} {Development of quantum
  interconnects for next-generation information technologies}} (\bibinfo {year}
  {2019}),\ \Eprint {https://arxiv.org/abs/1912.06642} {arXiv:1912.06642
  [quant-ph]} \BibitemShut {NoStop}%
\bibitem [{\citenamefont {Altman}\ \emph {et~al.}(2021)\citenamefont {Altman},
  \citenamefont {Brown}, \citenamefont {Carleo}, \citenamefont {Carr},
  \citenamefont {Demler}, \citenamefont {Chin}, \citenamefont {DeMarco},
  \citenamefont {Economou}, \citenamefont {Eriksson}, \citenamefont {Fu},
  \citenamefont {Greiner}, \citenamefont {Hazzard}, \citenamefont {Hulet},
  \citenamefont {Koll\'ar}, \citenamefont {Lev}, \citenamefont {Lukin},
  \citenamefont {Ma}, \citenamefont {Mi}, \citenamefont {Misra}, \citenamefont
  {Monroe}, \citenamefont {Murch}, \citenamefont {Nazario}, \citenamefont {Ni},
  \citenamefont {Potter}, \citenamefont {Roushan}, \citenamefont {Saffman},
  \citenamefont {Schleier-Smith}, \citenamefont {Siddiqi}, \citenamefont
  {Simmonds}, \citenamefont {Singh}, \citenamefont {Spielman}, \citenamefont
  {Temme}, \citenamefont {Weiss}, \citenamefont {Vu\ifmmode \check{c}\else
  \v{c}\fi{}kovi\ifmmode~\acute{c}\else \'{c}\fi{}}, \citenamefont
  {Vuleti\ifmmode~\acute{c}\else \'{c}\fi{}}, \citenamefont {Ye},\ and\
  \citenamefont {Zwierlein}}]{altman2019quantum}%
  \BibitemOpen
  \bibfield  {author} {\bibinfo {author} {\bibfnamefont {E.}~\bibnamefont
  {Altman}}, \bibinfo {author} {\bibfnamefont {K.~R.}\ \bibnamefont {Brown}},
  \bibinfo {author} {\bibfnamefont {G.}~\bibnamefont {Carleo}}, \bibinfo
  {author} {\bibfnamefont {L.~D.}\ \bibnamefont {Carr}}, \bibinfo {author}
  {\bibfnamefont {E.}~\bibnamefont {Demler}}, \bibinfo {author} {\bibfnamefont
  {C.}~\bibnamefont {Chin}}, \bibinfo {author} {\bibfnamefont {B.}~\bibnamefont
  {DeMarco}}, \bibinfo {author} {\bibfnamefont {S.~E.}\ \bibnamefont
  {Economou}}, \bibinfo {author} {\bibfnamefont {M.~A.}\ \bibnamefont
  {Eriksson}}, \bibinfo {author} {\bibfnamefont {K.-M.~C.}\ \bibnamefont {Fu}},
  \bibinfo {author} {\bibfnamefont {M.}~\bibnamefont {Greiner}}, \bibinfo
  {author} {\bibfnamefont {K.~R.}\ \bibnamefont {Hazzard}}, \bibinfo {author}
  {\bibfnamefont {R.~G.}\ \bibnamefont {Hulet}}, \bibinfo {author}
  {\bibfnamefont {A.~J.}\ \bibnamefont {Koll\'ar}}, \bibinfo {author}
  {\bibfnamefont {B.~L.}\ \bibnamefont {Lev}}, \bibinfo {author} {\bibfnamefont
  {M.~D.}\ \bibnamefont {Lukin}}, \bibinfo {author} {\bibfnamefont
  {R.}~\bibnamefont {Ma}}, \bibinfo {author} {\bibfnamefont {X.}~\bibnamefont
  {Mi}}, \bibinfo {author} {\bibfnamefont {S.}~\bibnamefont {Misra}}, \bibinfo
  {author} {\bibfnamefont {C.}~\bibnamefont {Monroe}}, \bibinfo {author}
  {\bibfnamefont {K.}~\bibnamefont {Murch}}, \bibinfo {author} {\bibfnamefont
  {Z.}~\bibnamefont {Nazario}}, \bibinfo {author} {\bibfnamefont {K.-K.}\
  \bibnamefont {Ni}}, \bibinfo {author} {\bibfnamefont {A.~C.}\ \bibnamefont
  {Potter}}, \bibinfo {author} {\bibfnamefont {P.}~\bibnamefont {Roushan}},
  \bibinfo {author} {\bibfnamefont {M.}~\bibnamefont {Saffman}}, \bibinfo
  {author} {\bibfnamefont {M.}~\bibnamefont {Schleier-Smith}}, \bibinfo
  {author} {\bibfnamefont {I.}~\bibnamefont {Siddiqi}}, \bibinfo {author}
  {\bibfnamefont {R.}~\bibnamefont {Simmonds}}, \bibinfo {author}
  {\bibfnamefont {M.}~\bibnamefont {Singh}}, \bibinfo {author} {\bibfnamefont
  {I.}~\bibnamefont {Spielman}}, \bibinfo {author} {\bibfnamefont
  {K.}~\bibnamefont {Temme}}, \bibinfo {author} {\bibfnamefont {D.~S.}\
  \bibnamefont {Weiss}}, \bibinfo {author} {\bibfnamefont {J.}~\bibnamefont
  {Vu\ifmmode \check{c}\else \v{c}\fi{}kovi\ifmmode~\acute{c}\else
  \'{c}\fi{}}}, \bibinfo {author} {\bibfnamefont {V.}~\bibnamefont
  {Vuleti\ifmmode~\acute{c}\else \'{c}\fi{}}}, \bibinfo {author} {\bibfnamefont
  {J.}~\bibnamefont {Ye}},\ and\ \bibinfo {author} {\bibfnamefont
  {M.}~\bibnamefont {Zwierlein}},\ }\bibfield  {title} {\bibinfo {title}
  {Quantum simulators: Architectures and opportunities},\ }\href
  {https://doi.org/10.1103/PRXQuantum.2.017003} {\bibfield  {journal} {\bibinfo
   {journal} {PRX Quantum}\ }\textbf {\bibinfo {volume} {2}},\ \bibinfo {pages}
  {017003} (\bibinfo {year} {2021})}\BibitemShut {NoStop}%
\bibitem [{Note1()}]{Note1}%
  \BibitemOpen
  \bibinfo {note} {Cluster is sometimes reserved for graphs allowing for
  universal quantum computing. In this work, however, we use the terms \protect
  \emph {cluster state} and \protect \emph {graph state} as
  synonyms.}\BibitemShut {Stop}%
\bibitem [{\citenamefont {Raussendorf}\ and\ \citenamefont
  {Briegel}(2001)}]{Raussendorf01}%
  \BibitemOpen
  \bibfield  {author} {\bibinfo {author} {\bibfnamefont {R.}~\bibnamefont
  {Raussendorf}}\ and\ \bibinfo {author} {\bibfnamefont {H.~J.}\ \bibnamefont
  {Briegel}},\ }\bibfield  {title} {\bibinfo {title} {A one-way quantum
  computer},\ }\href {https://doi.org/10.1103/PhysRevLett.86.5188} {\bibfield
  {journal} {\bibinfo  {journal} {Phys. Rev. Lett.}\ }\textbf {\bibinfo
  {volume} {86}},\ \bibinfo {pages} {5188} (\bibinfo {year}
  {2001})}\BibitemShut {NoStop}%
\bibitem [{\citenamefont {Raussendorf}\ \emph {et~al.}(2003)\citenamefont
  {Raussendorf}, \citenamefont {Browne},\ and\ \citenamefont
  {Briegel}}]{Raussendorf03}%
  \BibitemOpen
  \bibfield  {author} {\bibinfo {author} {\bibfnamefont {R.}~\bibnamefont
  {Raussendorf}}, \bibinfo {author} {\bibfnamefont {D.~E.}\ \bibnamefont
  {Browne}},\ and\ \bibinfo {author} {\bibfnamefont {H.~J.}\ \bibnamefont
  {Briegel}},\ }\bibfield  {title} {\bibinfo {title} {Measurement-based quantum
  computation on cluster states},\ }\href
  {https://doi.org/10.1103/PhysRevA.68.022312} {\bibfield  {journal} {\bibinfo
  {journal} {Phys. Rev. A}\ }\textbf {\bibinfo {volume} {68}},\ \bibinfo
  {pages} {022312} (\bibinfo {year} {2003})}\BibitemShut {NoStop}%
\bibitem [{\citenamefont {Fabre}\ and\ \citenamefont
  {Treps}(2020)}]{RevModPhys.92.035005}%
  \BibitemOpen
  \bibfield  {author} {\bibinfo {author} {\bibfnamefont {C.}~\bibnamefont
  {Fabre}}\ and\ \bibinfo {author} {\bibfnamefont {N.}~\bibnamefont {Treps}},\
  }\bibfield  {title} {\bibinfo {title} {Modes and states in quantum optics},\
  }\href {https://doi.org/10.1103/RevModPhys.92.035005} {\bibfield  {journal}
  {\bibinfo  {journal} {Rev. Mod. Phys.}\ }\textbf {\bibinfo {volume} {92}},\
  \bibinfo {pages} {035005} (\bibinfo {year} {2020})}\BibitemShut {NoStop}%
\bibitem [{\citenamefont {Verbeure}(2011)}]{Verbeure2011}%
  \BibitemOpen
  \bibfield  {author} {\bibinfo {author} {\bibfnamefont {A.~F.}\ \bibnamefont
  {Verbeure}},\ }\bibinfo {title} {Bose systems},\ in\ \href
  {https://doi.org/10.1007/978-0-85729-109-7_2} {\emph {\bibinfo {booktitle}
  {Many-Body Boson Systems: Half a Century Later}}}\ (\bibinfo  {publisher}
  {Springer London},\ \bibinfo {address} {London},\ \bibinfo {year} {2011})\
  pp.\ \bibinfo {pages} {7--26}\BibitemShut {NoStop}%
\bibitem [{\citenamefont {Menicucci}\ \emph {et~al.}(2011)\citenamefont
  {Menicucci}, \citenamefont {Flammia},\ and\ \citenamefont {van
  Loock}}]{Meni11}%
  \BibitemOpen
  \bibfield  {author} {\bibinfo {author} {\bibfnamefont {N.~C.}\ \bibnamefont
  {Menicucci}}, \bibinfo {author} {\bibfnamefont {S.~T.}\ \bibnamefont
  {Flammia}},\ and\ \bibinfo {author} {\bibfnamefont {P.}~\bibnamefont {van
  Loock}},\ }\bibfield  {title} {\bibinfo {title} {Graphical calculus for
  gaussian pure states},\ }\href {https://doi.org/10.1103/PhysRevA.83.042335}
  {\bibfield  {journal} {\bibinfo  {journal} {Phys. Rev. A}\ }\textbf {\bibinfo
  {volume} {83}},\ \bibinfo {pages} {042335} (\bibinfo {year}
  {2011})}\BibitemShut {NoStop}%
\bibitem [{\citenamefont {Briegel}\ and\ \citenamefont
  {Raussendorf}(2001)}]{Briegel01}%
  \BibitemOpen
  \bibfield  {author} {\bibinfo {author} {\bibfnamefont {H.~J.}\ \bibnamefont
  {Briegel}}\ and\ \bibinfo {author} {\bibfnamefont {R.}~\bibnamefont
  {Raussendorf}},\ }\bibfield  {title} {\bibinfo {title} {Persistent
  entanglement in arrays of interacting particles},\ }\href
  {https://doi.org/10.1103/PhysRevLett.86.910} {\bibfield  {journal} {\bibinfo
  {journal} {Phys. Rev. Lett.}\ }\textbf {\bibinfo {volume} {86}},\ \bibinfo
  {pages} {910} (\bibinfo {year} {2001})}\BibitemShut {NoStop}%
\bibitem [{\citenamefont {Van~den Nest}\ \emph {et~al.}(2006)\citenamefont
  {Van~den Nest}, \citenamefont {Miyake}, \citenamefont {D\"ur},\ and\
  \citenamefont {Briegel}}]{VandenNest06}%
  \BibitemOpen
  \bibfield  {author} {\bibinfo {author} {\bibfnamefont {M.}~\bibnamefont
  {Van~den Nest}}, \bibinfo {author} {\bibfnamefont {A.}~\bibnamefont
  {Miyake}}, \bibinfo {author} {\bibfnamefont {W.}~\bibnamefont {D\"ur}},\ and\
  \bibinfo {author} {\bibfnamefont {H.~J.}\ \bibnamefont {Briegel}},\
  }\bibfield  {title} {\bibinfo {title} {Universal resources for
  measurement-based quantum computation},\ }\href
  {https://doi.org/10.1103/PhysRevLett.97.150504} {\bibfield  {journal}
  {\bibinfo  {journal} {Phys. Rev. Lett.}\ }\textbf {\bibinfo {volume} {97}},\
  \bibinfo {pages} {150504} (\bibinfo {year} {2006})}\BibitemShut {NoStop}%
\bibitem [{\citenamefont {Walschaers}\ \emph
  {et~al.}(2017{\natexlab{b}})\citenamefont {Walschaers}, \citenamefont
  {Fabre}, \citenamefont {Parigi},\ and\ \citenamefont {Treps}}]{Walschaers17}%
  \BibitemOpen
  \bibfield  {author} {\bibinfo {author} {\bibfnamefont {M.}~\bibnamefont
  {Walschaers}}, \bibinfo {author} {\bibfnamefont {C.}~\bibnamefont {Fabre}},
  \bibinfo {author} {\bibfnamefont {V.}~\bibnamefont {Parigi}},\ and\ \bibinfo
  {author} {\bibfnamefont {N.}~\bibnamefont {Treps}},\ }\bibfield  {title}
  {\bibinfo {title} {Entanglement and wigner function negativity of multimode
  non-gaussian states},\ }\href
  {https://doi.org/10.1103/PhysRevLett.119.183601} {\bibfield  {journal}
  {\bibinfo  {journal} {Phys. Rev. Lett.}\ }\textbf {\bibinfo {volume} {119}},\
  \bibinfo {pages} {183601} (\bibinfo {year} {2017}{\natexlab{b}})}\BibitemShut
  {NoStop}%
\bibitem [{\citenamefont {Walschaers}\ and\ \citenamefont
  {Treps}(2020)}]{Walschaers19s}%
  \BibitemOpen
  \bibfield  {author} {\bibinfo {author} {\bibfnamefont {M.}~\bibnamefont
  {Walschaers}}\ and\ \bibinfo {author} {\bibfnamefont {N.}~\bibnamefont
  {Treps}},\ }\bibfield  {title} {\bibinfo {title} {Remote generation of wigner
  negativity through einstein-podolsky-rosen steering},\ }\href
  {https://doi.org/10.1103/PhysRevLett.124.150501} {\bibfield  {journal}
  {\bibinfo  {journal} {Phys. Rev. Lett.}\ }\textbf {\bibinfo {volume} {124}},\
  \bibinfo {pages} {150501} (\bibinfo {year} {2020})}\BibitemShut {NoStop}%
\bibitem [{\citenamefont {Ra}\ \emph {et~al.}(2017)\citenamefont {Ra},
  \citenamefont {Jacquard}, \citenamefont {Dufour}, \citenamefont {Fabre},\
  and\ \citenamefont {Treps}}]{Ra17}%
  \BibitemOpen
  \bibfield  {author} {\bibinfo {author} {\bibfnamefont {Y.-S.}\ \bibnamefont
  {Ra}}, \bibinfo {author} {\bibfnamefont {C.}~\bibnamefont {Jacquard}},
  \bibinfo {author} {\bibfnamefont {A.}~\bibnamefont {Dufour}}, \bibinfo
  {author} {\bibfnamefont {C.}~\bibnamefont {Fabre}},\ and\ \bibinfo {author}
  {\bibfnamefont {N.}~\bibnamefont {Treps}},\ }\bibfield  {title} {\bibinfo
  {title} {Tomography of a mode-tunable coherent single-photon subtractor},\
  }\href {https://doi.org/10.1103/PhysRevX.7.031012} {\bibfield  {journal}
  {\bibinfo  {journal} {Phys. Rev. X}\ }\textbf {\bibinfo {volume} {7}},\
  \bibinfo {pages} {031012} (\bibinfo {year} {2017})}\BibitemShut {NoStop}%
\bibitem [{\citenamefont {Wenger}\ \emph {et~al.}(2004)\citenamefont {Wenger},
  \citenamefont {Tualle-Brouri},\ and\ \citenamefont {Grangier}}]{Wenger04}%
  \BibitemOpen
  \bibfield  {author} {\bibinfo {author} {\bibfnamefont {J.}~\bibnamefont
  {Wenger}}, \bibinfo {author} {\bibfnamefont {R.}~\bibnamefont
  {Tualle-Brouri}},\ and\ \bibinfo {author} {\bibfnamefont {P.}~\bibnamefont
  {Grangier}},\ }\bibfield  {title} {\bibinfo {title} {Non-gaussian statistics
  from individual pulses of squeezed light},\ }\href
  {https://doi.org/10.1103/PhysRevLett.92.153601} {\bibfield  {journal}
  {\bibinfo  {journal} {Phys. Rev. Lett.}\ }\textbf {\bibinfo {volume} {92}},\
  \bibinfo {pages} {153601} (\bibinfo {year} {2004})}\BibitemShut {NoStop}%
\bibitem [{\citenamefont {Parigi}\ \emph {et~al.}(2007)\citenamefont {Parigi},
  \citenamefont {Zavatta}, \citenamefont {Kim},\ and\ \citenamefont
  {Bellini}}]{Parigi07}%
  \BibitemOpen
  \bibfield  {author} {\bibinfo {author} {\bibfnamefont {V.}~\bibnamefont
  {Parigi}}, \bibinfo {author} {\bibfnamefont {A.}~\bibnamefont {Zavatta}},
  \bibinfo {author} {\bibfnamefont {M.}~\bibnamefont {Kim}},\ and\ \bibinfo
  {author} {\bibfnamefont {M.}~\bibnamefont {Bellini}},\ }\bibfield  {title}
  {\bibinfo {title} {Probing quantum commutation rules by addition and
  subtraction of single photons to/from a light field},\ }\href
  {https://doi.org/10.1126/science.1146204} {\bibfield  {journal} {\bibinfo
  {journal} {Science}\ }\textbf {\bibinfo {volume} {317}},\ \bibinfo {pages}
  {1890} (\bibinfo {year} {2007})}\BibitemShut {NoStop}%
\bibitem [{\citenamefont {Zavatta}\ \emph {et~al.}(2004)\citenamefont
  {Zavatta}, \citenamefont {Viciani},\ and\ \citenamefont
  {Bellini}}]{Zavatta04}%
  \BibitemOpen
  \bibfield  {author} {\bibinfo {author} {\bibfnamefont {A.}~\bibnamefont
  {Zavatta}}, \bibinfo {author} {\bibfnamefont {S.}~\bibnamefont {Viciani}},\
  and\ \bibinfo {author} {\bibfnamefont {M.}~\bibnamefont {Bellini}},\
  }\bibfield  {title} {\bibinfo {title} {Quantum-to-classical transition with
  single-photon-added coherent states of light},\ }\href
  {https://doi.org/0.1126/science.1103190} {\bibfield  {journal} {\bibinfo
  {journal} {Science}\ }\textbf {\bibinfo {volume} {306}},\ \bibinfo {pages}
  {660} (\bibinfo {year} {2004})}\BibitemShut {NoStop}%
\bibitem [{\citenamefont {Lvovsky}\ \emph {et~al.}()\citenamefont {Lvovsky},
  \citenamefont {Grangier}, \citenamefont {Ourjoumtsev}, \citenamefont
  {Parigi}, \citenamefont {Sasaki},\ and\ \citenamefont
  {Tualle-Brouri}}]{Lvovsky20}%
  \BibitemOpen
  \bibfield  {author} {\bibinfo {author} {\bibfnamefont {A.~I.}\ \bibnamefont
  {Lvovsky}}, \bibinfo {author} {\bibfnamefont {P.}~\bibnamefont {Grangier}},
  \bibinfo {author} {\bibfnamefont {A.}~\bibnamefont {Ourjoumtsev}}, \bibinfo
  {author} {\bibfnamefont {V.}~\bibnamefont {Parigi}}, \bibinfo {author}
  {\bibfnamefont {M.}~\bibnamefont {Sasaki}},\ and\ \bibinfo {author}
  {\bibfnamefont {R.}~\bibnamefont {Tualle-Brouri}},\ }\bibfield  {title}
  {\bibinfo {title} {{Production and applications of non-Gaussian quantum
  states of light}},\ }\bibinfo {note} {arXiv:2006.16985}\BibitemShut {NoStop}%
\bibitem [{\citenamefont {Arzani}\ \emph {et~al.}(2017)\citenamefont {Arzani},
  \citenamefont {Treps},\ and\ \citenamefont {Ferrini}}]{Arzani17}%
  \BibitemOpen
  \bibfield  {author} {\bibinfo {author} {\bibfnamefont {F.}~\bibnamefont
  {Arzani}}, \bibinfo {author} {\bibfnamefont {N.}~\bibnamefont {Treps}},\ and\
  \bibinfo {author} {\bibfnamefont {G.}~\bibnamefont {Ferrini}},\ }\bibfield
  {title} {\bibinfo {title} {Polynomial approximation of non-gaussian unitaries
  by counting one photon at a time},\ }\href
  {https://doi.org/10.1103/PhysRevA.95.052352} {\bibfield  {journal} {\bibinfo
  {journal} {Phys. Rev. A}\ }\textbf {\bibinfo {volume} {95}},\ \bibinfo
  {pages} {052352} (\bibinfo {year} {2017})}\BibitemShut {NoStop}%
\bibitem [{\citenamefont {Yukawa}\ \emph {et~al.}(2013)\citenamefont {Yukawa},
  \citenamefont {Miyata}, \citenamefont {Yonezawa}, \citenamefont {Marek},
  \citenamefont {Filip},\ and\ \citenamefont {Furusawa}}]{Yukawa13}%
  \BibitemOpen
  \bibfield  {author} {\bibinfo {author} {\bibfnamefont {M.}~\bibnamefont
  {Yukawa}}, \bibinfo {author} {\bibfnamefont {K.}~\bibnamefont {Miyata}},
  \bibinfo {author} {\bibfnamefont {H.}~\bibnamefont {Yonezawa}}, \bibinfo
  {author} {\bibfnamefont {P.}~\bibnamefont {Marek}}, \bibinfo {author}
  {\bibfnamefont {R.}~\bibnamefont {Filip}},\ and\ \bibinfo {author}
  {\bibfnamefont {A.}~\bibnamefont {Furusawa}},\ }\bibfield  {title} {\bibinfo
  {title} {Emulating quantum cubic nonlinearity},\ }\href
  {https://doi.org/10.1103/PhysRevA.88.053816} {\bibfield  {journal} {\bibinfo
  {journal} {Phys. Rev. A}\ }\textbf {\bibinfo {volume} {88}},\ \bibinfo
  {pages} {053816} (\bibinfo {year} {2013})}\BibitemShut {NoStop}%
\bibitem [{\citenamefont {Hong}\ \emph {et~al.}(1987)\citenamefont {Hong},
  \citenamefont {Ou},\ and\ \citenamefont {Mandel}}]{PhysRevLett.59.2044}%
  \BibitemOpen
  \bibfield  {author} {\bibinfo {author} {\bibfnamefont {C.~K.}\ \bibnamefont
  {Hong}}, \bibinfo {author} {\bibfnamefont {Z.~Y.}\ \bibnamefont {Ou}},\ and\
  \bibinfo {author} {\bibfnamefont {L.}~\bibnamefont {Mandel}},\ }\bibfield
  {title} {\bibinfo {title} {Measurement of subpicosecond time intervals
  between two photons by interference},\ }\href
  {https://doi.org/10.1103/PhysRevLett.59.2044} {\bibfield  {journal} {\bibinfo
   {journal} {Phys. Rev. Lett.}\ }\textbf {\bibinfo {volume} {59}},\ \bibinfo
  {pages} {2044} (\bibinfo {year} {1987})}\BibitemShut {NoStop}%
\bibitem [{\citenamefont {Beenakker}\ \emph {et~al.}(2009)\citenamefont
  {Beenakker}, \citenamefont {Venderbos},\ and\ \citenamefont {van
  Exter}}]{PhysRevLett.102.193601}%
  \BibitemOpen
  \bibfield  {author} {\bibinfo {author} {\bibfnamefont {C.~W.~J.}\
  \bibnamefont {Beenakker}}, \bibinfo {author} {\bibfnamefont {J.~W.~F.}\
  \bibnamefont {Venderbos}},\ and\ \bibinfo {author} {\bibfnamefont {M.~P.}\
  \bibnamefont {van Exter}},\ }\bibfield  {title} {\bibinfo {title} {Two-photon
  speckle as a probe of multi-dimensional entanglement},\ }\href
  {https://doi.org/10.1103/PhysRevLett.102.193601} {\bibfield  {journal}
  {\bibinfo  {journal} {Phys. Rev. Lett.}\ }\textbf {\bibinfo {volume} {102}},\
  \bibinfo {pages} {193601} (\bibinfo {year} {2009})}\BibitemShut {NoStop}%
\bibitem [{\citenamefont {Walschaers}(2020)}]{Walschaers_2020}%
  \BibitemOpen
  \bibfield  {author} {\bibinfo {author} {\bibfnamefont {M.}~\bibnamefont
  {Walschaers}},\ }\bibfield  {title} {\bibinfo {title} {Signatures of
  many-particle interference},\ }\href
  {https://doi.org/10.1088/1361-6455/ab5c30} {\bibfield  {journal} {\bibinfo
  {journal} {Journal of Physics B: Atomic, Molecular and Optical Physics}\
  }\textbf {\bibinfo {volume} {53}},\ \bibinfo {pages} {043001} (\bibinfo
  {year} {2020})}\BibitemShut {NoStop}%
\bibitem [{\citenamefont {Somaschi}\ \emph {et~al.}(2016)\citenamefont
  {Somaschi}, \citenamefont {Giesz}, \citenamefont {De~Santis}, \citenamefont
  {Loredo}, \citenamefont {Almeida}, \citenamefont {Hornecker}, \citenamefont
  {Portalupi}, \citenamefont {Grange}, \citenamefont {Ant{\'o}n}, \citenamefont
  {Demory}, \citenamefont {G{\'o}mez}, \citenamefont {Sagnes}, \citenamefont
  {Lanzillotti-Kimura}, \citenamefont {Lema{\'\i}tre}, \citenamefont
  {Auffeves}, \citenamefont {White}, \citenamefont {Lanco},\ and\ \citenamefont
  {Senellart}}]{Senellart}%
  \BibitemOpen
  \bibfield  {author} {\bibinfo {author} {\bibfnamefont {N.}~\bibnamefont
  {Somaschi}}, \bibinfo {author} {\bibfnamefont {V.}~\bibnamefont {Giesz}},
  \bibinfo {author} {\bibfnamefont {L.}~\bibnamefont {De~Santis}}, \bibinfo
  {author} {\bibfnamefont {J.~C.}\ \bibnamefont {Loredo}}, \bibinfo {author}
  {\bibfnamefont {M.~P.}\ \bibnamefont {Almeida}}, \bibinfo {author}
  {\bibfnamefont {G.}~\bibnamefont {Hornecker}}, \bibinfo {author}
  {\bibfnamefont {S.~L.}\ \bibnamefont {Portalupi}}, \bibinfo {author}
  {\bibfnamefont {T.}~\bibnamefont {Grange}}, \bibinfo {author} {\bibfnamefont
  {C.}~\bibnamefont {Ant{\'o}n}}, \bibinfo {author} {\bibfnamefont
  {J.}~\bibnamefont {Demory}}, \bibinfo {author} {\bibfnamefont
  {C.}~\bibnamefont {G{\'o}mez}}, \bibinfo {author} {\bibfnamefont
  {I.}~\bibnamefont {Sagnes}}, \bibinfo {author} {\bibfnamefont {N.~D.}\
  \bibnamefont {Lanzillotti-Kimura}}, \bibinfo {author} {\bibfnamefont
  {A.}~\bibnamefont {Lema{\'\i}tre}}, \bibinfo {author} {\bibfnamefont
  {A.}~\bibnamefont {Auffeves}}, \bibinfo {author} {\bibfnamefont {A.~G.}\
  \bibnamefont {White}}, \bibinfo {author} {\bibfnamefont {L.}~\bibnamefont
  {Lanco}},\ and\ \bibinfo {author} {\bibfnamefont {P.}~\bibnamefont
  {Senellart}},\ }\bibfield  {title} {\bibinfo {title} {Near-optimal
  single-photon sources in the solid state},\ }\href@noop {} {\bibfield
  {journal} {\bibinfo  {journal} {Nature Photonics}\ }\textbf {\bibinfo
  {volume} {10}},\ \bibinfo {pages} {340} (\bibinfo {year} {2016})}\BibitemShut
  {NoStop}%
\bibitem [{\citenamefont {Lachman}\ \emph {et~al.}(2019)\citenamefont
  {Lachman}, \citenamefont {Straka}, \citenamefont {Hlou\ifmmode~\check{s}\else
  \v{s}\fi{}ek}, \citenamefont {Je\ifmmode~\check{z}\else \v{z}\fi{}ek},\ and\
  \citenamefont {Filip}}]{Lachman19}%
  \BibitemOpen
  \bibfield  {author} {\bibinfo {author} {\bibfnamefont {L.~c.~v.}\
  \bibnamefont {Lachman}}, \bibinfo {author} {\bibfnamefont {I.}~\bibnamefont
  {Straka}}, \bibinfo {author} {\bibfnamefont {J.}~\bibnamefont
  {Hlou\ifmmode~\check{s}\else \v{s}\fi{}ek}}, \bibinfo {author} {\bibfnamefont
  {M.}~\bibnamefont {Je\ifmmode~\check{z}\else \v{z}\fi{}ek}},\ and\ \bibinfo
  {author} {\bibfnamefont {R.}~\bibnamefont {Filip}},\ }\bibfield  {title}
  {\bibinfo {title} {Faithful hierarchy of genuine $n$-photon quantum
  non-gaussian light},\ }\href {https://doi.org/10.1103/PhysRevLett.123.043601}
  {\bibfield  {journal} {\bibinfo  {journal} {Phys. Rev. Lett.}\ }\textbf
  {\bibinfo {volume} {123}},\ \bibinfo {pages} {043601} (\bibinfo {year}
  {2019})}\BibitemShut {NoStop}%
\bibitem [{Note2()}]{Note2}%
  \BibitemOpen
  \bibinfo {note} {The case where $p_{WS} = 0.05$ forms an exception. Here we
  consider $74$ realizations.}\BibitemShut {Stop}%
\bibitem [{Note3()}]{Note3}%
  \BibitemOpen
  \bibinfo {note} {This reduction in the degree and clustering should be
  explained by a reduction in the weight of the connections because, from (\ref
  {Cij-Gaus}), we can demonstrate that the number of connection in the emergent
  network is higher for larger values of $p_{WS}$. This is related to the
  choice of the normalization given by the denominator in the elements of the
  correlation matrix, as explained in Sec.~\ref {sec:emer-corr}.}\BibitemShut
  {Stop}%
\bibitem [{\citenamefont {Navarrete-Benlloch}\ \emph
  {et~al.}(2012)\citenamefont {Navarrete-Benlloch}, \citenamefont
  {Garc\'{\i}a-Patr\'on}, \citenamefont {Shapiro},\ and\ \citenamefont
  {Cerf}}]{PhysRevA.86.012328}%
  \BibitemOpen
  \bibfield  {author} {\bibinfo {author} {\bibfnamefont {C.}~\bibnamefont
  {Navarrete-Benlloch}}, \bibinfo {author} {\bibfnamefont {R.}~\bibnamefont
  {Garc\'{\i}a-Patr\'on}}, \bibinfo {author} {\bibfnamefont {J.~H.}\
  \bibnamefont {Shapiro}},\ and\ \bibinfo {author} {\bibfnamefont {N.~J.}\
  \bibnamefont {Cerf}},\ }\bibfield  {title} {\bibinfo {title} {Enhancing
  quantum entanglement by photon addition and subtraction},\ }\href
  {https://doi.org/10.1103/PhysRevA.86.012328} {\bibfield  {journal} {\bibinfo
  {journal} {Phys. Rev. A}\ }\textbf {\bibinfo {volume} {86}},\ \bibinfo
  {pages} {012328} (\bibinfo {year} {2012})}\BibitemShut {NoStop}%
\bibitem [{\citenamefont {Gagatsos}\ and\ \citenamefont
  {Guha}(2019)}]{PhysRevA.99.053816}%
  \BibitemOpen
  \bibfield  {author} {\bibinfo {author} {\bibfnamefont {C.~N.}\ \bibnamefont
  {Gagatsos}}\ and\ \bibinfo {author} {\bibfnamefont {S.}~\bibnamefont
  {Guha}},\ }\bibfield  {title} {\bibinfo {title} {Efficient representation of
  gaussian states for multimode non-gaussian quantum state engineering via
  subtraction of arbitrary number of photons},\ }\href
  {https://doi.org/10.1103/PhysRevA.99.053816} {\bibfield  {journal} {\bibinfo
  {journal} {Phys. Rev. A}\ }\textbf {\bibinfo {volume} {99}},\ \bibinfo
  {pages} {053816} (\bibinfo {year} {2019})}\BibitemShut {NoStop}%
\bibitem [{\citenamefont {Chabaud}\ \emph
  {et~al.}(2017{\natexlab{b}})\citenamefont {Chabaud}, \citenamefont {Douce},
  \citenamefont {Markham}, \citenamefont {van Loock}, \citenamefont {Kashefi},\
  and\ \citenamefont {Ferrini}}]{PhysRevA.96.062307}%
  \BibitemOpen
  \bibfield  {author} {\bibinfo {author} {\bibfnamefont {U.}~\bibnamefont
  {Chabaud}}, \bibinfo {author} {\bibfnamefont {T.}~\bibnamefont {Douce}},
  \bibinfo {author} {\bibfnamefont {D.}~\bibnamefont {Markham}}, \bibinfo
  {author} {\bibfnamefont {P.}~\bibnamefont {van Loock}}, \bibinfo {author}
  {\bibfnamefont {E.}~\bibnamefont {Kashefi}},\ and\ \bibinfo {author}
  {\bibfnamefont {G.}~\bibnamefont {Ferrini}},\ }\bibfield  {title} {\bibinfo
  {title} {Continuous-variable sampling from photon-added or photon-subtracted
  squeezed states},\ }\href {https://doi.org/10.1103/PhysRevA.96.062307}
  {\bibfield  {journal} {\bibinfo  {journal} {Phys. Rev. A}\ }\textbf {\bibinfo
  {volume} {96}},\ \bibinfo {pages} {062307} (\bibinfo {year}
  {2017}{\natexlab{b}})}\BibitemShut {NoStop}%
\bibitem [{\citenamefont {Deshpande}\ \emph {et~al.}(2021)\citenamefont
  {Deshpande}, \citenamefont {Mehta}, \citenamefont {Vincent}, \citenamefont
  {Quesada}, \citenamefont {Hinsche}, \citenamefont {Ioannou}, \citenamefont
  {Madsen}, \citenamefont {Lavoie}, \citenamefont {Qi}, \citenamefont {Eisert},
  \citenamefont {Hangleiter}, \citenamefont {Fefferman},\ and\ \citenamefont
  {Dhand}}]{deshpande2021quantum}%
  \BibitemOpen
  \bibfield  {author} {\bibinfo {author} {\bibfnamefont {A.}~\bibnamefont
  {Deshpande}}, \bibinfo {author} {\bibfnamefont {A.}~\bibnamefont {Mehta}},
  \bibinfo {author} {\bibfnamefont {T.}~\bibnamefont {Vincent}}, \bibinfo
  {author} {\bibfnamefont {N.}~\bibnamefont {Quesada}}, \bibinfo {author}
  {\bibfnamefont {M.}~\bibnamefont {Hinsche}}, \bibinfo {author} {\bibfnamefont
  {M.}~\bibnamefont {Ioannou}}, \bibinfo {author} {\bibfnamefont
  {L.}~\bibnamefont {Madsen}}, \bibinfo {author} {\bibfnamefont
  {J.}~\bibnamefont {Lavoie}}, \bibinfo {author} {\bibfnamefont
  {H.}~\bibnamefont {Qi}}, \bibinfo {author} {\bibfnamefont {J.}~\bibnamefont
  {Eisert}}, \bibinfo {author} {\bibfnamefont {D.}~\bibnamefont {Hangleiter}},
  \bibinfo {author} {\bibfnamefont {B.}~\bibnamefont {Fefferman}},\ and\
  \bibinfo {author} {\bibfnamefont {I.}~\bibnamefont {Dhand}},\ }\href@noop {}
  {\bibinfo {title} {Quantum computational supremacy via high-dimensional
  gaussian boson sampling}} (\bibinfo {year} {2021}),\ \Eprint
  {https://arxiv.org/abs/2102.12474} {arXiv:2102.12474 [quant-ph]} \BibitemShut
  {NoStop}%
\bibitem [{\citenamefont {Walschaers}(2021)}]{Code}%
  \BibitemOpen
  \bibfield  {author} {\bibinfo {author} {\bibfnamefont {M.}~\bibnamefont
  {Walschaers}},\ }\href
  {https://archive.softwareheritage.org/browse/origin/directory/?origin_url=https://github.com/mwalschaers/ComplexNetworks}
  {\bibinfo {title} {Code: Cv quantum complex networks (version 2)}} (\bibinfo
  {year} {2021})\BibitemShut {NoStop}%
\end{thebibliography}%

\end{document}